\documentclass[11pt]{article}
\usepackage{graphicx}
\usepackage{amssymb,amsmath,epsfig,mathrsfs}
\usepackage[all,cmtip]{xy}
\usepackage{float}
\usepackage[normalem]{ulem}
\usepackage{amsthm}
\usepackage{color}
\usepackage{authblk}
\usepackage[utf8]{inputenc}
\usepackage[T1]{fontenc}
\usepackage{tikz-cd}
\usepackage{etoolbox}
\patchcmd{\footnotemark}{\stepcounter{footnote}}{\refstepcounter{footnote}}{}{}
\usepackage[colorlinks=true,urlcolor=red,citecolor=blue,hyperfootnotes=true]{hyperref}
\usepackage{setspace}
\usepackage{soul}

\newcommand{\be}{\begin{equation}}
\newcommand{\ee}{\end{equation}}
\def\beqa{\begin{eqnarray}}
\def\eeqa{\end{eqnarray}}
\def\bean{\begin{eqnarray*}}
\def\eean{\end{eqnarray*}}

\newcommand{\C}{\mathbb{C}}

\newcommand{\N}{\mathbb{N}}
\newcommand{\Hil}{\mathcal{H}}
\newcommand{\cD}{\mathcal{D}}
\newcommand{\Hind}{\mathcal{H}_{\Sigma}^{\rm ind}}
\newcommand{\Hphys}{\mathcal{H}_{\Sigma}^{\rm phys}}

\newcommand{\del}{\partial}
\newcommand{\Tr}[1]{\:{\rm Tr}\,#1}

\newcommand{\h}{\mathcal{H}}
\newcommand{\bigslant}[2]{{\raisebox{.2em}{$#1$}\left/\raisebox{-.2em}{$#2$}\right.}}

\usepackage{authblk}
\usepackage{mathtools}
\usepackage{braket}
\usepackage{stmaryrd}
\usepackage{caption}
\usepackage{subcaption}
\usepackage{bbm}
\usepackage{cite}
\usepackage{comment}
\usepackage[most]{tcolorbox}
\usepackage{tikz}

\makeatletter
\def\th@plain{%
  \thm@notefont{}
  \itshape 
}
\def\th@definition{%
  \thm@notefont{}
  \normalfont 
}
\makeatother

\newtheorem{theorem}{Theorem}[section]
\newtheorem{proposition}{Proposition}[section]
\newtheorem{corollary}[proposition]{Corollary}
\newtheorem{lemma}[proposition]{Lemma}
\newtheorem{conjecture}[proposition]{Conjecture}

\newtheorem{definition}{Definition}[section]
\theoremstyle{definition}
\newtheorem{remark}[proposition]{Remark}
\newtheorem{example}[proposition]{Example}
\AtEndEnvironment{example}{\null\hfill\ensuremath{\diamondsuit}}

\definecolor{crimson}{rgb}{0.7, 0.08, 0.24}
\definecolor{green(munsell)}{rgb}{0.0, 0.66, 0.47}
\definecolor{Blu}{rgb}{0.0, 0.25, 0.85}
\definecolor{mediumorchid}{rgb}{0.73, 0.33, 0.83}

\newcounter{lst}

\DeclarePairedDelimiter{\norm}{\lVert}{\rVert}

\hypersetup{
	pdftitle={},%
	pdfauthor={},%
	pdfsubject={},%
	pdfkeywords={},%
	colorlinks=true,%
	linkcolor=Blu,%
	linkbordercolor=crimson,
	citecolor=Blu,%
	linktocpage=true,%
    urlcolor=Blu,
	pageanchor=true
}

\numberwithin{equation}{section}

\title{\textbf{The Structure of the Continuum Limit of Spin Foams}}

\author[1,2]{Matteo Bruno\footnote{\texttt{bruno@cpt.univ-mrs.fr }}}
\author[3,4]{Eugenia Colafranceschi\footnote{\texttt{ecolafra@ucm.es}}}
\author[5,4]{Fabio M. Mele\footnote{\texttt{fmele1@lsu.edu}}}
\author[1,6]{Carlo Rovelli\footnote{\texttt{crovelli@uwo.ca}}}

\affil[1]{\textit{\normalsize{Aix-Marseille Univ, Universit\'e de Toulon, CNRS, CPT, Marseille, France}}}
\affil[2]{\textit{\normalsize{Physics Department, Sapienza University of Rome, P.za Aldo Moro 5, Rome, 00185, , Italy}}}
\affil[3]{\textit{\normalsize{Departamento de Física Teórica, Universidad Complutense de Madrid,
Plaza de las Ciencias 1, 28040 Madrid, Spain}}}
\affil[4]{\textit{\normalsize{Department of Physics \& Astronomy, Western University, N6A3K7, London ON, Canada}}}
\affil[5]{\textit{\normalsize{Department of Physics and Astronomy, Louisiana State University, Baton Rouge, LA 70803, USA}}}
\affil[6]{\textit{\normalsize{Perimeter Institute, 31 Caroline Street N, Waterloo ON, N2L2Y5, Canada}}}

\begin{document}
\include{epsf}
\maketitle
\vspace{-0.5cm}
\begin{abstract}
The Spin Foam approach to quantum gravity aims at providing a covariant path-integral formulation of canonical Loop Quantum Gravity.~Since spin foam amplitudes are defined through discretisations of spacetime, understanding the continuum limit of the theory remains a central open problem.~In this work, we investigate the structural aspects of this limit in a model-independent manner.

We introduce an axiomatic framework for spin foam amplitudes inspired by Atiyah’s formulation of Topological Quantum Field Theories (TQFTs).~In this setting, Hilbert spaces and amplitudes are assigned to combinatorial and topological data associated with triangulated manifolds.~By equipping the set of triangulations with suitable orders, this framework provides a precise notion of continuum limit and allows us to analyse its properties independently of any specific model.

We systematically investigate how the specifics of the limit procedure allow to go beyond TQFT in the continuum.~Under natural assumptions on the convergence of spin foam amplitudes, we establish a no-go result:~sufficiently strong notions of convergence necessarily lead to a topological theory.~Motivated by this obstruction, we weaken the notion of convergence and consider the continuum limit of spin foam amplitudes in a distributional sense, in the spirit of Refined Algebraic Quantisation.~Under this assumption, the amplitude associated with the cylinder defines a rigging map, yielding a canonical construction of the physical Hilbert space.~The resulting continuum amplitudes act as well-defined distributions on this space of physical states, characterising this formulation of the gravitational path integral as physical in a precise sense.
\end{abstract}

\newpage
\noindent\makebox[\linewidth]{\rule{\textwidth}{1pt}} 
\begin{spacing}{0.85}
\tableofcontents
\end{spacing}
\noindent\makebox[\linewidth]{\rule{\textwidth}{1pt}} 

\section{Introduction}\label{Sec:intro}

A central challenge in quantum gravity is to formulate a non-perturbative and background-independent description of quantum spacetime dynamics.~The Spin Foam approach does so by defining covariant transition amplitudes for spacetime regions from combinatorial data associated with a discretisation of geometry \cite{Rovelli_Vidotto_2014} (see also \cite{Baez:1999sr,Perez_2013,Bianchi:2017hjl,Engle_Speziale_2023} for reviews).~Understanding the continuum structure underlying these amplitudes is a central open problem in this approach, to which considerable effort has been devoted over the years (see e.g.~\cite{Reisenberger_Rovelli_1997,Reisenberger:2000fy,Dittrich:2008pw,Bahr:2010cq,Bahr:2011uj,Rovelli_Smerlak_2012,Dittrich:2012jq,Dittrich:2013xwa,Oriti:2014yla,Dittrich:2014ala,Rovelli_2022,Asante_Dittrich_Steinhaus_2023,Han:2025emp,Han:2026hyq} for a sample of the literature on the topic).

Here we analyse the structural aspects of the spin foam theory focusing on two closely related questions: \textit{What should the non-perturbative structure of quantum gravity look like when formulated in terms of transition amplitudes?} and, \textit{as spin foam models rely on transition amplitudes approximated by triangulations, what does it mean to take the limit of these approximations?}

These questions are tightly connected: the continuum structure that emerges depends crucially on how this limit is defined. 

Concerning the first point, a background-independent theory of gravity can be  expected to assign to each (closed) spatial boundary $\Sigma$ a Hilbert space $\Hil_\Sigma$ possibly depending on its topology, and to spacetime regions $M$ transition amplitudes $Z_M$ that are diffeomorphism invariant and do not depend on any bulk discretisation.~This  structure closely resembles that of a Topological Quantum Field Theory (TQFT) in the sense of Atiyah~\cite{Atiyah_1988}.~Topological quantum field theories were first introduced by Witten~\cite{Witten_1988a,Witten_1988b,Witten_1989} and subsequently axiomatised by Atiyah, whose framework assigns a Hilbert space to each closed $(d-1)$-dimensional manifold and a linear map to each $d$-dimensional manifold interpolating between such boundaries, in a way compatible with gluing and invariant under diffeomorphisms.~This thus provides in principle a clear mathematical framework for formulating a quantum theory of gravity, in which Hilbert spaces associated with closed Cauchy surfaces depend only on their topology, transition amplitudes depend only on boundary data and are diffeomorphism invariant, and composition is implemented by the gluing of spacetime regions.

However, quantum gravity cannot exactly satisfy Atiyah's axioms, for two related reasons.~First, Atiyah's framework leads to finite-dimensional boundary Hilbert spaces, whereas a physical theory of quantum gravity is expected to involve infinitely many independent degrees of freedom.~Accommodating infinite-dimensional spaces therefore requires relaxing some aspect of Atiyah's structure \cite{Rovelli_1993,Crane_1994,Barrett_1995}; which aspect must be modified will become clear through our analysis.

Second, in Atiyah’s framework $\Hil_\Sigma$ represents the physical Hilbert space of diffeomorphism-invariant states, and the transition amplitude associated with the cylinder over $\Sigma$ is required to act as the identity on it.~In other words, the cylinder does not change boundary states.~In canonical gravity, however, boundary states are initially defined in a kinematical Hilbert space and are subject to constraints.~Physical states are those annihilated by the Hamiltonian constraint.~Formally, one would like to obtain them by projecting kinematical states onto the kernel of the Wheeler–DeWitt operator \cite{DeWitt_67}.~If this operator has continuous spectrum and the value zero lies in its continuous (rather than point) spectrum, as generally expected, normalisable solutions do not exist and the physical Hilbert space cannot be realised as a subspace of the kinematical Hilbert space.~In this situation, there is no ordinary orthogonal projector onto physical states.

The appropriate framework is Refined Algebraic Quantisation (RAQ)~\cite{Giulini_Marolf_1999a,Giulini_Marolf_1999b} (or generalisations thereof; see e.g.~Ch.~30 in \cite{Thiemann_2007}).~In this framework, one starts from a dense domain $\mathcal D$ in the kinematical Hilbert space $\Hil_{\text{kin}}$ and constructs the physical Hilbert space in the algebraic dual $\mathcal D'$ via a rigging map, $\eta : \mathcal D \to \mathcal D'$, which implements a generalised projection onto solutions of the constraints.~The physical inner product is then defined through this map.~The ``projector'' onto physical states is therefore not an endomorphism on the kinematical Hilbert space, but a distributional-valued map defined within a Gelfand triple $\mathcal D \subset \Hil_{\text{kin}} \subset \mathcal D'$.

From this perspective, the cylinder amplitude in quantum gravity cannot act as the identity on a pre-existing physical Hilbert space. Rather, it must implement precisely this generalised projection: the transition amplitude associated with the cylinder should realise the rigging map, mapping kinematical boundary states to distributional physical states. 

The spin foam framework assumes that the exact transition amplitudes can be {\em approximated} by truncated amplitudes $Z_\Delta$ defined on a two-complex $\Delta$; here we take $\Delta$ to be specifically given by a triangulation.~Making precise the assumption that the discrete amplitudes  approximate a physical amplitude  requires specifying how a sequence (or, more generally, a net) of such amplitudes should converge, and in which space this convergence is to be understood.~This is precisely the second question raised above.

In a TQFT, a suitable construction of the amplitudes $Z_\Delta$ already provides the final answer; in particular, the amplitude associated with the cylinder is the identity.~In a theory of quantum gravity in dimension greater than three, however, the truncated amplitudes do not define the generalised projection required to implement the constraints.~For this reason, one must take a continuum limit.

In this work we show that the truncated amplitudes can satisfy a set of axioms structurally analogous to Atiyah’s framework, though not fulfilling all of its requirements.~We then identify the notion of convergence under the continuum limit that is required in order to obtain a continuum theory consistent with the structural properties expected of quantum gravity.~In particular, we show that, given our axioms at the truncated level, a mathematically rigorous continuum limit for $Z_\Delta$ can be defined.~This limit yields a well-defined amplitude $Z(\Sigma \times I)$ for the cylinder $\Sigma \times I$, which realises the rigging map.~Moreover, for a generic spacetime region $M$, the continuum amplitude $Z(M)$ defines physical states in a mathematically precise sense.

The derivation of our results therefore takes Atiyah’s TQFT framework as both a conceptual guide and a reference point for generalising from a topological theory to a non-topological one suitable for gravity in dimension greater than three. This  generalisation led us to explore several strategies for defining the continuum limit, in particular different notions of triangulation refinement, different refinement schemes (such as bulk refinement with fixed boundary, or coordinated refinement of bulk and boundary), and different spaces in which convergence could be formulated. Each of these attempts provides insight on the direction to pursue and on the difficulties that arise when attempting to extend Atiyah’s framework to non-topological quantum gravity. Because of the conceptual and pedagogical value of these intermediate strategies, we have chosen to present the work following the same path that led us to the final formulation, namely as a sequence of progressively more informed steps.

We begin in Section~\ref{Sec:TQFT} by reviewing Atiyah’s mathematical formulation of Topological Quantum Field Theories. In this framework, a TQFT is defined as a map $Z$ assigning algebraic structures (Hilbert spaces and linear maps) to topological ones (manifolds), subject to a set of axioms ensuring the consistent definition and composition of the transition amplitudes of the theory. A particularly important axiom for our discussion is the Identity axiom, which states that the amplitude associated with the cylinder over $\Sigma$ acts as the identity on $\mathcal H_\Sigma$. Together with the remaining axioms, this condition implies that the Hilbert spaces $\mathcal H_\Sigma$ are finite-dimensional. 

The basic ingredients of the spin foam approach to Loop Quantum Gravity are recalled in Section~\ref{Sec:SFQG}. The starting point is the observation that four-dimensional General Relativity can be written as a constrained topological field theory, namely BF theory supplemented by the simplicity constraint. Upon introducing a spatial hypersurface $\Sigma$, the choice of a timelike normal selects a preferred frame and breaks the internal gauge symmetry $SL(2,\mathbb{C})$ down to $SU(2)$. After quantisation, and upon considering a truncation of the degrees of freedom of the gravitational field encoded in a graph on $\Sigma$, one obtains a Hilbert space of quantum geometries of the hypersurface spanned by colourings of the graph ($SU(2)$ representations on links and intertwiners on nodes), namely spin network states.
The covariant dynamics of the theory is formulated in terms of transition amplitudes between such boundary states. These amplitudes are defined by assigning to a two-complex interpolating between boundary graphs, typically dual to a triangulation $\Delta$ of the four-dimensional region of spacetime bounded by $\Sigma$, a truncated amplitude $Z(\Delta)$ given as a product of local contributions associated with the vertices of the complex. The specification of the vertex amplitude specifies the dynamics.

Building on this structure, we present an axiomatic definition of spin foam models at the discrete level.~We assigns a Hilbert space $\mathcal H_\delta$ to a boundary triangulation $\delta$ and a vector $Z(\Delta)\in\mathcal H_\delta$ to a bulk triangulation $\Delta$ with boundary $\partial\Delta=\delta$.~We specify a set of properties analogous to those appearing in Atiyah’s formulation of TQFTs.~A crucial difference is the absence of an Identity axiom: at the truncated level there is no triangulation of the cylinder $\Sigma\times I$ whose amplitude acts as the identity on the boundary Hilbert space.~As we shall argue, the emergence of such a structure is expected to arise only in the continuum limit.

In Section~\ref{Sec:Continuum1} we take a first step towards the continuum by defining a limit procedure that removes the bulk regulator while keeping the boundary triangulation fixed.~We introduce a simple ordering on the set of bulk triangulations compatible with a given boundary triangulation $\delta$, based on the number of simplices.~This ``weak order'' turns the collection of truncated amplitudes into a net in the boundary Hilbert space.~Assuming that this net admits a limit in $\mathcal H_\delta$, we study its properties and show that the resulting continuum theory is necessarily topological in the sense of Atiyah.~This result provides a clear indication that bulk refinement at fixed boundary, defined simply by increasing the number of simplices, is too weak to capture propagating local degrees of freedom.

To obtain a notion of refinement that better reflects the physical intuition of resolving additional degrees of freedom, in Section~\ref{Sec:Continuum2} we instead consider refinements based on subdivisions of triangulations.~Since such refinements generally modify the boundary triangulation, the truncated amplitudes $Z(\Delta)$ no longer take values in a single Hilbert space.~To address this issue, we introduce the inductive limit of the family of boundary Hilbert spaces associated with different triangulations of a fixed boundary hypersurface.~This construction provides a common Hilbert space in which amplitudes corresponding to different boundary discretisations can be compared, and allows us to define a refined net of amplitudes $\tilde Z_\bullet$ associated with each manifold.

We analyse the properties of this net and study the consequences of assuming the existence of a continuum limit within the inductive Hilbert space.~We show that this assumption again leads to a Topological Quantum Field Theory.~This yields a no-go theorem:~if, for every manifold $M$, the net $\tilde Z_\bullet$ converges to an element of the inductive Hilbert space, the resulting assignment of Hilbert spaces and amplitudes defines a TQFT.~Since this cannot describe four-dimensional quantum gravity, this result indicates that the continuum amplitudes cannot, in general, exist as normalisable elements of the inductive Hilbert space.

For this reason, we further generalise the notion of continuum limit by allowing the amplitudes to converge only in a distributional sense.~More precisely, we introduce a Gelfand triple by selecting a dense domain $\mathcal D_\Sigma$ in the inductive Hilbert space and considering limits in its algebraic dual $\mathcal D_\Sigma'$.~Within this framework, the amplitude associated with the cylinder $\Sigma\times I$ defines a rigging map in the sense of Refined Algebraic Quantisation.~This construction allows us to interpret the continuum amplitudes $Z(M)$ as distributional objects acting on the domain $\mathcal D_\Sigma$.

In Section~\ref{Sec:HphysCLQG} we use the rigging map obtained from the distributional continuum limit to construct the physical Hilbert space following the standard prescription of RAQ. The rigging map defines a positive semidefinite Hermitian form on the domain $\mathcal D_\Sigma$, and the physical Hilbert space is obtained by the completion of the quotient of $\mathcal D_\Sigma$ by the kernel of this form. 

Finally, we study the role of gluing in this framework. Since the continuum amplitudes $Z(M)$ take values in the algebraic dual rather than in a Hilbert space, the Atiyah-type gluing property cannot hold in its usual form.~We therefore propose a weaker convolution property involving the rigging map, which plays the role of the propagator between boundary components.~Under this prescription, the amplitudes $Z(M)$ define linear functionals on the physical Hilbert space, providing a precise sense in which they can be interpreted as physical states.

We conclude in Section~\ref{Sec:conclusion} with a summary of the main results and an outline of future research directions.

The presentation is complemented by several appendices, which contain further technical background material, as well as proofs and additional discussion of the results in the main text.

\section{A Primer of Topological Quantum Field Theories}\label{Sec:TQFT}

Topological Quantum Field Theories (TQFTs) form a special class of quantum field theories defined on a differential manifold and whose partition function depends solely on its topology.~These theories provide a clean relation between topological and algebraic (quantum) structures.~They have been extensively investigated over the years, both from a mathematical and a physical standpoints.~They define a well-posed quantisation in both covariant and canonical approaches (see e.g.~\cite{Atiyah_1990} and references therein).~In this section, we  recall the mathematical formulation of TQFTs.~As anticipated in Sec.~\ref{Sec:intro}, this will serve as a reference framework throughout the paper.

A mathematical formulation of TQFTs can be given in the language of category theory~\cite{Baez_Dolan_1995,Bakalov_Kirillov_2001}. A TQFT is a dagger monoidal functor $Z:\mathsf{Cob}_d^{\rm {or}}\to\mathsf{FHilb}$ from the category of $d$-dimensional oriented cobordisms to the category of finite-dimensional Hilbert spaces.\footnote{The original formulation of Atiyah's axioms~\cite{Atiyah_1988} refers to a TQFT over a generic base ring $\Lambda$ defined as a functor from the category of topological spaces to the category of finitely generated $\Lambda$-modules.~As emphasised by Atiyah himself, however, in physical examples, the fixed base ring is typically a field as e.g.~$\mathbb R$, $\mathbb C$, or $\mathbb Z_p$.~In this work, it is sufficient to restrict ourselves to TQFTs over $\mathbb C$.} Despite this mathematical elegance and conciseness, for the purposes of this work it is convenient to dissect this definition and spell out the properties that the map $Z$ must satisfy in order to define a TQFT, following Atiyah's original proposal~\cite{Atiyah_1988}. These are collected in the Definition \ref{def:TQFT} below, where by manifold we indicate a compact oriented smooth manifold.

\begin{definition}[Topological Quantum Field Theory]
\label{def:TQFT}
    A $d$-dimensional (unitary) TQFT is a map $Z$ which
    \begin{itemize}
        \item to any closed $(d-1)$-dimensional manifold $\Sigma$, up to orientation-preserving diffeomorphisms, assigns a Hilbert space $\mathcal H_{\Sigma}\coloneqq Z(\Sigma)$; and
        \item to the equivalence class of diffeomorphic $d$-dimensional manifolds $M$ preserving the boundary identification assigns a vector $Z(M)$ in the Hilbert space $\mathcal H_{\partial M}\coloneqq Z(\del M)$.
    \end{itemize}
    These data satisfy the following properties:
    \begin{enumerate}
        \item \textbf{Involution:} $\mathcal H_{\Sigma^*}\simeq \mathcal H_{\Sigma}^*$, where $\Sigma^*$ denotes the manifold with opposite orientation and $\mathcal H_{\Sigma}^*$ denotes the strong continuous dual space of $\mathcal H_{\Sigma}$.
        \item \textbf{Factorisation:} $\mathcal H_{\Sigma_1\sqcup\Sigma_2}\simeq \mathcal H_{\Sigma_1}\otimes\mathcal H_{\Sigma_2}$, with $\otimes$ denoting the tensor product of Hilbert spaces.
        \item \textbf{Normalisation:} When $\Sigma=\varnothing$, then $\mathcal H_{\varnothing}\simeq \C$.
        \item \textbf{Dagger:} $Z(M^*)=R_{\Sigma}(Z(M))$, where $\Sigma=\del M$ and $R_{\Sigma}:\mathcal H_{\Sigma}\to \mathcal H_{\Sigma}^*$ is the Riesz map.
        \item \textbf{Gluing:} For $M\simeq M_1\cup_{\Sigma} M_2$, with $\del M_1=\Sigma_1^*\sqcup \Sigma$ and $\del M_2=\Sigma^*\sqcup \Sigma_2$, then $Z(M)=\braket{Z(M_1),Z(M_2)}_{\Sigma}$, where $\braket{\cdot,\cdot}_\Sigma$ is the natural pairing between $\mathcal H_{\Sigma}$ and $\mathcal H_{\Sigma}^*$.
        \item \textbf{Identity:} For $M=\Sigma\times I$ with $I$ an interval, so that $\del(\Sigma\times I)=\Sigma^*\sqcup\Sigma$, then $Z(\Sigma\times I)=\mathbbm{1}_{\mathcal H_{\Sigma}}$ in $\mathcal H_{\Sigma}^*\otimes \mathcal H_{\Sigma}$.
        \item \textbf{Normalisation:} When $M=\varnothing$, then $Z(\varnothing)=1\in \C$.
    \end{enumerate}
\end{definition}
A few comments on the consequences and interpretation of the above properties are in order. First, when $\partial M = \varnothing$, that is, when $M$ is a closed manifold, the Normalisation property implies that $Z(M)$ is a complex number and a topological invariant.~Physically, it can be interpreted as the partition function of the theory on the manifold $M$, namely the vacuum-to-vacuum amplitude associated with the process represented by $M$.

Recall that, given two Hilbert spaces $\mathcal{H}_1$ and $\mathcal{H}_2$, the tensor product $\mathcal{H}_1^* \otimes \mathcal{H}_2$ is canonically and isometrically isomorphic to the space of Hilbert–Schmidt operators $\mathcal{B}_{\mathrm{HS}}(\mathcal{H}_1,\mathcal{H}_2)$. This isomorphism allows us to associate to each $d$-dimensional manifold a linear operator between the Hilbert spaces assigned to its boundary components. Consequently, the natural pairing $\langle Z(M_1), Z(M_2) \rangle_{\Sigma}$ corresponds to the composition $Z(M_2) \circ Z(M_1)$ of the associated operators. The Gluing property therefore asserts that the composition of cobordisms corresponds to the composition of linear maps. As a result, $Z(M)$ can be computed by decomposing $M$ into simpler pieces, giving rise to the so-called ``Lego game'': manifolds and their invariants are constructed from elementary building blocks. This is one of the central features of TQFTs. It is worth noting that there is generally no unique way to decompose $M$, and hence $Z(M)$ can be computed in multiple equivalent ways. As a particular case of Axiom~5 with $\Sigma_1 = \varnothing = \Sigma_2$, the invariant $Z(M)$ associated with a closed $d$-manifold $M = M_1 \cup_{\Sigma} M_2$, obtained by gluing $M_1$ and $M_2$ along $\Sigma$, is given by the transition amplitude
\[
Z(M) = \langle Z(M_1) \mid Z(M_2) \rangle_{\mathcal{H}_{\Sigma}},
\]
between the states $Z(M_1)$ and $Z(M_2)$, as schematically illustrated in Fig.~\ref{fig:tqftgluing}.
\begin{figure}[t!]
    \centering
    \includegraphics[width=0.875\textwidth]{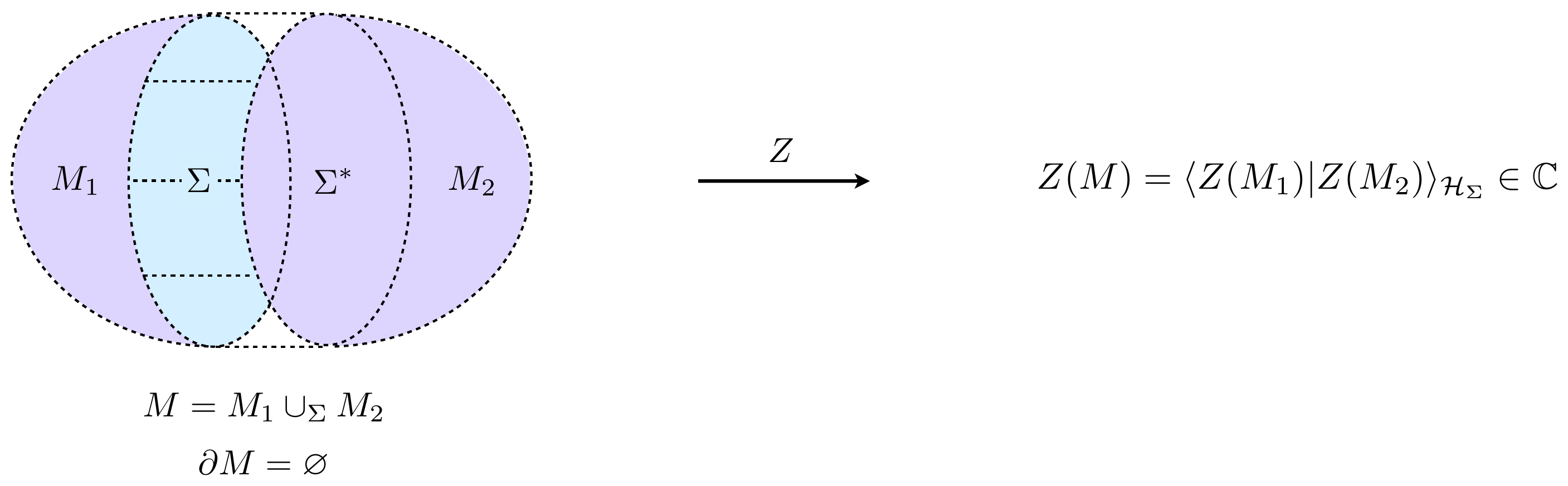}
    \caption{In a TQFT, the gluing of two $d$-manifolds $M_1$ and $M_2$ along their boundary $\Sigma$ is associated to a numerical invariant $Z(M)$, which only depends on the topology of the closed manifold $M=M_1\cup_{\Sigma}M_2$ and is given by the transition amplitude between the quantum states $Z(M_1)$ and $Z(M_2)$ in $\mathcal H_{\Sigma}$ respectively associated to $M_1$ and $M_2$.}
    \label{fig:tqftgluing}
\end{figure}

The Dagger property is related to the adjoint operation.~Indeed, if $\partial M = \Sigma_1^* \sqcup \Sigma_2$, then $Z(M^*) \in \mathcal{H}_{\Sigma_2}^* \otimes \mathcal{H}_{\Sigma_1}$, viewed as an element of $\mathcal{B}_{\mathrm{HS}}(\mathcal{H}_{\Sigma_2},\mathcal{H}_{\Sigma_1})$, is the adjoint of $Z(M) \in \mathcal{B}_{\mathrm{HS}}(\mathcal{H}_{\Sigma_1},\mathcal{H}_{\Sigma_2})$.~In particular, for closed $M$, the invariant $Z(M)$ becomes its complex conjugate under orientation reversal, so that these complex-valued invariants are sensitive to orientation.

The Identity axiom forces the Hilbert spaces $\mathcal{H}_{\Sigma}$ to be finite-dimensional.~This is crucial for the assignment of cobordisms to linear operators and, in particular, for $Z(\Sigma \times I)$ to be well defined.~Indeed, the identity operator belongs to $\mathcal{H}_{\Sigma}^* \otimes \mathcal{H}_{\Sigma} \simeq \mathcal{B}_{\mathrm{HS}}(\mathcal{H}_{\Sigma},\mathcal{H}_{\Sigma})$ if and only if $\mathcal{H}_{\Sigma}$ is finite-dimensional.~A further consistency check is provided by the torus: by identifying the two ends of the cylinder, one obtains
\[
Z(\Sigma \times S^1) = \operatorname{Tr}(\mathbbm{1}_{\mathcal{H}_{\Sigma}}) = \dim \mathcal{H}_{\Sigma},
\]
where $\operatorname{Tr}$ denotes the Hilbert space trace.~The finite dimensionality of $\mathcal{H}_{\Sigma}$ reflects the absence of local, propagating degrees of freedom in a TQFT, where the degrees of freedom are global variables.~Because of that, as shown in~\cite{Atiyah_1988,Bakalov_Kirillov_2001}, the Hilbert space $\mathcal{H}_\Sigma$  is the physical Hilbert space of the theory whose states are invariant under diffeomorphisms in the identity component (i.e. isotopic to the identity).

A TQFT can therefore be seen as a tool to derive topological invariants and their properties.~This feature, among others, has stimulated considerable mathematical interest in the subject.~By assigning algebraic structures (Hilbert spaces) to boundary manifolds and complex numbers to closed manifolds, TQFTs provide a powerful framework for translating topological problems into algebraic ones.~This approach has proved particularly successful in $d=2$ and $d=3$ dimensions~\cite{Abrams:1996ty,Kock_2004,Reshetikhin:1991tc,Witten_1989,Turaev_book} (see also~\cite{Barenz:2016nzn,Meusburger_2025} for developments in four dimensions).

The above set of axioms provides the conceptual starting point and primary source of inspiration for the present work.~Throughout the remainder of the paper, we will frequently refer to these axioms, adopting the notation and terminology introduced in this section.~They will serve both as a guiding framework for the constructions and results that follow and as a reference point for comparison with more general or less restrictive settings.

\section{Axioms for Spin Foams}\label{Sec:SFQG}

Taking inspiration from Atiyah’s axioms for TQFTs, we shall now propose a set of axioms for spin foam models of quantum gravity.~These will refine the viewpoint of spin foam models as combinatorial relatives of TQFTs whose embryonic form can be traced back to~\cite{Rovelli:2010vv,Rovelli_Smerlak_2012,Rovelli:2011eq} (for earlier ideas and related point of views, see also~\cite{Baez_1998,Baez:1999sr,Crane_1994,Oeckl_2003a}).~To this aim, let us start by briefly recalling the basics of the spin foam formulation of loop quantum gravity (for a more extensive presentation, we refer to~\cite{Rovelli_Vidotto_2014} for a textbook reference, and to~\cite{Baez:1999sr,Perez_2013,Bianchi:2017hjl,Engle_Speziale_2023} for reviews).~While Sec.~\ref{Sec:SFreview} may serve as an introduction to spin foams for the uninitiated, the acquainted reader may skip it and continue with Sec.~\ref{Sec:truncatedaxioms}.

\subsection{The Spin Foam Formulation of Loop Quantum Gravity}\label{Sec:SFreview}

The relation between gravity and topological field theories appears already in the Lorentzian 4-dimensional classical theory.~In its first-order formulation, General Relativity can in fact be written as a constrained topological field theory.~Specifically, a BF theory, with action ($16\pi G=1$)
    \begin{equation}
        S[e,\omega]= \int_M \Tr{B \wedge F[\omega]},
    \end{equation}
    together with the constraint
    \begin{equation}
    \label{eq:simpl}
        B = \star(e \wedge e) + \frac{1}{\gamma} (e \wedge e).
    \end{equation}
    Here, $e$ is the tetrad $1$-form, $B$ is a $\mathfrak{sl}(2,\C)$-valued $2$-form, $F$ is the curvature of the Lorentz connection $\omega$, $\star$ is the Hodge dual on internal indices, and $\gamma$ is the Barbero-Immirzi parameter. The constraint~\eqref{eq:simpl} is often called the \emph{simplicity constraint}, and $B$ is called $\gamma$-simple.~For a $4$-dimensional compact region $M$ bounded by a spacelike hypersurface $\Sigma$, the time-like normal to the hypersurface fixes a preferred frame, which is invariant under the action of the little group $SU(2)$ inside $SL(2,\C)$.~Correspondingly, the $B$ field splits into the generators of boosts $\vec{K}$ and rotations $\vec{L}$, and the simplicity constraint~\eqref{eq:simpl} then becomes
    \begin{equation}\label{eq:simpl2}
        \vec{K}=\gamma \vec{L}\;.
    \end{equation}
    As the simplicity constraint encodes, through $e$, the information about the local gravitational degrees of freedom not present in BF theory, the relation~\eqref{eq:simpl2} encodes the dynamical nature of General Relativity.~Moreover, the thermodynamics features of General Relativity, which are intimately related to the Einstein's equations~\cite{Jacobson_1995}, can be derived using solely the simplicity constraint, and vice versa~\cite{Smolin:2012ys,Smolin_2017} (see also~\cite{Bianchi:2012ui}).

    At the quantum level, the states of the theory encode the quantum geometry of the boundary hypersurface.~Owing to the construction of the spacelike hypersurface discussed above, they are objects in representations of $SU(2)$.~More specifically, one considers a truncation to a finite number of degrees of freedom of the gravitational field encoded into a graph $\Gamma$ on the boundary, to which it is associated a Hilbert space
    \begin{equation}\label{eq:bdyHkin}
    \mathcal{H}_\Gamma = L^2\!\left(SU(2)^{L}/SU(2)^N\right),
    \end{equation}
    where $L$ and $N$ denote the numbers of links and nodes of $\Gamma$, respectively.~A basis of $\mathcal H_{\Gamma}$ is provided by spin-network states labelled by the colourings of the graph $\Gamma$, namely a collection of $SU(2)$ spins $j$, one for each link, and intertwiners, i.e.\ projection onto the $SU(2)$-invariant subspace, one for each node.~These are the same states as in canonical loop quantum gravity where the spins define quanta of area, while the intertwiners define quanta of space volume~\cite{Rovelli_2004, Thiemann_2007, Rovelli_Vidotto_2014}. 

    The covariant dynamics of the theory describes then transition amplitudes between boundary quantum geometries.~This is encoded in a CW $2$-complex representing the four-dimensional region of quantum spacetime whose boundary is precisely $\Gamma$ (this $2$-complex is often the $2$-skeleton of the dual complex to a triangulation $\Delta$ of the bulk region, though this need not to be necessarily the case; see e.g.~\cite{Kamiński_Kisielowski_Lewandowski_2010}).~The $2$-complex, whose faces are labelled by the $SU(2)$ irreducible representations and whose edges are labelled by the intertwiners, is called a \emph{foam} and can be interpreted as the history of the spin network, which can generate new links and nodes~\cite{Reisenberger:1994aw,Reisenberger_Rovelli_1997,Baez_1998}~(see Fig.~\ref{fig:spinfoam} for a pictorial representation).~To each foam, one associates an amplitude given by
    \begin{equation}\label{eq:SFampl}
        W_\Delta(\{j_{f_{\rm bdy}}\})=\sum_{\{j_f\},\{i_e\}}\prod_{f\in F(\Delta)} (2j_f+1)\prod_{v\in V(\Delta)}A_v(j_{f_v},i_{e_v}),
    \end{equation}
    where $F(\Delta)$ and $V(\Delta)$ are the sets of the faces (2-cells) and vertices (0-cells) of the $2$-complex dual to $\Delta$, respectively. Here, $j_{f_v}$ and $i_{e_v}$ are the spin numbers and intertwiners associated with a face $f$ and a edge $e$ attached to the vertex $v$. Here, $\{j_{f_{\rm bdy}}\}$ is the collection of spins for faces that intersect the boundary.
    Note that this is essentially equivalent to a linear functional $Z_\Delta$ on the Hilbert space $\mathcal{H}_\Gamma$, where $\Delta$ denotes the triangulation mentioned above. In particular, $W_\Delta(\{j_{f_{\rm bdy}}\})$ is the integration kernel of the operator $Z_\Delta$ in the basis provided by the Peter-Weyl theorem.
    \begin{figure}[t!]
    \centering
    \includegraphics[scale=0.275]{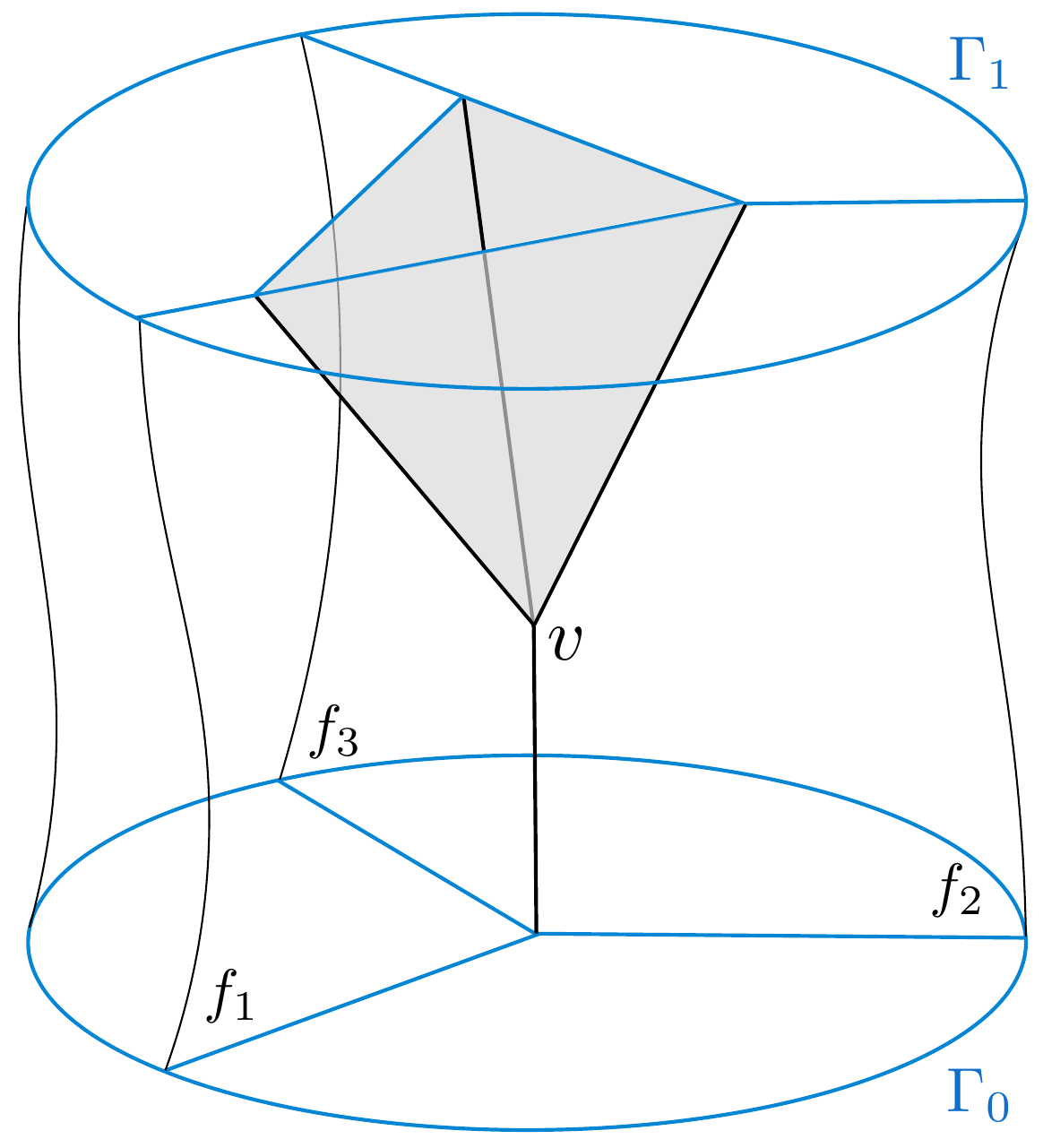}
    \caption{Graphical representation of a spin foam.~The boundary spin network states associated to the graphs $\Gamma_0$ and $\Gamma_1$ (in blue) encode the quantum geometry of the initial and final hyspersurfaces.~The 2-complex interpolating between them represents a quantum spacetime history (in black).~Vertices $v$ can generate new links and nodes as highlighted in grey.}
    \label{fig:spinfoam}
    \end{figure}

The structure of the amplitude~\eqref{eq:SFampl} reflects the locality principle, namely that the amplitudes of a story is the product of individual local amplitudes associated to separate regions of spacetime. The actual dynamics is encoded into the vertex amplitude $A_v$, whose explicit form specifies the different spin-foam models.~A leading proposal for the vertex amplitude of 4d Lorentzian quantum gravity is the so-called EPRL-FK model~\cite{Engle_Livine_Pereira_Rovelli_2008,Freidel_Krasnov_2008} which is derived by enforcing the simplicity constraints at the quantum level.~This is given as a linear functional on a spin network state $\psi$ around the vertex $v$
    \begin{equation}\label{eq:EPRL}
    A_v(\psi)=(P_{SL(2,\C)}Y_\gamma\,\psi)(\mathbbm{1})\;.
    \end{equation}
    Here $P_{SL(2,\C)}$ is the projector onto $SL(2,\C)$-invariant states, $Y_{\gamma}$ is a map embedding $SU(2)$ representations into $SL(2,\mathbb C)$ ones, and $(\mathbbm{1})$ denotes the evaluation of the spin-network state on the identity.~More specifically, the map $Y_{\gamma}$ is given by
    \begin{equation}
    \begin{matrix}
    Y_{\gamma}:&V^j&\to &V^{(p=\gamma j,k=j)}\,;\\
    &| j;m \rangle& \mapsto & |\gamma j,j;j,m \rangle\,,
    \end{matrix}
    \end{equation}
    where $p$ is a positive real number and $k$ is a non-negative half-integer characterising the principal series of unitary representations of $SL(2,\C)$, so that the carrier space $V^{(p,k)}$ of the $SL(2,\mathbb C)$-representations decomposes into irreducible representations of $SU(2)$ as
    \begin{equation}
    V^{(p,k)}=\bigoplus_{j\in\{ k+n\}_{n\in\mathbb{N}}}V^j\;.
    \end{equation}
    States in the image of $Y_{\gamma}$ satisfy the simplicity constraint in the weak sense, namely that $\langle \vec{K}\rangle=\gamma \langle \vec{L}\rangle$ holds in the large-$j$ limit.~Without the map $Y_{\gamma}$ in Eq.~\eqref{eq:EPRL}, Ooguri's quantization of BF theory is recovered~\cite{Ooguri_1992b}.~In other words, the map $Y_{\gamma}$ ensures that the relation~\eqref{eq:simpl2}, which turns BF theory into General Relativity, is implemented at the semiclassical level.~In the large-$j$ semiclassical regime, the vertex amplitude has been shown to approach the exponential of the Regge action~\cite{Regge_1961,Ponzano_Regge_1969}, the discretised version of the Einstein-Hilbert action, thus supporting the idea that the dynamics of discrete General Relativity is correctly captured by the model~\cite{Conrady:2008ea,Bianchi:2010mw,Magliaro:2011dz,Magliaro:2011qm,Han:2011re,Han:2021kll}.

    \smallskip
    
    The discretisation introduced by the bulk 2-complex and the boundary graph represents a truncation of the degrees of freedom of the theory which allows to specify the amplitude in a well-defined manner and, eventually, to compute it.~The full theory --- its Hilbert space $\mathcal H_{\partial M}$ and amplitude $Z \colon \mathcal{H}_{\partial M} \to \mathbb{C}$ --- should therefore be obtained from the truncated theory via a suitable continuum limit in which diffeomorphism invariance is restored.~The sense in which these truncated amplitudes can define the full amplitudes of the theory is the problem on which we focus in this paper.~Different strategies have been proposed in the literature, based on either summing over all possible complexes (see e.g.~\cite{Reisenberger_Rovelli_1997,Reisenberger:2000fy,Oriti:2014yla,Han:2025emp}) or refining the discretisation in various ways (see e.g.~\cite{Rovelli:2011eq,Rovelli_Smerlak_2012, Dittrich:2012jq,Dittrich:2013xwa,Dittrich:2014ala,Asante_Dittrich_Steinhaus_2023,Rovelli_2022}).\footnote{In~\cite{Rovelli_Smerlak_2012}, the possibility that, under certain conditions, summing and refinement may coincide in quantum gravity has been put forward.}
    
\subsection{Axioms in the Discrete}\label{Sec:truncatedaxioms}

The discretisation introduced by spin foam models in order to compute the amplitudes is encoded in the fact that the theory can be defined on \emph{abstract simplicial complexes}. Every abstract simplicial complex admits a geometric realisation as a topological space.~In the following we restrict our attention to those abstract simplicial complexes that define a topological manifold, or more precisely a piecewise-linear (PL) manifold.~Such complexes will be referred to as \emph{triangulations} according to the following definition~\cite{Hudson_1969}.
\begin{definition}[Triangulation]
A $d$-dimensional triangulation is an abstract simplicial complex of dimension $d$ such that the link of each vertex is a $(d-1)$-sphere.
\end{definition}
For details about the definition of triangulation, their properties, and related notation, we refer the reader to Appendix~\ref{App:Triang}.~We note that every smooth manifold admits a triangulation and therefore a PL structure~\cite{Whitehead_1940}.~Moreover, in dimension four and lower, every PL manifold is smoothable, making this framework sufficiently general for the purposes of this work.~With this premise, the relation between bulk and boundary discrete structures and algebraic (quantum) data in spin foam models discussed in Sec.~\ref{Sec:SFreview}, can be formalised as follows.

\begin{definition}[Spin foam model] \label{def:truncSF} A $d$-dimensional spin foam model is a map which
\begin{itemize}
    \item to a triangulation $\delta$, whose geometric realisation $|\delta|$ is a closed, oriented, smooth $(d{-}1)$-manifold $\Sigma$, assigns a (possibly infinite-dimensional) Hilbert space $\mathcal{H}_\delta$;~and
    \item to a triangulation $\Delta$, whose geometric realisation $|\Delta|$ is a compact, oriented, smooth $d$-manifold $M$ with boundary $\partial M = \Sigma$, and whose boundary $\partial \Delta = \delta$ is a triangulation of $\Sigma$, assigns a vector $Z(\Delta) \in \mathcal{H}_\delta$.
\end{itemize}
These data satisfy the following properties:
\begin{enumerate}
    \item \textbf{Involution:} For $\delta^*$ the $(d-1)$-triangulation with opposite orientation than $\delta$, i.e.~such that $|\delta^*|=|\delta|^*$, then $\mathcal{H}_{\delta^*} = \mathcal{H}_\delta^{\,*}$ with $\mathcal{H}_\delta^{\,*}$ the strong continuous dual space of $\mathcal{H}_\delta$;
    
    \item \textbf{Factorisation:} $\mathcal{H}_{\delta_1 \sqcup \delta_2} = \mathcal{H}_{\delta_1} \otimes \mathcal{H}_{\delta_2}$;
    
    \item \textbf{Normalisation:} $\mathcal{H}_{\delta = \varnothing} = \mathbb{C}$ and $Z(\Delta=\varnothing) = 1$;
    
    \item \textbf{Gluing:} If $\Delta \simeq \Delta_1 \cup_{\delta} \Delta_2$, with $\partial \Delta_1 = \delta_1^* \sqcup \delta$ and $\partial \Delta_2 = \delta^* \sqcup \delta_2$, then $Z(\Delta) = \langle Z(\Delta_1), Z(\Delta_2) \rangle_\delta$ with $\langle \cdot , \cdot \rangle_\delta$ the natural pairing between $\mathcal{H}_\delta$ and $\mathcal{H}_\delta^*$;
    
    \item \textbf{Dagger:} $Z(\Delta^*) = R_{\partial \Delta}(Z(\Delta))$, where $R_{\partial \Delta}$ is the Riesz map $R_{\partial \Delta} \colon \mathcal{H}_{\partial \Delta} \to \mathcal{H}_{\partial \Delta}^*$.
\end{enumerate}
\end{definition}

\begin{figure}[t!]
    \centering
    \includegraphics[scale=0.2]{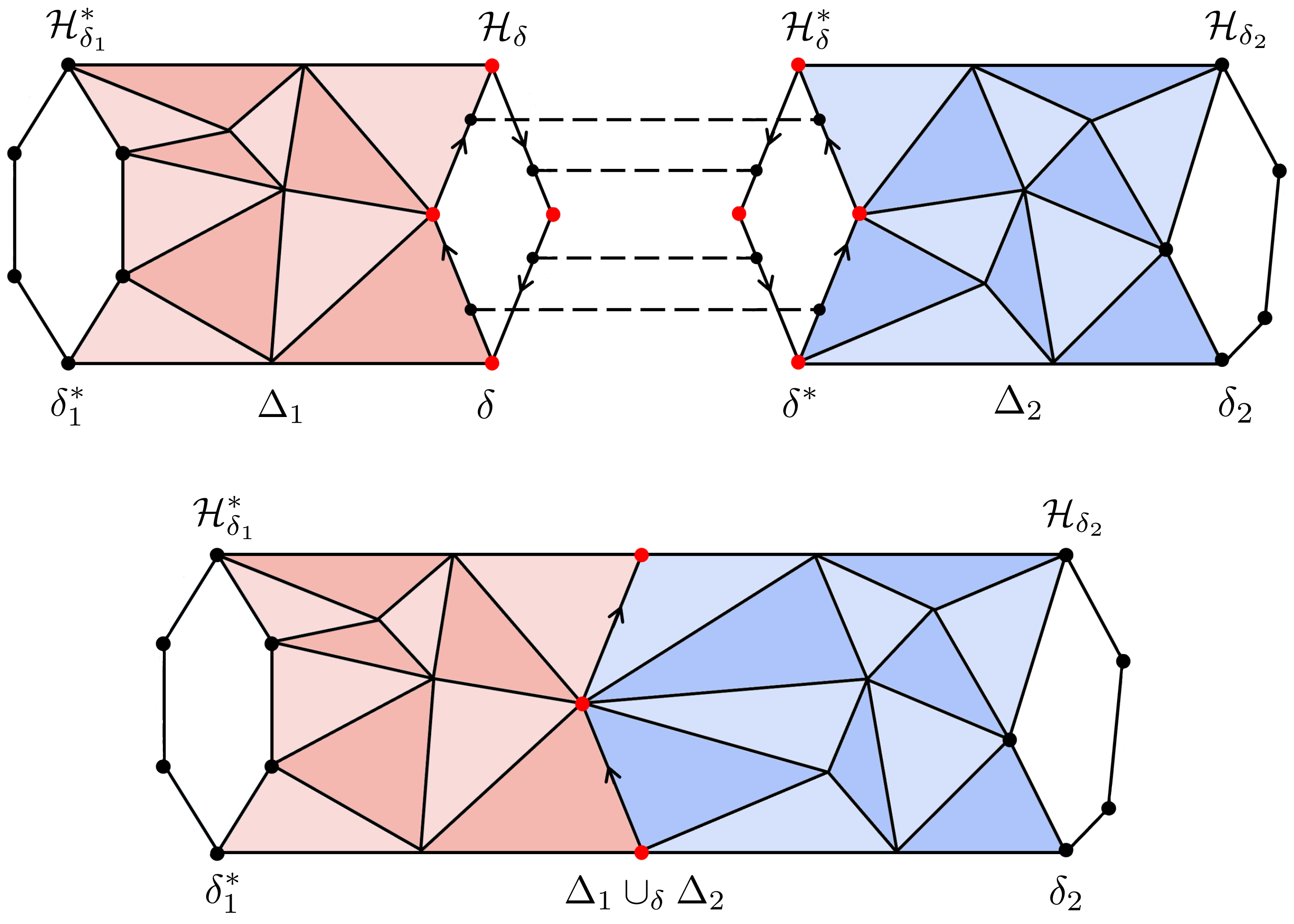}
    \caption{Gluing of triangulated manifolds in a spin foam model.}
    \label{fig:gluing_truncated}
\end{figure}

The list above shares similarities with the TQFT axioms (cfr.~Def.~\ref{def:TQFT} in Sec.~\ref{Sec:TQFT}).~However, there are some important differences.~The first obvious difference is that the theory so-defined depends, in principle, not only on the topological data but also on the combinatorial data of the triangulations.~Compatibly with the gluing of their geometric realisations, the gluing of triangulations is therefore defined through the identification face by face along a common boundary triangulation as depicted in Fig.~\ref{fig:gluing_truncated} (see Appendix~\ref{App:Triang} for a precise definition).~Moreover, we note that the above Definition~\ref{def:truncSF} does \emph{not} include an identity axiom.~Such an axiom would require the existence, for every closed
$(d{-}1)$-triangulation $\delta$, of a $d$-triangulation $\Delta_{\mathbbm{1}}$ homeomorphic to the cylinder $|\Delta_{\mathbbm{1}}|=\Sigma \times I$, such that $\partial\Delta_{\mathbbm{1}}=\delta^*\sqcup\delta$ and that acts as a unit under gluing.~That is, $\Delta \cup_\delta \Delta_{\mathbbm{1}} = \Delta$ for all $\Delta$ with $\partial\Delta={\delta'}^* \sqcup \delta$, for some $\delta'$, and $Z(\Delta_{\mathbbm{1}}) = \mathbbm{1}_{\Hil_\delta}$.~The natural candidate for such a unit is $\delta \times I$, but this object contains no bulk vertices and is therefore \emph{not} a triangulation of $\Sigma \times I$.~Including such an object would require to extend the framework beyond abstract simplicial complexes by admitting more general combinatorial two-complexes, which need not be homeomorphic to a manifold.~This was the case for instance in~\cite{Rovelli_Smerlak_2012} where a comparison with Atiyah's framework was also sketched.~Doing this, however, would weaken the correspondence between topological and combinatorial data that characterises spin foam models and, in particular, would be decremental for our later purposes of studying the recovery of TQFTs or the departure from them in the continuum.

Finally, the Hilbert space associated to boundary triangulations is not the physical Hilbert space of the theory but rather the kinematical Hilbert space, and it needs not to be finite-dimensional.~In concrete spin foam models, this is often taken to be the Hilbert space of the dual spin network graphs given in Eq.~\eqref{eq:bdyHkin}.~To keep our analysis model-independent, in what follows, we do not assume any specific model or form of the dynamics.~The framework of Definition \ref{def:truncSF} describes only the kinematical regularisation of the path integral and may thus apply equally well to discretised theories of gravity and topological theories.

\section{A First Step in the Continuum}\label{Sec:Continuum1}

Although the discrete structure allows for the explicit computation of transition amplitudes, it represents a truncation of the theory to finitely many degrees of freedom.~For this reason, we are ultimately interested in the continuum limit of the theory, which should capture its full physical content.~In close analogy with lattice QCD, this amounts to removing the regulator by means of a suitable limiting procedure.~In the present setting, however, there is no independent refinement parameter, analogous to the lattice spacing in lattice gauge theory, that can be tuned to a critical value or sent to zero.~Instead, the regulator is the abstract simplicial complex itself.~Notice that the same is true in classical Regge calculus, which defines a truncation of the number of degrees of freedom of General Relativity without introducing a scale~\cite{Regge_1961}.~This possibility, characteristic of theories formulated in a reparametrisation invariant language, has been discussed in~\cite{Rovelli_2022}.~Continuing the analogy with lattice gauge theory, where the continuum limit is obtained by sending the lattice spacing to zero while the number of lattice sites diverges accordingly, here only the latter aspect survives:~the continuum limit is realised through an increasingly refined simplicial complex.~The crucial issue is therefore not merely to increase the number of simplices, but to do so in a way that correctly captures the physical degrees of freedom.~In this section, we show that a refinement procedure based solely on the number of simplices is insufficient for this purpose and does not recover the full physical content of the theory.

In this context it is still possible to define a continuum limit by employing the mathematical notion of a net, which generalises the concept of a sequence (see Appendix~\ref{App:Triang} for details).~This allows us to define the continuum limit of transition amplitudes in a precise manner.~Our aim is to study how the structure of the continuum theory is affected by the specifics of how the limit is taken, particularly, the interplay between the removal of the bulk and and boundary regulators.~In this section, we start from considering the case in which the truncated boundary Hilbert space is held fixed.~In analogy with particle physics, this may be thought of as keeping the Hilbert space associated with a finite number of particles, while considering the amplitude computed at arbitrary order.~We define a specific limit procedure that keeps the boundary triangulation fixed but then show that this procedure leads to a topological theory.

\subsection{Weak Order}\label{Sec:weakorder}

Let $\mathcal{T}_M^\delta$ be the set of triangulations $\Delta$ homeomorphic to a $d$-manifold $M$ and such that the induced boundary triangulation satisfies $\partial\Delta=\delta$. Starting from the amplitudes $Z(\Delta)\in\mathcal H_\delta$ for $\Delta \in \mathcal{T}_M^\delta$, we aim to construct a vector $Z_{M;\delta}\in \mathcal{H}_\delta$ representing the transition amplitude associated with the boundary $\delta$ and independent of the bulk triangulation. To this end, we introduce a preorder on $\mathcal{T}_M^\delta$ based on the number of $d$-simplices. This preorder makes $\mathcal{T}_M^\delta$ into a directed set, thereby turning the collection of amplitudes $Z(\Delta)$ into a net. We then define $Z_{M;\delta}$ as the limit of this net.

Specifically, denoting by $N_d(\Delta)$ the number of $d$-simplices in a triangulation $\Delta$, we define
\begin{equation}\label{eq:weakorder}
    \Delta_1 \le \Delta_2\quad \text{if and only if}\quad N_d(\Delta_1)\le N_d(\Delta_2)\,.
\end{equation}
We refer to this preorder as the \textit{weak order}. This is schematically represented in Fig.~\ref{fig:weak}. We emphasise that the increment in the number of simplices does not need to be distributed evenly over all of the bulk.
\begin{figure}[t]
    \centering
    \includegraphics[width=0.75\textwidth]{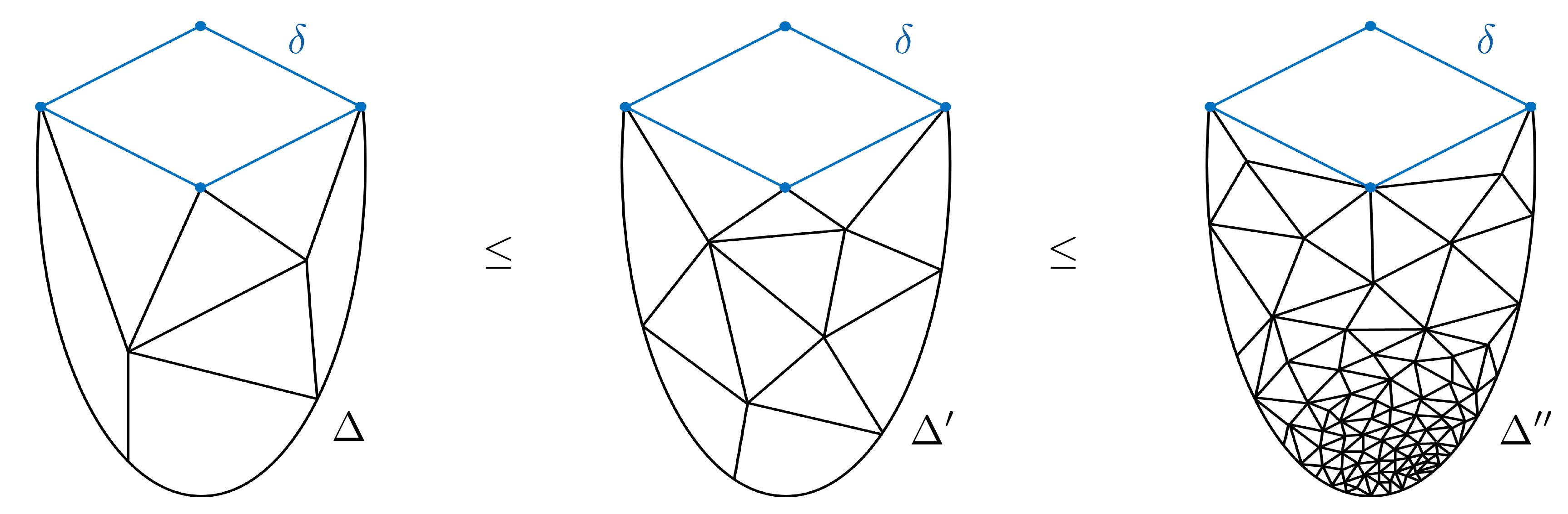}
    \caption{Graphic representation of the weak order. Notice that the increment of simplices could be localised.}
    \label{fig:weak}
\end{figure}

The set $\mathcal{T}_M^\delta$ equipped with the weak order is a directed set.\footnote{This uses the fact that $\mathcal{T}_M^\delta$ is non-empty for any triangulation $\delta$ of $\partial M$; see Lemma~\ref{lem:nonemptyT} and Prop.~\ref{prop:w-ord} in Appendix~\ref{App:Triang}.} We can then define a net for each $M$ and $\delta$:
\begin{equation}
\begin{split}
    Z_\bullet:~\mathcal{T}_M^\delta~&\to~\Hil_\delta \\ \Delta~~&\mapsto~Z_\Delta \coloneqq  Z(\Delta)
\end{split}
\end{equation}
Since $\mathcal{H}_\delta$ is a Hausdorff topological vector space, the limit, if it exists, is unique. We denote it by 
\begin{equation}
\label{eq:Z_Md}
    Z_{M;\delta} \coloneqq  \lim_{\Delta \in \mathcal{T}_M^\delta} Z_\Delta
\end{equation}
We assume the existence of this limit in $\h_\delta$, and proceed to study its properties and the consequences of such an assumption. 

\subsection{Properties of the Limit and TQFT}\label{Sec:weaklimTQFT}

The properties of the limit~\eqref{eq:Z_Md} constrain the resulting continuum theory. As we shall see, the states $Z_{M;\delta}$ exhibit several remarkable features; however, these features are in tension with the expected physical content of the theory. This tension can be traced back to the choice of weak order, which does not fully capture the physical intuition behind refinement of the discretisation. Despite this limitation, we are able to prove a robust result that does not rely on model-specific details, and therefore applies to a broad class of spin foam models.

\smallskip

For the rest of the section, we assume that the limit~\eqref{eq:Z_Md} exists in $\mathcal H_\delta$ for any $d$-manifold $M$ and any boundary triangulation $\delta$. We first need a technical lemma, proved in Appendix~\ref{App:Proofs1}.
\begin{lemma}
    \label{lem:subnet}
    Let $ M\simeq N_1\cup_{\Sigma}N_2$ be a $d$-manifold with $\partial N_1=\Sigma_1^*\sqcup\Sigma$, $\partial N_2=\Sigma^*\sqcup\Sigma_2$ and $\partial  M=\Sigma_1^*\sqcup \Sigma_2$. The subset of $\mathcal{T}_{M}^{\delta_1^* \sqcup \delta_2}$ consisting of triangulations that can be written as the gluing of a triangulation of $N_1$ and a triangulation of $N_2$ along some triangulation $\delta$ of $\Sigma$ (see Fig.~\ref{fig:gluing_truncated}) is the image of the map
    \[
    \begin{matrix}
        \gamma: & \mathcal{T}_{N_1}^{\delta_1^* \sqcup \delta}\times \mathcal{T}_{N_2}^{\delta^* \sqcup \delta_2} & \to & \mathcal{T}_{M}^{\delta_1^* \sqcup \delta_2};\\
        & (\Delta_1,\Delta_2) & \mapsto & \Delta_1\cup_\delta \Delta_2 \,.
    \end{matrix}
    \]
    The map $\gamma$ is order-preserving and cofinal. As a consequence, $Z_{\gamma(\bullet)}$ is a subnet of $Z_{\bullet}$.
\end{lemma}
In fact, this Lemma allows us to consider in the limit just the triangulations that split nicely on the two manifolds $N_1$ and $N_2$. This has clear consequences on the limit of the gluing.
\begin{proposition} \label{prop:w-gluing}
    Let $M=N_1\cup_\Sigma N_2$, with $N_1=\Sigma_1^*\sqcup\Sigma$ and $N_2=\Sigma^*\sqcup\Sigma_2$. Given triangulations $\delta,\delta_1,\delta_2$ of $\Sigma,\Sigma_1,\Sigma_2$, respectively, the following holds:
    \begin{equation}
    \label{eq:gluing}
        Z_{N_1\cup_\Sigma N_2; \delta_1^* \sqcup \delta_2}= \langle Z_{N_1; \delta_1^* \sqcup \delta}, Z_{N_2; \delta^* \sqcup \delta_2} \rangle_\delta\,.
    \end{equation}
\end{proposition}
The proof is in Appendix~\ref{App:Proofs1} and uses, together with the previous Lemma~\ref{lem:subnet}, Lemma~\ref{lem:contProd}.~The basic idea behind the proof consists of the following two main steps.~First, owing to Lemma~\ref{lem:subnet}, the LHS of \eqref{eq:gluing} can be written as a limit of pairings $\langle Z_{\Delta_1},Z_{\Delta_2}\rangle_{\delta}$ over $\mathcal{T}_{N_1}^{\delta_1^* \sqcup \delta}\times \mathcal{T}_{N_2}^{\delta^* \sqcup \delta_2}$.~Second,  Lemma~\ref{lem:contProd} allows us to take the limit of each argument of the pairing separately.~It is worth noting that the same result can in fact be obtained by refining the triangulation on only one side of the decomposition.~More precisely, fixing a triangulation of $N_1$ and refining only the triangulation of $N_2$ already yields the limit state $Z_{N_1\cup_\Sigma N_2;\delta_1^*\sqcup\delta_2}$; the same holds with the roles of $N_1$ and $N_2$ exchanged.~This follows from Lemma~\ref{lem:subnet}, which ensures that the corresponding families of triangulations define subnets of $Z_\bullet$.~The technical argument is given in Remark~\ref{rmk:loc-ref} of Appendix~\ref{App:Proofs1}. 

This behaviour has a clear conceptual interpretation: refining the triangulation only on one side of the decomposition, i.e.\ capturing more and more degrees of freedom just in a localised region, is sufficient to recover the same limiting amplitude for the glued manifold. As we discuss below, this feature indicates that the resulting theory cannot support propagating local degrees of freedom.

An immediate consequence of the gluing property is that the limit of the cylinder defines a projector.
\begin{corollary}\label{cor:w-proj}
    The state $Z_{\Sigma\times I;\delta^* \sqcup \delta}\in \h_\delta^*\otimes\h_\delta$, identified as an operator in $\mathcal{B}_{\rm HS}(\h_\delta,\h_\delta)$, is a projector
    \begin{equation}
    \label{eq:w-proj}
      \langle Z_{\Sigma\times I;\delta^* \sqcup \delta} , Z_{\Sigma\times I;\delta^* \sqcup \delta} \rangle_\delta = Z_{\Sigma\times I;\delta^* \sqcup \delta}.
    \end{equation}
\end{corollary}
\noindent
To prove that $Z_{\Sigma\times I;\delta^* \sqcup \delta}$ is a projection, 
we use the gluing property~\eqref{eq:gluing} with 
$N_1 = N_2 = \Sigma\times I$:
\begin{equation}
Z_{(\Sigma\times I)\cup_\Sigma(\Sigma\times I);\delta^* \sqcup \delta}
=
\big\langle 
Z_{\Sigma\times I;\delta^* \sqcup \delta},\,
Z_{\Sigma\times I;\delta^* \sqcup \delta}
\big\rangle_\delta.    
\end{equation}
Next, observe that $(\Sigma\times I)\cup_\Sigma (\Sigma\times I) \simeq \Sigma\times I$, which implies
\begin{equation}
  Z_{(\Sigma\times I)\cup_\Sigma(\Sigma\times I);\delta^* \sqcup \delta}
= 
Z_{\Sigma\times I;\delta^* \sqcup \delta}.  
\end{equation}
Combining the two equalities yields the desired formula. We now prove the last property, already stated in~\cite{Hellmann_2011}, which is pivotal for the physical interpretation.
\begin{proposition}
    \label{prop:w-range} The ranges of the operators $Z_{\Sigma\times I;\delta^* \sqcup \delta}:\mathcal{H}_{\delta}\to\mathcal{H}_{\delta}$ associated to different triangulations $\delta$ of $\Sigma$ are all isomorphic.~That is,  
   \begin{equation}
   {\rm Ran}(Z_{\Sigma\times I;\delta^* \sqcup \delta})\simeq {\rm Ran}(Z_{\Sigma\times I;\delta'^* \sqcup \delta'})\;,
   \end{equation}
   for every pair $\delta,\delta'$ of triangulations of $\Sigma$.
\end{proposition}
The proof is in Appendix~\ref{App:Proofs1}, and it is a simple exercise of linear algebra.

\smallskip

The significance of Proposition~\ref{prop:w-range} is that, combined with the gluing property and the cylinder projector, it implies that any spin foam model admitting such a limit necessarily defines a TQFT. Indeed, one may define the Hilbert space by
\[
\mathcal{H}_\Sigma \coloneqq \mathrm{Ran}(Z_{\Sigma \times I;\,\delta^* \sqcup \delta}),
\]
which is independent of the choice of triangulation $\delta$.~Furthermore, by Proposition~\ref{prop:w-gluing}, the state $Z_{M;\delta}$ always lies in $\mathcal{H}_\Sigma \coloneqq\mathrm{Ran}(Z_{\Sigma \times I;\,\delta^* \sqcup \delta})$.~Moreover, the states $Z_{M;\delta}$ associated with different boundary triangulations can be identified via the isomorphisms induced by $Z_{\Sigma \times I;\,\delta^* \sqcup \delta'}$, and are thus \emph{de facto} independent of the choice of boundary discretisation. In particular, the gluing property of Proposition~\ref{prop:w-gluing} becomes precisely the gluing axiom of a TQFT.~Therefore, together with the identifications of Proposition~\ref{prop:w-range}, the theory satisfies Atiyah's axioms and thereby defines a TQFT.~The physical intuition behind this is that not only, as discussed after Proposition~\ref{prop:w-gluing}, the amplitudes can be glued across any given intermediate triangulation, but also the resulting amplitude is independent of the chosen triangulation.~The theory is thus insensitive to the truncation of degrees of freedom on the intermediate hypersurface.~This independence from the intermediate triangulation implies that no local information is carried across the gluing hypersurface and, consequently, the theory cannot encode local propagating degrees of freedom.

We have thus established that, if a spin foam model admits a continuum limit with respect to the weak order, then it necessarily defines a TQFT.~Clearly, this cannot be the case for a model aiming to describe four-dimensional quantum gravity.~As illustrated in Sec.~\ref{Sec:examplePRTVO} below, this is instead the case for 3d quantum gravity, which is widely known to be successfully described by a topological spin foam model.~In this respect, we also notice that if, in particular, the net $Z_\bullet \colon \mathcal{T}_M^\delta \to \mathcal{H}_\delta$ is constant, namely independent of the choice of bulk triangulation, then the continuum limit exists trivially.~Since any two triangulations in $\mathcal{T}_M^\delta$ are related by a finite sequence of Pachner moves, invariance under Pachner moves implies constancy of the net, thus yielding the following corollary.
\begin{corollary}\label{cor:PachnerInv}
    If a $d$-dimensional spin foam model is invariant under $d$-dimensional Pachner moves, it is a Topological Quantum Field Theory.
\end{corollary}

\subsection{Example: The Ponzano-Regge Model}\label{Sec:examplePRTVO}

We now illustrate the general construction of Section~\ref{Sec:Continuum1} in a concrete setting. In particular, we show that the Ponzano-Regge model~\cite{Ponzano_Regge_1969}, in its Turaev-Viro regularised form~\cite{Turaev_Viro_1992}, provides an explicit example of a spin foam model for which the continuum limit with respect to the weak order exists and leads to a TQFT.

In this example, the transition amplitudes are independent of the bulk triangulation.~As a consequence, the net of amplitudes over $\mathcal T_M^\delta$ is constant, and the limit~\eqref{eq:Z_Md} exists trivially.~The structure derived abstractly in Sec.~\ref{Sec:Continuum1} is therefore realised explicitly.

We first introduce the model; the reader already familiar with it can skip to Sec.~\ref{sec:invariants}, where we introduce the form of the transition amplitudes relevant for the example and fix the notation used in the subsequent discussion, which follows the presentation of~\cite{Turaev_Viro_1992}.

\subsubsection{The Model}

The Ponzano-Regge model is a definition of a quantum gravity theory on a triangulated (non-oriented) three-dimensional manifold.

The model assigns an irreducible representation of $SU(2)$ to each edge of the triangulation, labelled by a spin $j\in \frac{1}{2}\mathbb{Z}$, interpreted as specifying the metric length of that edge. There is a quantum amplitude for each assignment of spins to the triangulation:
\begin{equation}
\label{eq:PR_weight}
W(j_1, j_2, \dots) =\prod_{\text{interior edges}}(-1)^{2j}(2j+1)\prod_{\text{interior triangles}}(-1)^{j_1+j_2+j_3} \prod_{\text{tetrahedra}}\left\{\begin{matrix}j_1 & j_2 & j_3 \\ j_4 & j_5 & j_6 \end{matrix}\right\}
\end{equation}
where a $6j$-symbol is assigned to each \textit{colored} tetrahedron, namely a tetrahedron together with an assignment of spins to its edges. This quantity is defined without choosing an orientation for the tetrahedron. The weight $W(j_1, j_2, \dots)$ is non-zero only if the following admissibility conditions are satisfied: $j_1+j_2+j_3 \in \mathbb{Z}$ and $j_1+j_2-j_3 \ge 0$ up to permutations. These conditions ensure that each triangle carries a non-degenerate Euclidean geometry.

The partition function is given by the sum over all possible spin assignments (colourings) in the interior of the triangulated manifold:
\begin{equation}
\label{eq:PR}
Z=\sum_{j_1, j_2, \dots}W(j_1, j_2, \dots).
\end{equation}

In the semiclassical (large-spin) limit, the $6j$ symbol relates to the Regge action for a single tetrahedron. The weight $W(j_1, j_2, \dots)$ can therefore be interpreted as a discrete path-integral weight for three-dimensional Einstein-Hilbert gravity written in Regge form. The model~\eqref{eq:PR} can thus be interpreted as a discrete model for three-dimensional Einstein gravity, where the sum runs over metrics in the interior of the manifold subject to fixed boundary data.

However, this is not the full story. First, the set of irreducible representations is infinite, so the sum may diverge and a regularisation is required to make it well defined. Second, the partition function should be independent of the triangulation of the manifold: if $M$ is closed, $Z$ should depend only on the topology of $M$; if $M$ has a boundary, $Z$ should depend only on the topology of $M$ and the boundary data (which are fixed and constitute the physical boundary conditions of the problem). Note that these requirements correspond precisely to the existence of a well-defined amplitude $Z_{M;\delta}$ independent of the bulk triangulation, as assumed in Sec.~\ref{Sec:Continuum1}.

Several regularisation strategies have been proposed. Ponzano and Regge used a simple cutoff procedure, namely an upper bound on the spin variables, which is then removed by a suitable rescaling. However, the resulting limit is not well defined in every case (see e.g.~\cite{Barrett_Naish-Guzman_2008}), and therefore this simple cutoff regularisation does not in general lead to invariance under changes of triangulation.

Another strategy is to view the Ponzano-Regge model as a limit of the Turaev-Viro model, where $SU(2)$ is replaced by its quantum deformation $U_q(\mathfrak{sl}(2,\C))$. When the deformation parameter $q$ is a root of unity, only finitely many irreducible representations appear and the partition function is always well defined. The regularisation of the Ponzano-Regge model then consists in taking the limit $q\to 1$.\footnote{Also, in~\cite{Barrett_Naish-Guzman_2008} Barrett and Naish-Guzman gave a reformulation of the Ponzano-Regge model (with observables) in terms of $SU(2)$ group variables on the dual edges, where regularisation is defined by suitably removing excess delta functions. They provide a topological criterion for the partition function to be well defined; when satisfied, the partition function can be expressed in terms of a topological invariant, and so is independent on the triangulation (as well as on the regularisation procedure).}

\subsubsection{State Sum Invariants}\label{sec:invariants}

For the purpose of illustrating our general results, we work with the Turaev-Viro regularisation of the Ponzano-Regge model and therefore consider a finite set of spins $I$. Given a triangulation $\Delta$ of the manifold, a colouring is a map assigning to each edge of $\Delta$ an element of $I$. For an edge $e$, the weight associated with its coloured version $c(e)$ is denoted by $w_{c(e)}$; for a tetrahedron $\tau$, the $6j$ symbol associated with its colouring $c(\tau)$ is denoted by $|6j|_{c(\tau)}$.

Let $M$ be a closed triangulated three-dimensional manifold, and let $v$ denote the number of its vertices. The state sum for a closed manifold is defined by
\begin{equation}
\label{eq:TV_closed}
Z_M = \sum_{c(\Delta)} w^{-2v} \prod_{e} w_{c(e)}^2 \prod_{\tau} |6j|_{c(\tau)} .
\end{equation}
In contrast to Eq.~\eqref{eq:PR_weight}, there is no weight associated to triangles, since this contribution is absorbed into the definition of $|6j|$.

Turaev and Viro proved that the quantity $Z_M$ does not depend on the choice of triangulation of $M$. In the language of Sec.~\ref{Sec:Continuum1}, this means that the net of amplitudes over $\mathcal T_M^\delta$ is constant, so that the continuum limit~\eqref{eq:Z_Md} exists trivially.

Let now $M$ be a compact triangulated three-dimensional manifold with boundary triangulation $\delta$. Let $v$ denote the number of vertices of $M$, and let $b$ denote the number of vertices lying on the boundary. The colouring of the boundary triangulation $\delta$ is denoted by $\alpha$. The state sum for a manifold with boundary is defined by
\begin{equation}
\label{eq:TV_open}
Z_{M;\delta}(\alpha) =
\sum_{c}
w^{-2v+b}
\prod_{\text{boundary } e} w_{c(e)}
\prod_{\text{bulk } e} w_{c(e)}^2
\prod_{\tau} |6j|_{c(\tau)},
\end{equation}
where the sum is taken over admissible colourings extending the boundary colouring $\alpha$.

Turaev and Viro proved that, for any compact three-manifold $M$ with triangulated boundary and any admissible colouring $\alpha$ of $\delta$, all extensions of the boundary triangulation to the bulk yield the same value of $Z_{M;\delta}(\alpha)$. This provides an explicit realisation of the assumption made in Sec.~\ref{Sec:Continuum1} that $Z_{M;\delta}$ is independent of the bulk triangulation.

For $\partial M=\Sigma_+ \sqcup \Sigma_-$ with $\delta=\delta_+ \sqcup \delta_-$ we introduce the notation $Z_{M;\delta_+,\delta_-}(\alpha,\beta)\coloneqq Z_{M;\delta_+ \sqcup \delta_-}(\alpha \sqcup\beta)$, where $\alpha$ is a colouring on $\delta_+$ and $\beta$ is a colouring on $\delta_-$.

\subsubsection{Functional Nature of the Invariants}

Let $\Sigma$ be a closed surface and $\delta$ a triangulation of $\Sigma$. 
One defines a complex vector space $C(\Sigma,\delta)$ generated by admissible colourings of $(\Sigma,\delta)$ and equips it with the inner product for which admissible colourings form an orthonormal basis. 
Upon completion with respect to the induced norm, this defines a Hilbert space $\mathcal H_{\delta}$.
In the case $\Sigma=\varnothing$, one sets $\mathcal H_\varnothing=\mathbb{C}$. This construction realises concretely the assignment $\delta\mapsto\mathcal H_\delta$ of the definition of a spin foam model given in Section~\ref{Sec:truncatedaxioms}. 
We now describe how the amplitudes $Z_{M;\delta}$ naturally define linear maps between these boundary Hilbert spaces. 
It is convenient to express this structure in the language of cobordisms, which provides a compact formulation of the gluing property and of the cylinder projector discussed in Sec.~\ref{Sec:weaklimTQFT}.

A cobordism $W=\{M;(\Sigma_+,\delta_+),(\Sigma_-,\delta_-)\}$ between $(\Sigma_+,\delta_+)$ and $(\Sigma_-,\delta_-)$ consists of a compact $3$-manifold $M$ together with embeddings of $(\Sigma_\pm,\delta_\pm)$ into $\partial M$. 
To such a cobordism one associates a $\mathbb{C}$-linear homomorphism
\begin{equation}
\Phi_W : \mathcal H_{\Sigma_+,\delta_+} \to \mathcal H_{\Sigma_-,\delta_-},
\end{equation}
defined by
\begin{equation}
\label{eq:PhiQ}
\Phi_W(\alpha)=\sum_\beta Z_{M;\delta_+,\delta_-}(\alpha,\beta)\,\beta.
\end{equation}

A closed manifold $M$ can be viewed as a cobordism $\varnothing \to \varnothing$, in which case the corresponding map $\Phi_W:\mathbb{C}\to \mathbb{C}$ is given by multiplication by the invariant $Z_M$.

Because the state sum $Z_{M;\delta}$ is independent of the extension of the boundary triangulation to the bulk, the homomorphism $\Phi_W$ depends only on the boundary data and is therefore well defined.

We now show that these maps are compatible with composition, which reproduces the gluing property~\eqref{eq:gluing}. 
Consider two cobordisms $W_1=\{M_1;(\Sigma_1,\delta_1),(\Sigma_2,\delta_2)\}$ and $W_2=\{M_2;(\Sigma_2,\delta_2),(\Sigma_3,\delta_3)\}$. Their composition is the cobordism $W_2\circ W_1=\{M_1\cup M_2;(\Sigma_1,\delta_1),(\Sigma_3,\delta_3)\}$ obtained by gluing along $\Sigma_2$. 
The associated maps satisfy, as follows from the definition,
\begin{equation}
\Phi_{W_2\circ W_1}=\Phi_{W_2}\circ \Phi_{W_1},
\end{equation}
which is precisely the gluing property~\eqref{eq:gluing}.

Consider next the unit cobordism $\mathrm{id}_{\Sigma,\delta}=\{\Sigma\times[0,1];(\Sigma,\delta),(\Sigma,\delta)\}$. From compatibility with composition and the identity 
$\mathrm{id}_{\Sigma,\delta}\circ \mathrm{id}_{\Sigma,\delta}=\mathrm{id}_{\Sigma,\delta}$ 
it follows that
\begin{equation}
\Phi_{\mathrm{id}_{\Sigma,\delta}}\circ \Phi_{\mathrm{id}_{\Sigma,\delta}}
=\Phi_{\mathrm{id}_{\Sigma,\delta}},
\end{equation}
so that $\Phi_{\mathrm{id}_{\Sigma,\delta}}$ is an idempotent endomorphism of $\mathcal H_{\delta}$. This reproduces the cylinder projector property~\eqref{eq:w-proj}. The assignments
\begin{equation}
(\Sigma,\delta)\mapsto\mathcal H_{\delta},
\qquad
W\mapsto\Phi_W
\end{equation}
therefore respect composition. 
However, for the unit cobordism one may have
\begin{equation}
\Phi_{\mathrm{id}_{\Sigma,\delta}}\neq \mathbbm{1}_{\mathcal H_{\delta}},
\end{equation}
so that these assignments define only a semifunctor.~To restore functoriality, we first analyse the kernel of the cylinder operator. Elements in the kernel of $\Phi_{\mathrm{id}_{\Sigma,\delta}}$ are annihilated by every cobordism map. Indeed, for any cobordism $W$ one has
\begin{equation}
\Phi_W\circ \Phi_{\mathrm{id}_{\Sigma,\delta}}=\Phi_W,
\end{equation}
so that $\Phi_W(\alpha)=0$ whenever $\Phi_{\mathrm{id}_{\Sigma,\delta}}(\alpha)=0$.

One therefore defines the quotient
\begin{equation}
Q_{\delta}:=\mathcal H_{\delta}/\ker(\Phi_{\mathrm{id}_{\Sigma,\delta}}),
\end{equation}
which coincides with the image of the idempotent $\Phi_{\mathrm{id}_{\Sigma,\delta}}$. 
On $Q_{\delta}$ the cylinder acts as the identity.

Each homomorphism $\Phi_W$ then induces a well-defined $\mathbb{C}$-linear map
\begin{equation}
\Psi_W:Q_{\delta_+}\to Q_{\delta_-},
\end{equation}
since $\Phi_W$ maps $\ker(\Phi_{\mathrm{id}_{\Sigma_+,\delta_+}})$ into $\ker(\Phi_{\mathrm{id}_{\Sigma_-,\delta_-}})$.
The assignments
\begin{equation}
(\Sigma,\delta)\mapsto Q_{\delta},
\qquad
W\mapsto\Psi_W
\end{equation}
satisfy
\begin{equation}
\Psi_{W_2\circ W_1}=\Psi_{W_2}\circ \Psi_{W_1},
\qquad
\Psi_{\mathrm{id}_{\Sigma,\delta}}=\mathbbm{1}_{Q_{\delta}},
\end{equation}
and therefore define a functor from the category of cobordisms of triangulated surfaces to the category of $\mathbb{C}$-modules.

Lastly, although $Q_{\delta}$ is defined using a triangulation $\delta$, it depends only on the topology of $\Sigma$. 
Given two triangulations $\delta_0$ and $\delta_1$ of $\Sigma$, one can choose a triangulation of the cylinder $\Sigma\times[0,1]$ restricting to $\delta_0$ and $\delta_1$ on the two boundary components. 
This defines a cobordism
\begin{equation}
W=\{\Sigma\times[0,1];(\Sigma,\delta_0),(\Sigma,\delta_1)\}
\end{equation}
and hence an induced map
\begin{equation}
\Psi_W:Q_{\delta_0}\to Q_{\delta_1}.
\end{equation}
By triangulation invariance this map is independent of the chosen interpolation and is an isomorphism. 
Thus, all $Q_{\delta}$ are canonically isomorphic and are identified with a single space $\mathcal H_\Sigma$.\footnote{This space has been shown in~\cite{Ooguri:1991ib, Ooguri_1992a} to be isomorphic to the space of gauge-invariant functions on the moduli space of flat $SU(2)$ connections over the 2-dimensional boundary surface, thus providing the equivalence of the Turae-Viro state sum model with the quantisation of 3d gravity in the continuum by Witten~\cite{Witten_1988b}.} 
This provides us with a concrete realisation of Proposition~\ref{prop:w-range}.

One therefore obtains a functor
\begin{equation}
\Sigma\longmapsto\mathcal H_\Sigma,
\qquad
W\longmapsto\Psi_W,
\end{equation}
from the category of cobordisms of topological surfaces to the category of $\mathbb{C}$-modules. As such, the theory so obtained is a TQFT (cfr.~Sec.~\ref{Sec:TQFT} and corollary \ref{cor:PachnerInv} in Sec.~\ref{Sec:weaklimTQFT}).

\section{Refinement and Continuum Limit}\label{Sec:Continuum2}

The limit procedure defined in Sec.~\ref{Sec:Continuum1} led us to find a TQFT in the continuum.~The reason for this is that a refinement procedure based on the weak order~\eqref{eq:weakorder}, while keeping the boundary triangulation fixed, forces the theory to be insensitive to the local truncation of degrees of freedom on the boundary hypersuface which we glue along (cfr.~Prop.~\ref{prop:w-gluing} and Cor.~\ref{cor:w-proj}).~In this Section, we aim to provide a  notion of refinement that is more closely aligned to the physical intuition of capturing more degrees of freedom in a sensitive way.~In particular, we connect it with the mathematical concept of \textit{stellar subdivision}, as already suggested e.g.~in \cite{Dittrich_Geiller_2015}.

\subsection{Subdivision and Inductive Limit}\label{Sec:SubdivIndLim}

Pachner moves are not the only option to modify a triangulation without changing the underlying manifold. When a triangulation is, in some suitable sense (cfr.~Def.~\ref{def:subdiv}), finer than a given one, we refer to it as a \textit{subdivision}.~There are many ways to produce subdivisions. When a subdivision is obtained through a sequence of refining Alexander moves, it is called a \textit{stellar subdivision} (see App.~\ref{App:Triang} for details).~This will be our notion of refinement.~Notice that, in contrast to Pachner moves, Alexander moves can alter the boundary triangulation.~Furthermore, they induce a partial order $\preceq$ on the set $\mathcal{T}_M$ of triangulations diffeomorphic to a manifold $M$, thereby turning it into a partially ordered, directed set (for a proof, see Prop.~\ref{refin:poset}).\footnote{We shall sometimes refer to the partial order $\preceq$ as the \emph{strong order} to distinguish it from the weak order of Sec.~\ref{Sec:weakorder}.}

Since the Alexander moves change the boundary triangulation, the map $Z$ does not provide a suitable definition for a net.~Indeed, the map $Z_{\bullet}:\mathcal{T}_M\to ?$ given by $Z_\Delta=Z(\Delta)$ jumps from a Hilbert space to another, so that the target space of $Z_{\bullet}$ is not defined.~To sidestep this issue, we shall consider the inductive limit of a suitable inductive family of boundary Hilbert spaces.~To this aim, let us first recall the definition of inductive limit. 
\begin{definition}[Inductive limit]\label{def:inductive}
    A inductive family $(\{A_i\}_{i\in I},\{\iota_{ji}\}_{j\geq i\in I})$ is a collection of algebraic spaces $A_i$, with $i$ index in some directed poset $I$, together with a collection of injective maps $\iota_{ji}:A_i\to A_j,\ j\geq i\in I$, such that
    \begin{enumerate}
        \item $\iota_{ii}=\mathbbm{1}_{A_i}$,
        \item $\iota_{kj}\circ\iota_{ji}=\iota_{ki},\ i\leq j\leq k\in I$.
    \end{enumerate}
    The inductive limit is the space defined by
    \[A^{\rm ind}=\lim_{\longrightarrow}A_i=\bigsqcup_{i\in I}A_i/\sim,\]
    where $(x_i,i)\sim (y_j,j)$, for $x_i\in A_i$, $y_j\in A_j$, iff there exists $k\geq i,j$ such that $\iota_{ki}(x_{i})=\iota_{kj}(y_{j})$. With an inductive limit, we obtain also canonical injections $\iota_i:A_i\to A^{\rm ind}$ given by $x_i\mapsto [x_i,i]$, with $[x_i,i]$ the equivalence class of $x_i$.
\end{definition}
The inductive family of our interest is $(\{\mathcal H\}_{\delta\in\mathcal T_{\Sigma}},\{\iota_{\delta'\delta}\}_{\delta'\succeq\delta})$, consisting of the collection of boundary Hilbert spaces $\Hil_\delta$ with $\delta\in\mathcal{T}_{\Sigma}$, for a fixed connected closed $(d-1)$-dimensional manifold $\Sigma$, together with isometric linear injective maps $\iota_{\delta'\delta}:\Hil_{\delta}\to\Hil_{\delta'}$, $\delta'\succeq\delta\in\mathcal{T}_{\Sigma}$.~The inductive limit $\h_{\Sigma}^{\rm ind}=\lim_{\delta\in\mathcal{T}_{\Sigma}}\h_{\delta}$ is a Hilbert space (see Lemma~\ref{rm:lim} in App.~\ref{App:Proofs2} for a proof).~For the boundary spaces given by~\eqref{eq:bdyHkin}, a typical example of the inductive maps is provided by the embeddings of spin networks into bigger graphs by addition of trivial $SU(2)$ representations on the extra links.~In this case, the inductive limit is obtained by imposing an equivalence relation according to which two states are equivalent if they can be related by this identification, or if they are mapped into each other by the group of the automorphisms of the graph.

Before introducing the net associated to each $d$-dimensional manifold $M$, which will be the object of study in the remainder of the section, let us discuss how the properties of the truncated theory given in Sec.~\ref{Sec:truncatedaxioms} can be translated at the level of $\h_{\Sigma}^{\rm ind}$.~First of all, it is easy to show that an inductive limit is defined also in the case in which $\Sigma$ has finite connected components and the limit factorizes into the tensor product of the inductive limits associated to each component:
\begin{proposition}
\label{prop:factor}
    $\h_{\Sigma_1\sqcup\Sigma_2}^{\rm ind}$ and $\h_{\Sigma_1}^{\rm ind}\otimes\h_{\Sigma_2}^{\rm ind}$ are isomorphic as Hilbert spaces.
\end{proposition}
For the details of the construction of the inductive family and the proof of the Prop.~\ref{prop:factor}, see App.~\ref{App:Proofs2}.\\
To take into account the dagger structure of the truncated level, we require the compatibility condition $\iota_{\delta_2^*\delta_1^*}\circ R_{\delta_1}=R_{\delta_2}\circ \iota_{\delta_2\delta_1}$, where $R_\delta$ is the Riesz map $R_\delta:\h_\delta\to \h^*_\delta$.~This allows us to establish the following result, whose proof is given in App.~\ref{App:Proofs2}.
\begin{proposition}
\label{prop:dual}
    $\h_{\Sigma^*}^{\rm ind}= (\h_{\Sigma}^{\rm ind})^*$, and the following diagram
    \begin{equation}
    \label{dualcd}
    \begin{tikzcd} 
& \h_{\delta_1} \arrow{ld}[swap]{\iota_{\delta_1}} \arrow[r, "R_{\delta_1}"] &  \h_{\delta_1}^* \arrow[rd, "\iota_{\delta_1^*}"] & \\ 
      \h_{\Sigma}^{\rm ind} \arrow{rrr}[swap]{R_{\Sigma}}& & & (\h_{\Sigma}^{\rm ind})^*\\
       & \h_{\delta_2} \arrow[from=uu, crossing over, "\iota_{\delta_2\delta_1}" near start]\arrow[lu, "\iota_{\delta_2}"]\arrow{r}[swap]{R_{\delta_2}} &  \h_{\delta_2}^* \arrow{ru}[swap]{\iota_{\delta_2^*}} \arrow[from=uu, crossing over, "\iota_{\delta_2^*\delta_1^*}" near start]&  
    \end{tikzcd}
\end{equation}
commutes.~Furthermore, $\h_{\varnothing}^{\rm ind}\simeq \C$.
\end{proposition}
We are now ready to define a net associated to each $d$-dimensional manifold $M$.~This is given by
\begin{equation}
\label{eq:ind-net}
\begin{matrix}
    \tilde{Z}_\bullet:& \mathcal{T}_M &\to & \Hil^{\rm ind}_{\del M}\\
     & \Delta &\mapsto &\iota_{\del\Delta}\left(Z(\Delta)\right)\,,
\end{matrix}
\end{equation}
and satisfies the following properties:
\begin{enumerate}
    \item \textbf{Normalisation:} $\tilde{Z}_{\varnothing}=1$;
    \item \textbf{Dagger:} $\tilde{Z}_{\Delta^*}=R_{\Sigma}(\tilde{Z}_{\Delta})$, with $R_{\Sigma}$ the Riesz map $R_{\Sigma}:\Hind\to(\Hind)^*$;
    \item \textbf{Gluing:} Let $\Delta_1\in\mathcal{T}_{M_1}$ and $\Delta_2\in\mathcal{T}_{M_2}$ be such that $\partial\Delta_1=\delta_1^*\sqcup\delta$ and $\partial\Delta_2=\delta^*\sqcup\delta_2$, with $|\delta|\simeq \Sigma$.~Then,
    $$
    \tilde{Z}_{\Delta_1\cup_{\delta}\Delta_2}=\braket{\tilde{Z}_{\Delta_1},\tilde{Z}_{\Delta_2}}_{\Sigma}\;,
    $$
    where $\langle\cdot,\cdot\rangle_{\Sigma}$ denotes the dual pairing between $\h_{\Sigma}^{\rm ind}$ and $(\h_{\Sigma}^{\rm ind})^*$.
\end{enumerate}

Property 1 follows directly from the normalisation property of the truncated theory (cfr.~Def.~\ref{def:truncSF}).~As for property 2, we have
$$
\tilde{Z}_{\Delta^*}\overset{\eqref{eq:ind-net}}{=}\iota_{\partial\Delta^*}(Z(\Delta^*))=\iota_{\partial\Delta^*}\circ R_{\partial\Delta}(Z(\Delta))=R_{\Sigma}\circ\iota_{\partial\Delta}({Z}(\Delta))\overset{\eqref{eq:ind-net}}{=}R_{\Sigma}(\tilde{Z}_{\Delta})\;,
$$
where, in the second equality, we used the dagger property of the truncated theory (cfr.~Def.~\ref{def:truncSF}) and, in the third equality, we used~\eqref{dualcd}.\footnote{We recall from Sec.~\ref{Sec:TQFT} (see also property 5 in Def.~\ref{def:truncSF} of Sec.~\ref{Sec:truncatedaxioms}) that the property is referred to as \emph{dagger property} because, interpreting $Z(\Delta)$ as an operator, it is related to the adjoint operation, namely $Z(\Delta^*)=Z(\Delta)^\dagger$.} Finally, property 3 follows from Prop.~\ref{prop:factor} together with the gluing property at the truncated level (cfr.~Def.~\ref{def:truncSF}).~Namely,
\begin{equation}
    \begin{split}
        \tilde{Z}_{\Delta_1\cup_{\delta}\Delta_2}&=\iota_{\delta_1^*}\otimes\iota_{\delta_2}({Z}(\Delta_1\cup_{\delta}\Delta_2))\\
        &=\iota_{\delta_1^*}\otimes\iota_{\delta_2}(\braket{Z(\Delta_1),Z(\Delta_2)}_{\delta})\\
        &=\braket{\iota_{\delta_1^*}\otimes\iota_{\delta}(Z(\Delta_1)),\iota_{\delta^*}\otimes\iota_{\delta_2}(Z(\Delta_2))}_{\Sigma}.
    \end{split}
\end{equation}
    
\subsection{A No-Go Theorem}\label{Sec:NoGo}

For given $\Delta$ such that $|\Delta|=M$, the states $\tilde{Z}_{\Delta}=\iota_{\del\Delta}\left(Z(\Delta)\right)\in\h_{\del M}^{\rm ind}$ represent the truncation of the field modes as seen in the continuum.~One might therefore be hoping that a limit procedure based on the improved notion of refinement, which affects both the bulk and the boundary, would not result in a TQFT.~However, some further care is still required in that the details of the convergence of the nets in~\eqref{eq:ind-net} will also matter.~We can, in fact, put a bound on the good behaviour of the convergence in order for the theory to include non-topological properties.~Specifically, the following theorem holds.
\begin{theorem}[No-go theorem]
\label{Thm:no-go}
    Suppose that, for each $M$, the net $\tilde{Z}_\bullet: \mathcal{T}_M\to \Hil^{\rm ind}_{\del M}$ converges to an element
    \begin{equation}
        Z(M)=\lim_{\Delta\in\mathcal{T}_M}\tilde{Z}_{\Delta}\in\Hil^{\rm ind}_{\del M}.
    \end{equation}
    Then, the map 
    \begin{equation*}
        \begin{matrix}
            \Sigma\mapsto\Hind,\\
            M\mapsto Z(M)
        \end{matrix}
    \end{equation*}
    defines a Topological Quantum Field Theory.
\end{theorem}
The proof can be found in Appendix~\ref{App:Proofs2}. The idea of the proof is the following.~Owing to convergence, we may restrict our attention to the subnet of triangulations that split appropriately under the chosen cut.~Norm convergence for every manifold then allows us to rewrite the limit over such triangulations as two independent limits, one for each of the two pieces of the manifold.

A more physical intuition arises from interpreting the states $\tilde{Z}_{\Delta}$ as truncations of the continuum degrees of freedom.~From this perspective, the crucial step in the proof is that truncating the degrees of freedom at the interface, gluing the resulting states via the inner product, and then taking the continuum limit is equivalent to first taking the continuum limit of the truncated states and only afterwards performing the gluing.~This interchange is possible precisely because the inductive Hilbert space provides a consistent embedding of all finite truncations into the continuum, so that its elements encode the same continuum state at different levels of truncation.~As a consequence, the gluing operation is insensitive to the particular truncation used at the interface and depends only on the topological decomposition of the manifold, thereby recovering the characteristic gluing behaviour of a TQFT.

This is a no-go theorem in the sense that a model that aims to describe gravity in four dimensions cannot admit such a limit.~In other words, a necessary condition to have a theory that it is not topological is that at least one net $\tilde{Z}_{\bullet}:\mathcal{T}_M\to \Hil^{\rm ind}_{\del M}$ is not a Cauchy net in $\Hil^{\rm ind}_{\del M}$. We refer to Appendix~\ref{App:Triang} for the definition of a Cauchy net.~This motivates us to consider more general cases in which, for some $M$, the limit does not converge in the inductive Hilbert space.~More specifically, in the next subsection, we consider the possibility of non-normalisable (distributional) states as result of the limit.~This will allow us to make contact with the well-known tools from Refined Algebraic Quantisation~\cite{Giulini_Marolf_1999a,Giulini_Marolf_1999b}.

\subsection{Distributional Continuum Limit and Rigging Map}

\noindent
In order to break the hypothesis of Theorem~\ref{Thm:no-go}, and still have the limit well-defined for each manifold, we generalise the notion of convergence using the concept of \emph{Gelfand triple}.~Without loss of generality, we suppose there exists a topological vector space $\mathcal{D}_{\Sigma}$ as a dense subset of $\mathcal{H}_{\Sigma}^{\rm ind}$ such that the limit of the nets \eqref{eq:ind-net} converges in its algebraic dual $\mathcal{D}_{\Sigma}'$.~More specifically, the net $\tilde{Z}_{\bullet}$ can be seen as a net in the continuous dual $R_{\Sigma}(\tilde{Z}_{\Delta})\in(\mathcal{H}_{\Sigma}^{\rm ind})^*$ equipped with the norm (or strong) topology.~However, every functional on $\mathcal{H}_{\Sigma}^{\rm ind}$ is also a functional on $\mathcal{D}_{\Sigma}$, thus there exists a set-theoretic inclusion $(\mathcal{H}_{\Sigma}^{\rm ind})^*\hookrightarrow\mathcal{D}_{\Sigma}'$, which is also a topological immersion since the topology on the strong dual $(\mathcal{H}_{\Sigma}^{\rm ind})^*$ is finer than the pointwise topology on the algebraic dual $\mathcal{D}_{\Sigma}'$.~This defines the Gelfand triple
\[\mathcal{D}_{\Sigma}\subset\Hind\simeq(\Hind)^*\subset\mathcal{D}'_{\Sigma}.\]
While the net lies in $(\mathcal{H}_{\Sigma}^{\rm ind})^*$, there is no obstruction to ask that the limit exists in $\mathcal{D}_{\Sigma}'$.~For each $M$, let us denote the limit by $Z(M)\in\mathcal{D}_{\Sigma}'$.~That is, for all $\psi\in\mathcal{D}_{\Sigma}$,
\begin{equation}
\label{dLim}
    Z(M)[\psi]=\lim_{\Delta\in\mathcal{T}_M}\braket{\tilde{Z}_{\Delta}|\psi}_{\mathcal H_{\Sigma}^{\rm ind}}\;.
\end{equation}
In what follows, we shall make a couple of assumptions on the choice of the dense subset $\mathcal D_\Sigma$.
\begin{enumerate}
    \item\label{Assumption1} \textbf{Inner product space:} $\mathcal{D}_{\Sigma}$ is a topological vector space and a dense subset of $\mathcal{H}_{\Sigma}^{\rm ind}$.~It is equipped with the inner product induced by $\mathcal{H}_{\Sigma}^{\rm ind}$ and inherits its topology.~Its algebraic dual $\mathcal{D}_{\Sigma}'$, with the topology of pointwise convergence, is a Hausdorff topological vector space.
    \item\label{Assumption2}\textbf{Factorisation:} $\mathcal{D}_{\bigsqcup_{i}\Sigma_i}=\bigotimes_i\mathcal{D}_i$, with $\{\Sigma_i\}$ a finite collection of closed connected $(d-1)$-manifolds. $\mathcal{D}_{\bigsqcup_{i}\Sigma_i}$ is endowed naturally with an inner product, hence with a topology, derived by the tensor product structure.
    \item\label{Assumption3}\textbf{Involution:} $\mathcal{D}_{\Sigma}\simeq\mathcal{D}_{\Sigma^*}$.~More precisely, $\mathcal{D}_{\Sigma}$ and $\mathcal{D}_{\Sigma^*}$ are anti-isomorphic via Riesz map.~Namely, given the Riesz map $R_{\Sigma}:\mathcal{H}_{\Sigma}^{\rm ind}\to(\mathcal{H}_{\Sigma}^{\rm ind})^*$, the two dense subspaces satisfy $R_{\Sigma}(\mathcal{D}_{\Sigma})=\mathcal{D}_{\Sigma^*}$.
\end{enumerate}

\noindent
In Loop Quantum Gravity, $\cD_\Sigma$ is typically taken to be the span (i.e., the space of finite linear combinations) of the spin-network states (or s-knots, after imposing the quantum spatial diffeomorphism constraint).

Owing to the assumptions \ref{Assumption1}{\color{blue}.}-\ref{Assumption3}{\color{blue}.}, the functional $Z(\Sigma\times I)\in\mathcal D_{\Sigma^*\sqcup\Sigma}'$ associated to the cylinder can be thought of as a map $Z(\Sigma\times I):\cD_{\Sigma^*}\otimes\cD_{\Sigma}\to\mathbb C$, so that, for any $\phi\in\mathcal D_{\Sigma}$, it yields a functional on $\cD_{\Sigma}$ given by $Z(\Sigma\times I)[R_{\Sigma}(\phi)\otimes\,\bullet\,]\in\cD_{\Sigma}'$.~This allows us to define a so-called \emph{rigging map} from $\cD_{\Sigma}$ to $\cD_{\Sigma}'$ as demonstrated by the Theorem~\ref{Thm:riggingmap} below, whose proof is given in App.~\ref{App:Proofs2}.

\begin{theorem}[Rigging map]\label{Thm:riggingmap}
Let $\Sigma$ be a closed oriented $(d-1)$-dimensional manifold.~The map $P_{\Sigma}:\cD_{\Sigma}\to\cD_{\Sigma}'$ given by
\be\label{def:PSigma}
    P_{\Sigma}(\phi)[\psi]\coloneqq Z(\Sigma\times I)[R_{\Sigma}(\phi)\otimes\psi]\;,\qquad \phi,\psi\in\cD_{\Sigma}\;,
\ee
defines a rigging map.~That is, it satisfies the following properties:
\begin{enumerate}
    \item[i.]\emph{Anti-linearity:} $P_{\Sigma}(\alpha\phi_1+\beta\phi_2)[\psi]=\left(\overline{\alpha}P_{\Sigma}(\phi_1)+\overline{\beta}P_{\Sigma}(\phi_2)\right)[\psi]$, for any $\phi_1,\phi_2,\psi\in\cD_{\Sigma}$ and $\alpha,\beta\in\mathbb C$;
    \item[ii.]\emph{Reality:} $P_{\Sigma}(\phi)[\psi]=\overline{P_{\Sigma}(\psi)[\phi]}$, for any $\phi,\psi\in\cD_{\Sigma}$;
    \item[iii.]\emph{Positive semidefiniteness:} $P_{\Sigma}(\psi)[\psi]\geq0$, for any $\psi\in\cD_{\Sigma}$.
\end{enumerate}
\end{theorem}

\noindent
Having a rigging map at our disposal allows us to define a positive semi-definite Hermitian form on $\cD_{\Sigma}$
\be\label{eq:prodphys}
(\phi,\psi)_{\Sigma}\coloneqq P_{\Sigma}(\phi)[\psi]\qquad\forall\,\phi,\psi\in\cD_{\Sigma}\;,
\ee
and, as customary in refined algebraic quantization (RAQ)~\cite{Giulini_Marolf_1999a,Giulini_Marolf_1999b,Thiemann_2007}, a physical Hilbert space.~More details on this will be provided in the next section where the properties of $P_{\Sigma}$ in relation with the physical Hilbert space are also studied.

Under the change of orientation of $\Sigma$, compatibly with the above assumption \ref{Assumption3}, the rigging map $P_{\Sigma^*}:\cD_{\Sigma^*}\to\cD_{\Sigma^*}'$, associated to $\Sigma^*\times I$, is related to $P_{\Sigma}:\cD_{\Sigma}\to\cD_{\Sigma}'$, associated to $\Sigma\times I$, via the Riesz map $R_{\Sigma}$.~More specifically, we have
\be\label{rigging:Sigmastar}
P_{\Sigma}(\phi)[\psi]=\overline{P_{\Sigma^*}\left(R_{\Sigma}(\phi)\right)\left[R_{\Sigma}(\psi)\right]}\;,
\ee
where the bar on the RHS denotes complex conjugation.~The formula \eqref{rigging:Sigmastar} is contained in Lemma~\ref{lem:Psigmastar} in App.~\ref{App:Proofs2} to which we refer for details.

In Canonical Quantum Gravity, the rigging map is a useful tool to provide the states in the kernel of quantum constraints.~Moreover, the rigging map intertwines the observables on the auxiliary Hilbert space with the actual physical observables.~A formal proposal to construct the (generalised) projector onto the kernel of the Hamiltonian constraint within the framework of spin foams and to compute the expectation values of physical observables was put forward in \cite{Rovelli:1998dx}.~This was rigorously realised in 3-dimensional loop quantum gravity with zero cosmological constant later on in \cite{Noui_Perez_2005}.~Our construction of the rigging map \eqref{def:PSigma} based on the cylinder $\Sigma\times I$ is compatible with these works, however, in our present framework, since we do not subscribe to any specific model, neither an explicit relation with the constraints nor a specified algebra of physical observables can be exhibited.~Nevertheless, we can use the compatibility of observables with the rigging map as a prescription to define physical observables.
\begin{definition}[Physical observables]
    Let $\operatorname{A}\in \mathcal{L}(\Hind)$ be such that $\cD_{\Sigma}\subset\mathrm{Dom}(\operatorname{A})\cap\mathrm{Dom}(\operatorname{A}^{\dagger})$ and $\operatorname{A}(\cD_{\Sigma})\subset\cD_{\Sigma},\,\operatorname{A}^{\dagger}(\cD_{\Sigma})\subset\cD_{\Sigma}$.~Let $\operatorname{A}$ be compatible with the rigging map, i.e.~$\operatorname{A}(\ker(P_{\Sigma}))\subset\ker(P_{\Sigma})$ and $\operatorname{A}^{\dagger}(\ker(P_{\Sigma}))\subset\ker(P_{\Sigma})$.~Then, $\operatorname{A}$ is called a \emph{physical observable} if \[P_{\Sigma}(\operatorname{A}\psi)[\phi]=P_{\Sigma}(\psi)[\operatorname{A}^\dagger\phi]\quad,\quad\  \forall\;\psi,\phi\in\cD_{\Sigma}\;.\]
\end{definition}
This definition ensures that the adjoint on $\Hind$ is also the physical adjoint. Notice that the request to be compatible with the rigging map is just to be well-defined on the image on $P_{\Sigma}$.~Indeed, that condition is equivalent to $P_{\Sigma}(\operatorname{A}(\psi+\phi))=P_{\Sigma}(\operatorname{A}\psi)$ for all $\phi\in\ker(P_{\Sigma})$.

\section{Physical Hilbert Space, Physical States, and their Properties}\label{Sec:HphysCLQG}

As anticipated in the previous section, the rigging map $P_{\Sigma}$ defined in Eq.~\eqref{def:PSigma} allows us to define the \emph{physical Hilbert space} of the theory as typically done in RAQ~\cite{Giulini_Marolf_1999a,Giulini_Marolf_1999b,Thiemann_2007}.~That is
\be\label{eq:Hphys}
\Hphys\coloneqq \overline{\bigslant{\cD_{\Sigma}}{\ker(P_{\Sigma})}}^{\scriptscriptstyle \norm{\cdot}}\;,
\ee
where the closure is taken with respect to the norm induced by the positive semi-definite Hermitian form~\eqref{eq:prodphys}. Indeed, since the left and right radical of the Hermitian form~\eqref{eq:prodphys} are both equal to $\ker(P_{\Sigma})$\footnote{The two radicals are equal because the form is Hermitian.~Furthermore, by Eq.~\eqref{eq:prodphys}, clearly $\ker(P_{\Sigma})$ is contained in the left radical.~Moreover, if $\phi$ is in the left radical, that is $(\phi,\psi)_{\Sigma}=P_{\Sigma}(\phi)[\psi]=0$ for any $\psi\in\cD_{\Sigma}$, then $P_{\Sigma}(\phi)=0\in\cD_{\Sigma}'$ so that $\phi\in\ker(P_{\Sigma})$.}, the form passes to the quotient, defining a inner product
\begin{equation}\label{def:physinnprod}
    \braket{\phi_{\mathrm{phys}}|\psi_{\mathrm{phys}}}_{\Hphys}\coloneqq P_{\Sigma}(\phi)[\psi],
\end{equation}
where $\phi_{\mathrm{phys}},\psi_{\mathrm{phys}}$ are the equivalence classes of $\phi,\psi\in\cD_{\Sigma}$, respectively. This construction induces a inner product also on $\mathrm{Ran}(P_{\Sigma})$ simply by
\begin{equation}\label{def:RanPinnprod}
    \braket{P_{\Sigma}(\phi)|P_{\Sigma}(\psi)}_{\mathrm{Ran}(P_{\Sigma})}=P_{\Sigma}(\psi)[\phi],
\end{equation}
and $P_{\Sigma}$ turns out to be an isometric anti-isomorphism from $\Hphys$ to $\overline{\mathrm{Ran}(P_{\Sigma})}^{\scriptscriptstyle \norm{\cdot}}$.~The Hilbert space $\Hphys$ is the physical Hilbert space of the theory: in quantum gauge theories, this is the space of gauge-invariant states \cite{Halliwell:1990qr}; in quantum gravity, it is given by the space of solutions to the Wheeler--DeWitt equation~\cite{Hartle_Hawking_1983,Halliwell:1990qr,Rovelli:1998dx,Noui_Perez_2005} (see also \cite{Alesci_Thiemann_Zipfel_2012,Thiemann_Zipfel_2014}).

In what follows, we aim to study the properties of $\Hphys$ as defined in Eq.~\eqref{eq:Hphys}, i.e.~involution, normalisation, factorisation, and gluing.~In particular, we aim to investigate whether or not the objects $Z(M)$ can be regarded as physical states as one might expect from formal considerations about the gravitational path integral \cite{Hartle_Hawking_1983,Halliwell:1990qr}, and concretely realised in 3-dimensional quantum gravity (cfr.~Sec.~\ref{Sec:examplePRTVO} and the discussion below Def.~\ref{def:TQFT} in Sec.~\ref{Sec:TQFT}).~As we shall see below, even though $\Hphys$ shares some properties with the Hilbert spaces arising in TQFTs, a major difference emerges in the gluing property, which can be formulated only in a weaker form, namely as a convolution property.~This in turn leads to a different, yet crucial, interpretation of the objects $Z(M)$.

\subsection{Involution and Factorisation}\label{Sec:HphysInvFact}

Let us start with the involution property.~The following proposition, whose proof is given in App.~\ref{App:Proofs3}, shows that it holds also at the level of the physical Hilbert space.  
\begin{proposition}[Involution]
\label{prop:physdual} Let $\Sigma$ be a closed oriented $(d-1)$-manifold, $\Sigma^*$ the manifold with opposite orientation, and $\Hphys$ as defined in~\eqref{eq:Hphys}.~Then,
    $\Hil^{\rm phys}_{\Sigma^*}\simeq(\Hphys)^*.$
\end{proposition}
\noindent
As for the characterisation of the empty manifold $M=\varnothing$, the following results are a simple consequence of Eqs.~\eqref{dLim} and \eqref{def:PSigma}, together with the properties $\mathcal{T}_\varnothing=\varnothing$, $\tilde{Z}_\varnothing=1$, and $\Hil^{\rm ind}_{\varnothing}=\C$ (cfr.~Sec.~\ref{Sec:SubdivIndLim}).
\begin{proposition}[Empty manifold and normalisation]
    Let $M=\varnothing$ be the empty manifold.~Then, the following statements hold true:
    \begin{enumerate}
        \item $Z(\varnothing)=1\in\C$;
        \item $P_{\varnothing}(z_1)[z_2]=\bar{z}_1z_2,$ for all $z_1,z_2\in\cD_{\varnothing}=\C$;
        \item $\Hil^{\rm phys}_{\varnothing}=\C$.
    \end{enumerate}
\end{proposition}
Coming now to the factorisation of the physical Hilbert space, clearly, it depends on the dynamics and, in particular, on the properties of the amplitudes $Z((\Sigma_1\sqcup\Sigma_2)\times I)$ of two disconnected cylinders. To spell this out explicitly, we have to go back for a moment to the theory in the discrete as defined in Sec.~\ref{Sec:truncatedaxioms}, and specify the behaviour of the truncated spin foam amplitudes for disconnected manifolds.~As we will comment on shortly below, motivated by the behaviour of known spin foam models, as well as by the Hamiltonian constraint constructed in canonical Loop Quantum Gravity, let us introduce the notion of canonical spin foam models according to the following definition.
\begin{definition}[Canonical spin foam model]
    \label{def:split}
    A $d$-dimensional spin foam model is said to be \emph{canonical} if, in addition to the properties given in \emph{Def.~\ref{def:truncSF}}, the discrete amplitudes satisfy
    \begin{equation}\label{eq:canSF}
        Z(\Delta_1\sqcup\Delta_2)=Z(\Delta_1)\otimes Z(\Delta_2)\in\h_{\partial\Delta_1}\otimes\h_{\partial\Delta_2}\;,
    \end{equation}
    for any $d$-dimensional triangulations $\Delta_1,\Delta_2$ of compact oriented manifolds $M_1,M_2$.
\end{definition}
This is typically the case for the spin foam models whose amplitudes are given by \eqref{eq:SFampl} in which, as discussed in Sec.~\ref{Sec:SFreview}, the product of individual local discretised amplitudes associated to separate (disjoint) components reflects the locality principle.~It is worth noticing that, however, the requirement \eqref{eq:canSF} is not completely harmless as it forces us to consider the whole spacetime to be consisting of solely the two disconnected components. Indeed, owing to the factorisation property of the inductive family of the boundary Hilbert spaces (cfr.~Lemma~\ref{connIndF}), the requirement \eqref{eq:canSF} implies an analogous decomposition of the net \eqref{eq:ind-net} for disconnected $\Delta_1,\Delta_2$, namely $\tilde{Z}_{\Delta_1\sqcup\Delta_2}=\tilde{Z}_{\Delta_1}\otimes\tilde{Z}_{\Delta_2}$. As a consequence, for a canonical spin foam model, the rigging map in the continuum decomposes according to the following lemma proved in Appendix~\ref{App:Proofs3}.
\begin{lemma}\label{lem:factphys}
    For a canonical $d$-dimensional spin foam model as given in \emph{Def.~\ref{def:split}}, the rigging map, as constructed in \eqref{def:PSigma}, decomposes as
    \begin{equation}
        P_{\Sigma_1\sqcup\Sigma_2}(\phi)[\psi]=P_{\Sigma_1}(\phi_1)[\psi_1]P_{\Sigma_2}(\phi_2)[\psi_2]\;,
    \end{equation}
    for any closed $(d-1)$-manifolds $\Sigma_1$ and $\Sigma_2$, and any decomposable tensors $\phi=\phi_1\otimes\phi_2,\,\psi=\psi_1\otimes\psi_2\in\cD_{\Sigma_1\sqcup\Sigma_2}=\cD_{\Sigma_1}\otimes\cD_{\Sigma_2}$.
\end{lemma}
Lemma \ref{lem:factphys} allows us to characterise $P_{\Sigma_1\sqcup\Sigma_2}:\cD_{\Sigma_1}\otimes\cD_{\Sigma_2}\to\cD_{\Sigma_1\sqcup\Sigma_2}'$ as a tensor product of operators
\begin{equation}
    P_{\Sigma_1\sqcup\Sigma_2}=P_{\Sigma_1}\otimes P_{\Sigma_2}\;,
\end{equation}
from which it follows that $\ker P_{\Sigma_1\sqcup\Sigma_2}=\ker P_{\Sigma_1}\otimes \cD_{\Sigma_2}+\cD_{\Sigma_1}\otimes\ker P_{\Sigma_2}$.~This is crucial to prove the factorisation of the physical Hilbert space obtained from a canonical spin foam model.
\begin{proposition}[Factorisation]\label{prop:factphys}
    The physical Hilbert space resulting from a canonical $d$-dimensional spin foam model factorises
    \begin{equation*}
        \h^{\rm phys}_{\Sigma_1\sqcup\Sigma_2}\simeq\h^{\rm phys}_{\Sigma_1}\otimes \h^{\rm phys}_{\Sigma_2}\;,
    \end{equation*}
    for any closed $(d-1)$-manifolds $\Sigma_1,\Sigma_2$.
\end{proposition}
The proof of Proposition \ref{prop:factphys} is given in Appendix~\ref{App:Proofs3}.~Notice that this is something we expect to hold also in canonical Loop Quantum Gravity, where the gauge-invariant Hilbert space factorises into factors, one for each connected component of $\Sigma$, and the quantum Hamiltonian operator is a decomposable operator, providing that the full solution is the tensor product of solutions on each component.~In fact, the Thiemann's regularised Hamiltonian operator acts locally on the vertices of a spin-network state.~As a consequence, when a spin network consists of two disjoint graphs, the action of the Hamiltonian operator on one factor does not affect the other. In other words, the operator acts independently on each connected component of the spin network, reflecting the fact that its definition involves only local structures around the vertices \cite{Thiemann_1998a,Thiemann_1998b,Thiemann_1998c,Thiemann_2007}.

\subsection{A Proposal for the Gluing}\label{Sec:GluingProposal}
 
Let us now focus on the gluing property.~First of all, we notice that, owing to the requirement of the limit of the nets $\tilde{Z}_{\bullet}$ to belong to the algebraic dual $\cD_{\Sigma}'$, the objects $Z(M)$ in Eq.~\eqref{dLim} do not take values in an inner product space.~Therefore, it is not possible to introduce a notion of gluing as in TQFT (cfr.~property \emph{5} in Def.~\ref{def:TQFT} and Fig.~\ref{fig:tqftgluing}).~However, this is a feature, not a bug.~Indeed, we do not expect such a structure to hold in our setting.~A gluing property that associates the dual pairing of inner product spaces to the gluing of manifolds as in TQFT, is a very strong assumption; indeed, as discussed in Sec.~\ref{Sec:TQFT}, it ultimately forces the Hilbert space of the theory to be finite-dimensional.~In a non-topological quantum theory of gravity, the more we refine the bulk triangulations of the $d$-dimensional manifolds we glue together, the more we refine the triangulation of the $(d-1)$-dimensional boundary manifold along which we glue, and the elements of the net depend non trivially on the truncations of degrees of freedom captured by the discretisations.~In analogy with the convolution of the path integral propagators in ordinary quantum theories, where the composition amounts to summing over all of the intermediate kinematical degrees of freedom, we expect that, in the context of the spin foam nets \eqref{eq:ind-net}, the gluing $M\cup_{\Sigma}N$ of two manifolds $M$ and $N$ along a common boundary $\Sigma$ will amount to summing over the local geometric degrees of freedom encoded into the boundary states.~We therefore propose the following convolution-type gluing at the continuum level.
\begin{conjecture}[Gluing]\label{conj:gluing}
    For any d-dimensional manifolds $M,N$ such that $\partial M=\Sigma_1^*\sqcup \Sigma$ and $\partial N=\Sigma^*\sqcup \Sigma_2$; and, for any decomposable vector $\psi_1\otimes \psi_2\in\mathcal{D}_{\Sigma_1^*\sqcup\Sigma_2}$, we have
    \begin{equation}\label{eq:ZPZgluingconv}
        Z(M\cup_\Sigma N)[\psi_1\otimes\psi_2]=\sum_{\alpha,\beta\in\mathscr I}Z(M)[\psi_1\otimes \varphi_\alpha]P_{\Sigma}(\varphi_\alpha)[\varphi_\beta]Z(N)[R_\Sigma(\varphi_\beta)\otimes\psi_2]\;,
    \end{equation}
    where $\{\varphi_\alpha\}_{\alpha\in\mathscr I}$ is an orthonormal basis of $\mathcal{D}_{\Sigma}$.~Moreover, for every $(d-1)$-dimensional closed manifold $\Sigma$, the rigging map \eqref{def:PSigma} satisfies a convolution property
    \be\label{eq:PSigmaconvolution}
    P_{\Sigma}(\phi)[\psi]=\sum_{\alpha\in\mathscr I}P_{\Sigma}(\phi)[\varphi_\alpha]P_{\Sigma}(\varphi_\alpha)[\psi]\;,
    \ee
    for all $\phi,\psi\in\cD_{\Sigma}$,
\end{conjecture}
\begin{figure}[t!]
\noindent\makebox[\textwidth]{%
        \includegraphics[width=1.1\textwidth]{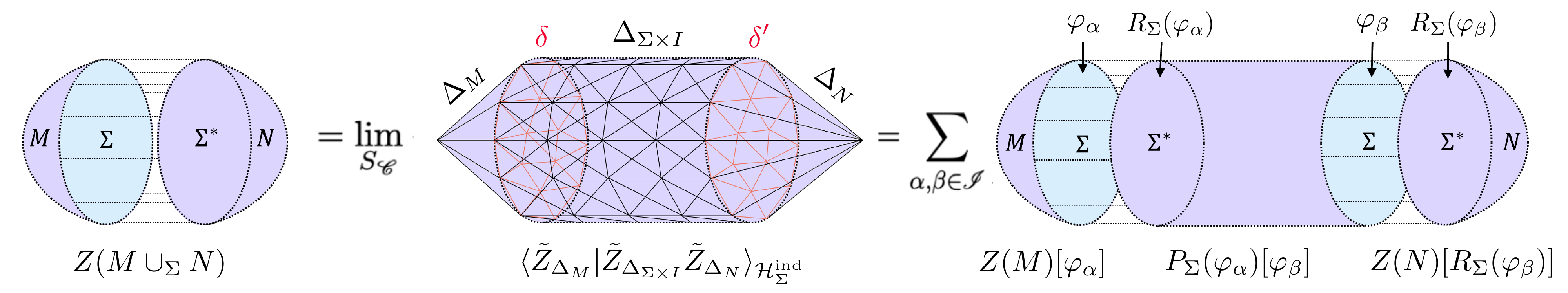}}
         \caption{Illustration of the property \eqref{eq:ZPZgluingconv} in the case $\Sigma_1=\varnothing=\Sigma_2$ and of the key technical step discussed in Appendix~\ref{AppSec:gluing} consisting of the restriction to the subnet defined over the directed set $S_\mathscr{C}$ of triangulations of $M\cup_\Sigma N$ that split in a suitable manner.~The second equality illustrates the propagator-like convolution gluing \eqref{eq:ZPZgluingconv}.}
         \label{Fig:ZPZcomp}
\end{figure}

A possible strategy to prove the above conjecture is discussed in Appendix~\ref{AppSec:gluing}.~This relies on two main technical steps schematically illustrated in Fig.~\ref{Fig:ZPZcomp}.~The first step consists of restricting to the subnet over the directed set of $d$-dimensional triangulations that share the same boundary triangulations along the gluing surface.~The existence of such a subnet is proved in Lemma~\ref{lem:gen_cutsubnet}.~The second step consists of interchanging the limit with the infinite sum over the index $\alpha$, which contains at most countably many non-zero elements as the $\tilde{Z}_{\Delta}$, entering the definitions \eqref{dLim} and \eqref{def:PSigma} of $Z$ and $P_{\Sigma}$, are Hilbert-Schmidt operators (see Appendix~\ref{AppSec:gluing} for details).~However, to the best of our knowledge, unlike sequences for which there exist known mathematical results such as the dominated convergence theorem, there is no established theorem or set of conditions guaranteeing that this interchange is valid in general for nets.\footnote{Note that, if the Hilbert spaces are finite-dimensional, the sum is finite and can be trivially interchanged with the limit.} Therefore, strictly speaking, the propagator-like convolution properties in Eqs.~\eqref{eq:ZPZgluingconv} and \eqref{eq:PSigmaconvolution} hold the status of a conjecture.

In a similar spirit to the dominated convergence theorem for sequences, where pointwise convergence by itself is not enough and the existence of an integrable (or summable) majorant is required to interchange sums and limits, the pointwise convergence of the nets $\tilde{Z}_{\bullet}$ in the algebraic dual is not enough so that we expect additional assumptions to be required to guarantee sufficient control over the spin foam amplitudes, which should converge in a suitably regular manner.~From this perspective, as more internal structures are created under refinement of the triangulations, this issue may be related to model-specific problems discussed in the literature, such as the control of bubble divergences.

As a consistency check for the convolution property \eqref{eq:ZPZgluingconv} (and similarly for \eqref{eq:PSigmaconvolution}), we notice that, when the discrete amplitudes are invariant under refinement of the (bulk) triangulations, the limits are trivial and Eq.~\eqref{eq:ZPZgluingconv} reduces to the gluing product of a TQFT given in Sec.~\ref{Sec:TQFT}.~This is precisely what happens in the case of the Ponzano-Regge model as discussed in Example~\ref{Ex:PRTVO2} below, which also illustrates the restriction to the subnet mentioned earlier.

\begin{example}\label{Ex:PRTVO2}
For the sake of brevity, we follow the original discussion by Ooguri in \cite{Ooguri:1991ib,Ooguri_1992a}, slightly adapted to match some notation already introduced in Sec.~\ref{Sec:examplePRTVO}.~For the explicit reformulation in terms of a rigging map between the dense subset of cylindrical spin network functionals and its algebraic dual, we refer to \cite{Noui_Perez_2005}.~Let $M \cup_{\Sigma} (\Sigma\times[0,1]) \cup_{\Sigma} N$ be a decomposition of a closed three-dimensional manifold $M\cup_{\Sigma}N$ into three parts, where $\Sigma$ is a closed surface and $\partial M = \partial N = \Sigma$.~Since the partition function of the Ponzano-Regge model is independent of the bulk triangulations, one may choose a triangulation such that $M$, $N$, and $\Sigma\times[0,1]$ do not share tetrahedra.~The partition function then factorises as
\begin{equation}\label{ex:subnetPRTVO}
Z_{M \cup_{\Sigma} N}=
\sum_{\alpha,\beta}
Z_{M;\delta}(\alpha)\,
Z_{\Sigma\times I;\delta,\delta'}(\alpha,\beta)\,
Z_{N;\delta'}(\beta),
\end{equation}
where $\alpha$ and $\beta$ denote boundary colorings.~Eq.~\eqref{ex:subnetPRTVO} is an explicit example of the restriction to the subnet over the directed set $S_\mathscr{C}$ of triangulations adapted to the decomposition of the underlying manifold (cfr.~Lemma~\ref{lem:gen_cutsubnet} in App.~\ref{AppSec:gluing} and Fig.~\ref{Fig:ZPZcomp}), where the limits are trivial owing to triangulation independence.

Since $Z_{\Sigma\times I;\delta,\delta'}$ is independent of the triangulation of the interior of $\Sigma\times I$, it satisfies the composition property
\begin{equation}
\sum_{\beta}
Z_{\Sigma\times I;\delta,\delta''}(\alpha,\beta)
Z_{\Sigma\times I;\delta'',\delta'}(\beta,\gamma)
=
Z_{\Sigma\times I;\delta,\delta'}(\alpha,\gamma).
\end{equation}
One can therefore define an operator $\Phi_{\mathrm{id}_{\Sigma,\delta}}$ acting on functions of the boundary colorings by
\begin{equation}
\Phi_{\mathrm{id}_{\Sigma,\delta}}[f](\alpha)
=
\sum_{\beta}
Z_{\Sigma\times I;\delta,\delta}(\alpha,\beta)\,f(\beta).
\end{equation}
This operator acts as a projection operator, i.e.~$\Phi_{\mathrm{id}_{\Sigma,\delta}}\circ \Phi_{\mathrm{id}_{\Sigma,\delta}}=\Phi_{\mathrm{id}_{\Sigma,\delta}}$.~Eq.~\eqref{ex:subnetPRTVO} can therefore be written as
\begin{equation}\label{ex:glcvPRTVO}
Z_{M \cup_{\Sigma} N}=
\sum_{\alpha,\beta}
\Phi_{\mathrm{id}_{\Sigma,\delta}}[Z_{M_1;\delta}](\alpha)\,
Z_{\Sigma\times I;\delta,\delta'}(\alpha,\beta)\,
\Phi_{\mathrm{id}_{\Sigma,\delta'}}[Z_{M_2;\delta'}](\beta)\;.
\end{equation}
Eq.~\eqref{ex:glcvPRTVO} is the analogue of \eqref{eq:ZPZgluingconv} for the Ponzano-Regge model.~More specifically, it shows that, in the Ponzano-Regge model, the states propagating from $M$ to $N$ via $\Sigma\times I$ are projected by the operator $\Phi_{\mathrm{id}_{\Sigma,\delta}}$.~One can therefore define the physical Hilbert space associated to $\Sigma$ as the subspace selected by the projector $\Phi_{\mathrm{id}_{\Sigma,\delta}}$, namely $f\in\Hil_{\delta} \Leftrightarrow f=\Phi_{\mathrm{id}_{\Sigma,\delta}}[f]$, with inner product given by
\begin{equation}
\braket{f_1|f_2}_{\Hil_{\delta}}=
\sum_{\alpha,\beta}
f_1(\alpha)\,
Z_{\Sigma\times I;\delta,\delta}(\alpha,\beta)\,
f_2(\beta)\;.
\end{equation}
As discussed in Sec.~\ref{Sec:Continuum1}, the boundary spaces $\Hil_{\delta}$ and $\Hil_{\delta'}$ for any $\delta$ and $\delta'$ are isomorphic through $Z_{\Sigma\times I;\delta,\delta'}$, so that the final Hilbert space $\Hil_{\Sigma}$ is identified with the range of the latter and the partition function of $M\cup_{\Sigma}N$ is given by the inner product on it, as it should since the Ponzano-Regge model is a TQFT.  
\end{example}

In light of the above discussion, the gluing rules \eqref{eq:ZPZgluingconv} and \eqref{eq:PSigmaconvolution} can be thought of as providing us with a weaker form of the gluing property in TQFT which generalises it to the case in which the Hilbert spaces are not necessarily finite dimensional and in which the convolution with the amplitude associated to the cylinder propagates the local dynamical degrees of freedom encoded into the boundary states.\footnote{The fact that a TQFT-type of gluing is recovered with triangulation independence is in principle also compatible with the possibility that the theory might approach a topological phase in certain regimes where quantum gravity could be studied as a perturbation of a TQFT.~This was suggested for example in \cite{Rovelli_2022} by studying simple reparametrisation invariant systems and, more recently, in the context of stacked spin foam models \cite{Han:2025emp,Han:2026hyq}.~A detailed analysis is however beyond the scope of the present work.}   

As we shall discuss in the next subsection, outside of the TQFT setting, the gluing Conjecture~\ref{conj:gluing} \emph{does} lead to a crucial difference compared to the case of a TQFT like the Ponzano-Regge model:~unlike the latter case, where the amplitudes $Z(M)$ are elements of the physical Hilbert space defined by the cylinder,\footnote{By the Riesz representation theorem, this is equivalent to $Z(M)$ yielding continuous functionals on physical states.} the propagator-like convolution \eqref{eq:ZPZgluingconv} implies that the continuum amplitude associated to a manifold $M$ different from the cylinder is not necessarily a continuous functional on physical states. 

\subsection{\texorpdfstring{$Z(M)$}{Z(M)} are Distributions on \texorpdfstring{$\Hphys$}{\mathcal H\^phys\_Sigma}}

The gluing properties \eqref{eq:ZPZgluingconv} and \eqref{eq:PSigmaconvolution} have an important consequence for the limit $Z(M)$ associated to a $d$-dimensional manifold $M$.~First of all, the following convolution property between $Z(M)$ and $P_{\del M}$ holds.
\begin{proposition}\label{prop:ZconvP}
    For any d-dimensional manifold $M$, and any state $\psi\in\mathcal D_{\del M}$, we have
 \begin{equation}
     Z(M)[\psi]=\sum_{\alpha\in\mathscr I}Z(M)[\varphi_\alpha]P_{\del M}(\varphi_\alpha)[\psi]\;,
 \end{equation}
 where $\{\varphi_\alpha\}_{\alpha\in\mathscr I}$ is an orthonormal basis of $\mathcal{D}_{\del M}$.
\end{proposition}
\begin{proof}
    The proof is straightforward.~Recalling in fact that $M\cup_{\partial M}(\partial M\times I)\simeq M$, the convolution property \eqref{eq:ZPZgluingconv} implies that
    \begin{align*}
        Z(M)[\psi]&=Z(M\cup_{\partial M}\partial M\times I)[1\otimes\psi]\\
        &=\sum_{\alpha,\beta\in\mathscr I}Z(M)[1\otimes \varphi_\alpha]P_{\partial M}(\varphi_\alpha)[\varphi_\beta]Z(\partial M\times I)[R_\Sigma(\varphi_\beta)\otimes\psi]\\
        &=\sum_{\alpha,\beta\in\mathscr I}Z(M)[\varphi_\alpha]P_{\partial M}(\varphi_\alpha)[\varphi_\beta]P_{\partial M}(\varphi_\beta)[\psi]\\
        &=\sum_{\alpha\in\mathscr I}Z(M)[\varphi_\alpha]P_{\partial M}(\varphi_\alpha)[\psi]\;,
    \end{align*}
    where, in the third equality, we used the definition \eqref{def:PSigma} of the rigging map and, in the last equality, we used the convolution property \eqref{eq:PSigmaconvolution}.
\end{proof}
\noindent
As an immediate consequence of Proposition~\ref{prop:ZconvP}, we have then the following corollary on the functional nature of the $Z(M)$.
\begin{corollary}\label{cor:ZMdistrib}
    For each $M$, with $\partial M=\Sigma$, $Z(M)$ defines a distribution on $\Hphys$.
\end{corollary}
\begin{proof}
    By Proposition~\ref{prop:ZconvP}, if $\psi\in\ker P_{\Sigma}$, then $Z(M)[\psi]=0$.~Hence, $Z(M)$ defines a linear functional on $\cD_{\Sigma}/\ker(P_{\Sigma})$.
\end{proof}
This characterises the limits $Z(M)$ as physical in a precise sense: they are algebraic linear functionals on $\Hphys$ whose common dense domain is the projection of $\mathcal{D}_{\Sigma}$.~This distributional nature of $Z(M)$ represents a key difference with the case of a TQFT like the Ponzano-Regge model for 3-dimensional quantum gravity.~Indeed, as anticipated at the end of the previous subsection, unlike the case of a TQFT where the $Z(M)$ are elements of the physical Hilbert space (see also Sec.~\ref{Sec:TQFT}) --- or, equivalently, continuous linear functionals on it --- in our present setting, the $Z(M)$ are not necessarily continuous functionals on physical states.~Borrowing intuition from ordinary constrained quantum systems, whose physical Hilbert space may admit a basis of improper eigenstates of physical observables, the functionals $Z(M)$ can be thought of as the improper eigenstates of some operators on $\Hil_{\del M}^{\rm phys}$.~This is in line with the second level of Gelfand triple to define physical observables as discussed in \cite{Thiemann_2007} (cfr.~Ch.~30), which now involves the physical Hilbert space itself, namely
\be
\mathcal D_{\Sigma}^{\rm phys}\subset\Hphys\subset(\mathcal D_{\Sigma}^{\rm phys})'\;,
\ee
with $\mathcal D_{\Sigma}^{\rm phys}=\cD_{\Sigma}/\ker(P_{\Sigma})$ and $(\mathcal D_{\Sigma}^{\rm phys})'$ its algebraic dual.~An explicit interpretation in terms of Dirac observables would of course require a throughout analysis of the observable algebras and their properties from the discrete theory to the continuum.~We leave this to future investigations and comment on this point in the outlook of Sec.~\ref{Sec:conclusion}.

The possibility of capturing the presence of local physical observables as a direct consequence of the gluing properties \eqref{eq:ZPZgluingconv} and \eqref{eq:PSigmaconvolution} further underscores the role of the latter as a distinguishing feature from topological theories.~In this respect, let us close this section by noticing that the distributional nature of the $Z(M)$ allows the property \eqref{eq:ZPZgluingconv} to descend on physical states as
\be\label{eq:convonphys}
Z(M\cup_\Sigma N)[\psi_1\otimes\psi_2]=\sum_{a,b\in\mathcal I}\mu_{ab}Z(M)[\psi_1\otimes e_a]Z(N)[\mathfrak{R}_\Sigma(e_b)\otimes\psi_2]\;,
\ee
where $\psi_1\otimes \psi_2\in\mathcal{H}_{\Sigma_1^*\sqcup\Sigma_2}^{\rm phys}$, $\mathfrak{R}_{\Sigma}$ denotes the Riesz map on $\Hphys$, $\{e_a\}_{a\in\mathcal I}$ is an orthonormal basis of $\Hphys$, and $\mu_{ab}$ are complex coefficients obtained by collecting together the components (in the basis $\{e_a\}$) of the system of generators of $\Hphys$ resulting from the basis elements of $\mathcal D_{\Sigma}$ in \eqref{eq:ZPZgluingconv} and of the intermediate cylinder amplitudes in \eqref{eq:ZPZgluingconv}.~Thus, the convolution-type gluing can be entirely expressed at the level of the physical Hilbert space.

\section{Conclusion and Outlook}\label{Sec:conclusion}

In this work we have proposed an axiomatic framework for spin foam models, within which the structure of the continuum limit can be analysed in a systematic way, independent of any specific model.

Inspired by Atiyah’s axioms for Topological Quantum Field Theories, we defined a spin foam model as a rule assigning algebraic data (Hilbert spaces and states) to combinatorial and topological data (triangulations), subject to a set of axioms (cfr.~Def.~\ref{def:truncSF}).~More precisely, to suitable $(d-1)$-dimensional abstract simplicial complexes $\delta$, which we referred to as \emph{boundary triangulations}, we associate a Hilbert space $\mathcal{H}_\delta$, while to $d$-dimensional abstract simplicial complexes $\Delta$, which we referred to as \emph{triangulations}, we associate a vector $Z(\Delta)\in\mathcal{H}_{\partial\Delta}$.~In contrast with Atiyah's framework, the Hilbert spaces appearing here are generally infinite-dimensional.~This feature originates from the presence of additional combinatorial data and from the consequent absence of a genuine identity morphism in the category of triangulations.

Starting from this axiomatic formulation, we introduced a notion of continuum limit by equipping the set of triangulations with suitable partial orders.~We first considered a particularly simple preorder, which we called the \emph{weak order}.~This order allows to take limits of bulk triangulations while keeping the boundary triangulation fixed.~Within this setting, we obtained a first general result according to which, if the amplitudes of a spin foam model are invariant under Pachner moves, then the model necessarily describes a Topological Quantum Field Theory (cfr.~Props.~\ref{prop:w-gluing},~\ref{prop:w-range}, and Cor.~\ref{cor:PachnerInv}).~However, this order does not faithfully capture the physical intuition of a refinement of triangulations encoding more and more degrees of freedom of the theory, and the requirement of convergence with respect to this order turns out to be too restrictive.

We therefore introduced the notion of \emph{refinement}, namely an order on the set $\mathcal{T}_M$ of triangulations of a manifold $M$ generated by stellar subdivisions and, more generally, Alexander moves. This construction turns $\mathcal{T}_M$ into a directed partially ordered set. Since this order does not preserve the boundary triangulation, we constructed a Hilbert space associated with a boundary topology $\Sigma$ by means of the direct limit
\[
\mathcal{H}^{\mathrm{ind}}_\Sigma \coloneqq 
\lim_{\longrightarrow}\mathcal{H}_\delta
=
\bigsqcup_{\delta\in \mathcal{T}_\Sigma}\mathcal{H}_\delta \Big/ \sim .
\]
For each $d$-dimensional manifold $M$, a spin foam model then defines a net 
$\tilde{Z}_\bullet : \mathcal{T}_M \to \mathcal{H}^{\mathrm{ind}}_{\partial M}$.~Assuming then that every such net converges to a limit point inside this Hilbert space, the resulting continuum theory is necessarily topological (Thm.~\ref{Thm:no-go}).~This yields a no-go theorem showing that strong notions of convergence inevitably lead to Topological Quantum Field Theories and therefore cannot describe a physical theory of quantum gravity.

Motivated by this observation, we considered a weaker notion of convergence inspired by Refined Algebraic Quantisation.~This can be implemented rigorously using the framework of generalised \emph{Gelfand triples}.~In this setting, we assumed that the nets $\tilde{Z}_\bullet$ converge as linear functionals in the algebraic dual $\mathcal{D}'_{\partial M}$ of a dense subspace $\mathcal{D}_{\partial M}\subset\mathcal{H}^{\mathrm{ind}}_{\partial M}$, namely
\[
Z(M)[\psi]
=
\lim_{\Delta\in\mathcal{T}_M}
\langle \tilde{Z}_\Delta \mid \psi \rangle_{\mathcal{H}^{\mathrm{ind}}_{\partial M}},
\qquad
\forall\,\psi\in\mathcal{D}_{\partial M}.
\]
Remarkably, by solely combining the assumption on the distributional convergence of the amplitudes together with the axioms in the discrete, we obtained the main result of the paper in Thm.~\ref{Thm:riggingmap}, according to which the functional associated with the cylinder $Z(\Sigma\times I)$ defines a rigging map
\[
P_\Sigma(\phi)[\psi]
\coloneqq
Z(\Sigma\times I)[R_\Sigma(\phi)\otimes\psi],
\qquad
\phi,\psi\in\mathcal{D}_\Sigma\;.
\]
This allowed us to construct the physical Hilbert space $\Hphys$ using the standard techniques of Refined Algebraic Quantisation.~The properties of the rigging map and of the resulting physical Hilbert space inherit structural features from the discretised formulation and resemble the factorisation and involution axioms appearing in Atiyah’s framework (cfr.~Props.~\ref{prop:physdual},~\ref{prop:factphys}).~Importantly, our axiomatic framework together with the distributional nature of the continuum limit prevents the definition of a TQFT-type of gluing.~Inspired by the propagator-like behaviour of the path integral, we therefore conjectured a weaker gluing property (Eq.~\eqref{eq:ZPZgluingconv}), which takes the form of a convolution and reduces to the TQFT-gluing with triangulation independence.~With this gluing convolution, the continuum amplitudes $Z(M)$ admit a precise physical interpretation as distributions (algebraic functionals) on the physical Hilbert space (Prop.~\ref{prop:ZconvP} and Cor.~\ref{cor:ZMdistrib}).

Taken together, these results indicate that the natural continuum description of spin foam models is not given by ordinary Hilbert-space limits, but rather by distributional amplitudes among which the cylinder amplitude implements the physical inner product through a rigging map.~The structural properties of the resulting continuum theory are summarised in the box below where, despite of its interpretation as the physical Hilbert space, we use the notation $\mathcal{H}_\Sigma$ rather than $\Hphys$ to keep the discussion independent of its explicit construction.
\smallskip

\begin{tcolorbox}[breakable,boxrule=1pt,rounded corners,colframe=gray!65,colback=white]
\textbf{Summary (Continuum spin foam).}~\emph{The continuum limit of a spin foam model assigns:
\begin{itemize}
    \item to any closed $(d-1)$-dimensional manifold $\Sigma$, up to orientation-preserving diffeomorphisms, a Hilbert space $\mathcal H_{\Sigma}$; and
    \item to the equivalence class of diffeomorphic $d$-dimensional manifolds $M$ a densely defined functional $Z(M)$ on the Hilbert space $\mathcal H_{\partial M}$.
\end{itemize}
These data satisfy the following properties:
\begin{enumerate}
    \item\textbf{Involution:}~$\mathcal H_{\Sigma^*}\simeq \mathcal H_{\Sigma}^*$, with $\Sigma^*$ the manifold with opposite orientation to $\Sigma$ and $\Hil_{\Sigma}^*$ the strong continuous dual space of $\Hil_{\Sigma}$.
    \item\textbf{Factorisation:}~$\mathcal H_{\Sigma_1\sqcup\Sigma_2}\simeq \mathcal H_{\Sigma_1}\otimes\mathcal H_{\Sigma_2}$.
    \item\textbf{Normalisation:}~If $\Sigma=\varnothing$, then $\mathcal H_{\varnothing}\simeq \mathbb{C}$.
    \item\textbf{Gluing:}~For any d-dimensional manifolds $M,N$, such that $\partial M=\Sigma_1^*\sqcup \Sigma$ and $\partial N=\Sigma^*\sqcup \Sigma_2$, and for any decomposable vector $\psi_1\otimes \psi_2\in\mathcal{H}_{\Sigma_1^*\sqcup\Sigma_2}$,
    \[Z(M\cup_\Sigma N)[\psi_1\otimes\psi_2]=\sum_{a,b\in\mathcal I}\mu_{ab}Z(M)[\psi_1\otimes e_a]Z(N)[\mathfrak{R}_\Sigma(e_b)\otimes\psi_2]\;,\]
    where $\{e_a\}_{a\in\mathcal I}$ is an orthonormal basis of $\mathcal{H}_{\Sigma}$, $\mathfrak{R}_\Sigma:\mathcal H_{\Sigma}\to\mathcal H_{\Sigma}^*$ is the Riesz map on $\mathcal H_{\Sigma}$, and $\mu_{ab}$ some complex weights.
    \item\textbf{Inner product:}~For any $\psi_1\in\mathcal H_\Sigma^*$ and $\psi_2\in\mathcal H_\Sigma$ such that $\psi_1\otimes\psi_2\in\mathrm{Dom}(Z(\Sigma\times I))$,
    \[
    Z(\Sigma\times I)[\psi_1\otimes\psi_2]
    =
    \langle \mathfrak{R}_\Sigma(\psi_1)\mid\psi_2\rangle_{\mathcal H_\Sigma}.
    \]
    \item\textbf{Normalisation:}~If $M=\varnothing$, then $Z(\varnothing)=1\in\mathbb{C}$.
\end{enumerate}
}
\end{tcolorbox}
\smallskip

Several directions for future work naturally emerge from our analysis.~First, it would be interesting to reformulate the framework developed in the present work from a $C^*$-algebraic perspective.~Such an approach could allow to define an algebra of physical observables and derive its properties directly from the truncated theory, potentially providing an independent derivation of the physical interpretation of the amplitudes $Z(M)$ without relying on the conjectured gluing property.~This perspective may also open the way to the study of the associated von Neumann algebras and to a precise notion of subsystems.

Another natural direction, not unrelated to the previous one, is the extension of the axiomatic framework to include manifolds with $(d-2)$-dimensional strata.~Similarly to the guiding role of Atiyah's framework in our analysis, once-extended Topological Quantum Field Theories, which require the language of higher categories, provides a useful starting point from a mathematical perspective.~From a physical perspective, such an extension would allow to incorporate corners in a systematic way, whose role in gravitational theories has received increasing attention in recent years (see e.g.~\cite{Donnelly_Freidel_2016,Speranza:2017gxd,Freidel_Geiller_Pranzetti_2020a,Freidel_Geiller_Pranzetti_2020b,Freidel_Geiller_Pranzetti_2021,Ciambelli:2021nmv,Carrozza:2022xut} in the context of classical gravity, and \cite{Freidel:2016bxd,Freidel:2018pbr,Freidel:2019ees,Freidel:2023bnj} for discussions in Loop Quantum Gravity and related contexts).

The constructions presented in this paper may also offer a first step towards a rigorous formulation of the sum over topologies.~Since amplitudes associated with manifolds sharing the same boundary naturally belong to the same vector space, it becomes possible to consider linear combinations of cobordisms in a mathematically well-defined way.~This could provide a functional-analytic interpretation for partition functions of the type appearing in Group Field Theory.

A generalisation along these directions would also enable a comparison of recent results on gravitational entropy across different quantum gravity approaches.~In the recent work~\cite{Colafranceschi_Dong_Marolf_Wang_2024}, holographic entropy is derived from a small set of axioms for the Euclidean gravitational path integral, which defines algebras of observables composed along $(d-2)$-boundaries.~Extending our framework as discussed above, together with an algebraic reformulation of the amplitudes and the appropriate continuum limit, would allow for a direct comparison with the axiomatic framework and algebraic constructions of~\cite{Colafranceschi_Dong_Marolf_Wang_2024}.~This could clarify whether these axioms, and the resulting holographic entropy, should be expected to hold in theories where the gravitational path integral arises as the continuum limit of discretised geometries.

Finally, on a purely mathematical side, it would be interesting to investigate whether the continuum structure emerging in this work admits a functorial formulation analogous to that of Topological Quantum Field Theories.~In particular, it would be interesting to investigate possible connections with the recent framework of Compositional Quantum Field Theories \cite{Oeckl:2022mvg} which refines the general boundary formulation proposed in \cite{Oeckl_2003a}.~More generally, one may ask whether spin foam amplitudes define a functorial quantum field theory once the continuum limit is taken in the distributional sense described above.

\section*{Acknowledgments}
\addcontentsline{toc}{section}{Acknowledgments}

We thank Eugenio Bianchi, Bianca Dittrich, Jonathan Engle, Domenico Fiorenza, Laurent Freidel, Muxin Han, Wojciech Kami\'nski, Don Marolf, Daniele Oriti, and Francesca Vidotto for helpful discussions.~We also thank the organisers of the International Loop Quantum Gravity Seminars for the opportunity to present an early version of this work in Spring 2025, as well as the participants whose questions and comments provided a valuable source of inspiration in the later stages of the work.~MB acknowledges support by grant ID\# 63132 from the John Templeton Foundation, managed by the Center for SpaceTime and the Quantum.~MB also acknowledges Western University for hospitality and the QISS cross-consortium visit program for financial support during two visits where part of this work was developed.~EC is currently supported by the ATRAE project PR28/23 ATR2023-145735 of the Agencia Estatal de Investigación (Spain).~FMM's work at Louisiana State University is supported by the NSF grants PHY-2409402, PHY-2110273.~FMM's research at Western University was supported by Francesca Vidotto’s Canada Research Chair in the Foundation of Physics, and NSERC Discovery Grant ``Loop Quantum Gravity: from Computation to Phenomenology''.~Western University is located in the traditional territories of Anishinaabek, Haudenosaunee, L\=unaap\'eewak and Chonnonton Nations.~This project/publication was also made possible through the support of the ID\# 62312 grant from the John Templeton Foundation, as part of the \href{https://www.templeton.org/grant/the-quantum-information-structure-of-spacetime-qiss-second-phase}{``The Quantum Information Structure of Spacetime'' Project (QISS)}, and through the support of the \href{https://withoutspacetime.org}{WOST, WithOut SpaceTime project}, led by the Center for Spacetime and the Quantum (CSTQ), and supported by Grant ID\# 63683 from the John Templeton Foundation.~The opinions expressed in this project/publication are those of the author(s) and do not necessarily reflect the views of the John Templeton Foundation.

\appendix

\section{Triangulations and Nets}
\label{App:Triang}
In this Appendix, we review some definition and well-known results about triangulations, nets, and their limits.~We shall mainly focus on those aspects which provide the technical background to the main body of the manuscript.~Specifically, in the first part of this appendix --- which follows the exposition in Willard's book~\cite{Willard_2004} and to which we refer for further details --- we focus on the generalities of the theory of nets and their limit.~In the second part, we focus on triangulations and the different order relations they can be equipped with.
\begin{definition}[Preorder]
    A preorder is a binary relation $\leq$ that is reflexive and transitive.
\end{definition}

\begin{definition}[Directed set]
    A directed set is a nonempty set $P$ together with a preorder $\leq$, and the property that every pair of elements has an upper bound, i.e. $\forall\,i,j\in P,\,\exists\,k\in P\ \mathrm{s.t.}\ i\leq k, j\leq k.$ 
\end{definition}

\begin{definition}[Poset]
    A partially ordered set (or poset) is a set $P$ equipped with a preorder $\leq$ that is antisymmetric. Namely, $i\leq j$ and $j\leq i$ implies $i=j$.
\end{definition}
Note that, in a directed set or even in a poset, not every pair of elements needs to be comparable.~When this is the case, the set is called \textit{total}.

On a directed set, it is possible to define a net (a generalisation of a sequence) as a map $x_{\bullet}:P\to X$, where $X$ is a topological space, or equivalently as a collection $\{x_i\}_{i\in P}$, and perform limits. 
\begin{definition}[Limit of a net]
    A point $x\in X$ is a limit point of the net $x_{\bullet}$ if, for every neighbourhood $U$ of $x$, there exists some $i\in P$ such that $x_j\in U$ for every $j\in P$ with $i\leq j$.
\end{definition}
If $X$ is Hausdorff, the limit point, if it exists, is \emph{unique} and we indicate it as $\lim_{i\in P}x_i$.\\
A net $x_i$ in a metric space $X$ with distance $d$ is said to be a \emph{Cauchy net} if, for every $\varepsilon>0$, there exists $k\in P$ such that, for all $i,j\geq k$, one has $d(x_i,x_j)\leq\varepsilon$.~A metric space is complete if and only if every Cauchy net is convergent.

An important notion often used in proving results is that of a subnet of a net.
 \begin{definition}[Subnet]
 \label{def:subnet}
     Let $P$ and $I$ be two directed sets and let $x_{\bullet}:P\to X$ and $s_{\bullet}:I\to X$ be two nets over them, $s_{\bullet}$ is a \emph{subnet} of $x_{\bullet}$ if there exists an $h:I\to P$ such that $s_i=x_{h(i)},\,\forall i\in I$, with the following properties:
     \begin{enumerate}
         \item\emph{order-preserving:} if $i\leq j$, then $h(i)\leq h(j)$;
         \item\emph{cofinal:} for every $a\in P$, there exists some $b\in h(I)$ such that $a\leq b$.
     \end{enumerate}
 \end{definition}

 \noindent
 A key property relating the limits of the subnets of a given net to its limit is the following necessary and sufficient condition:~a net $x_{\bullet}$ converges to $x$ \emph{if and only if} every subnet converges to $x$ (see \cite{Willard_2004}, Thm.~11.5).~This property allows us to prove the following useful result.
\begin{lemma}
\label{lem:contProd}
    Let $a_{\bullet}:A\to X$ and $b_{\bullet}:B\to Y$ be two nets, with $A,B$ directed posets and $X,Y$ Hausdorff topological spaces, such that $a=\lim_A a_{\bullet}$ in $X$ and $b=\lim_B b_{\bullet}$ in $Y$.~Let $f:X\times Y\to Z$ be a continuous function from $X\times Y$ to a Hausdorff topological space $Z$.~Then, $f(a,b)=\lim_{A\times B}f(a_{\bullet},b_{\bullet})$.
\end{lemma}
\begin{proof}
    $A\times B$ with the product order is a directed poset. The product order is $(i,j)\geq (i',j')$ if and only if $i\geq i'$ and $j\geq j'$. We are going to consider the net $(a_i,b_j):A\times B\to X\times Y$. Notice that the projections $\mathrm{pr}_A:A\times B\to A$ and $\mathrm{pr}_B:A\times B\to B$ are order-preserving and cofinal. Thus,
    \begin{align*}
        &a=\lim_{A}a_{i}=\lim_{A\times B}a_{\mathrm{pr}_A (i,j)}=\lim_{A\times B}\pi_X(a_{i},b_j),\\
        &b=\lim_{B}b_{i}=\lim_{A\times B}b_{\mathrm{pr}_B (i,j)}=\lim_{A\times B}\pi_Y(a_{i},b_j).
    \end{align*}
    Where $\pi_{X/Y}$ are the projectors on the two factors of $X\times Y$. Since a net in a product space converges if and only if each of its projections converges, and the limits of the projections are given above, we obtain \[\lim_{A\times B}(a_i,b_j)=(a,b).\]
    By continuity of $f$, the thesis follows.
\end{proof}

The directed set of interest in this work is that of triangulations with a suitable order.~In the remainder of this appendix, let us then recall their definition, together with some useful properties.~This will further clarify the notation used in the main text.~We refer the reader to Hudson's book~\cite{Hudson_1969} for a detailed exposition.~Roughly speaking, a  triangulation is a regularly-enough simplicial complex.~To make this more precise, let us first introduce the notion of abstract simplicial complex, followed by that of two important subcollections of it, namely stars and links.
\begin{definition}[Abstract simplicial complex]
An abstract simplicial complex $\Delta$ is a collection of non-empty finite subsets of some set such that for every set $X$ in $\Delta$, and every non-empty subset $Y\subset X$, the set $Y$ also belongs to $\Delta$.~An element of $\Delta$ is called a face (or simplex) and its dimension is its cardinality minus 1.~The dimension of the complex is the largest dimension of any of its faces.~A vertex is an element of the union of all the faces of $\Delta$.
\end{definition}
\begin{definition}[Star and link]
The star of a face is a set containing every face of $\Delta$ such that the union is a face in $\Delta$.~The link of a vertex is a set containing every face of $\Delta$ that does not contain the vertex but whose union is a face of $\Delta$.    
\end{definition}
\noindent
Given these notions, a triangulation is defined as follows.
\begin{definition}[Triangulation]
    A $d$-dimensional triangulation is an abstract simplicial complex with dimension $d$ such that the link of each vertex is a $(d-1)$ sphere.
\end{definition}
An abstract simplicial complex $\Delta$ carries a natural (and unique) topology so that its geometric realisation $|\Delta|$ is a topological space.~As per the above definition, the geometric realisation of a triangulation $|\Delta|$ is a $d$-dimensional topological manifold and admits a piecewise linear structure (PL-structure).~Notice that every smooth manifold admits a triangulation, and so a (unique, once the triangulation is fixed) PL-structure~\cite{Whitehead_1940}.~The converse is not true in general but it turns out to be true in $d\leq 4$ dimensions, that is our case of interest.~Thus, even though, formally, we should speak in terms of a piecewise linear category, we can consider all \textit{smoothable}.

Of particular interest in the present work will be the set of the triangulations of a given PL-manifold $M$ which is defined as follows.
\begin{definition}[Triangulations of a PL-manifold]
    $\mathcal{T}_M=\left\{\text{triangulation}\ \Delta \ \text{s.t.}\ |\Delta|\simeq M\right\}$ is the set of all triangulations PL-homeomorphic to a compact PL-manifold $M$, 
\end{definition}
When $|\Delta|$ is orientable, we can consider $M$ an oriented manifold and $\Delta$ carrying an according orientation. We indicate with $\Delta^*$ the triangulation with inverse orientation.~As for the manifold, we would like to define the gluing of two triangulations $\Delta_1,\Delta_2$ by the identification face by face along a common boundary triangulation $\delta$.~More precisely, let $\Delta_1$ and $\Delta_2$ be such that $\partial\Delta_1=\delta_1\sqcup\delta$ and $\del\Delta_2=\delta^*\sqcup \delta_2$.~In analogy with the case of manifolds, we consider a orientation-inverting simplicial isomorphism $f:\delta\to\delta^*$ and the gluing given by $\Delta_1\cup_{\delta}\Delta_2=\Delta_1\sqcup\Delta_2/\sim_f$.

Triangulations can be manipulated without altering the PL-homeomorphism class via the so-called \textit{Pachner moves}~\cite{Pachner_1991} (for a review see~\cite{Lickorish_1999}).~If $M$ has a boundary $\Sigma=\partial M$, we can properly talk about a boundary triangulation $\partial\Delta$, which is a triangulation PL-homeomorphic to $\Sigma$.~The Pachner moves maintain the boundary triangulation fixed, hence it makes sense to define the set of all triangulations PL-homeomorphic to a manifold $M$ with given boundary triangulation $\delta$.~We shall denote it by $\mathcal{T}_M^{\delta}$. 
It is not totally trivial that such sets are non-empty for every triangulation of the boundary, however, this turns out ot be true:
\begin{lemma}
\label{lem:nonemptyT}
For any triangulation $\delta$ of $\partial M$, the set 
$\mathcal{T}_M^\delta$ is non-empty.
\end{lemma}
\begin{proof}
We start by noticing that, given a triangulation $\Delta_0$ of $M$, this induces a triangulation $\delta_0 \coloneqq  \partial \Delta_0$ of $\partial M$ with non-empty $\mathcal T_M^{\delta_0}$ by construction. The non-trivial step is to show that this is the case for any other boundary triangulation $\delta$. To this aim, given any triangulation $\delta$ of $\partial M$, let us consider the cylinder $\partial M \times I$ and choose a triangulation $\Delta_C$ of it such that 
$\partial \Delta_C = \delta_0^* \sqcup \delta$; this is always possible since any two triangulations of the same manifold are related by a sequence of Pachner moves~\cite{Hellmann_2011,Barrett_Westbury_1996}. Gluing $\Delta_0$ and $\Delta_C$ along their common boundary triangulation $\delta_0$ yields a triangulation 
\[
\Delta \coloneqq  \Delta_0 \cup_{\delta_0} \Delta_C
\]
of the manifold
\[
M' \coloneqq  M \cup_{\partial M} (\partial M \times I),
\]
whose boundary is $\delta$. Since $M' \simeq M$, we may regard $\Delta$ as a triangulation of $M$ with boundary $\delta$, hence $\Delta \in \mathcal{T}_M^\delta$. Thus, $\mathcal{T}_M^\delta\neq\varnothing$ for any triangulation $\delta$ of $\partial M$.
\end{proof}
\noindent
For given $M$ and $\delta$, the set $\mathcal{T}_M^{\delta}$ can be equipped with a preorder which we refer to as \emph{weak order}.
\begin{definition}[Weak order]
    Let $\Delta\in \mathcal{T}_M^{\delta}$ be a triangulation of a $d$-dimensional PL-manifold $M$, with boundary triangulation $\delta$, and denote by $N_d(\Delta)$ the number of $d$-simplices in $\Delta$.~The preorder $\leq$ on $\mathcal{T}_M^{\delta}$ is defined as $\Delta_1\leq\Delta_2$ if and only if $N_d(\Delta_1)\leq N_d(\Delta)$ and is called \emph{weak order}.
\end{definition}
Since every two triangulations in $\mathcal{T}_M^{\delta}$ can be related to one another by a sequence of Pachner moves~\cite{Pachner_1991,Lickorish_1999}, it follows that:
\begin{proposition}
\label{prop:w-ord}
    The couple $(\mathcal{T}_M^{\delta},\leq)$ is a directed set.
\end{proposition}
\begin{proof}
    By the previous Lemma, the sets $\mathcal{T}_M^{\delta}$ are non-empty for any triangulation $\delta$ with $|\delta|\simeq \del M$.~The binary relation $\leq$ on $\mathcal{T}_M^{\delta}$ is reflexive, i.e. $\Delta\leq \Delta$, and transitive, i.e.~$\Delta_1\leq \Delta_2$ and $\Delta_2\leq \Delta_3$, then $\Delta_1\leq \Delta_3$.~Hence, $\mathcal{T}_M^{\delta}$ is a preordered set.~It is also total in that, given $\Delta_1$ and $\Delta_2$, we can always compare $N_d(\Delta_1)$ and $N_d(\Delta_2)$ because they are natural numbers. Since it is total, $\mathcal{T}_M^{\delta}$ is a directed set.~Explicitly, given $\Delta_1, \Delta_2$, we can always compare them, so let's suppose $\Delta_1\leq \Delta_2$, an upper bound $\Delta$ is given by a single Pachner move $1\to d+1$ on $\Delta_2$.
\end{proof}
A central role in this work is played by \emph{subdivisions}.~Subdivisions provide a notion of refinement that is closer to our physical intuition, as they allow to increase the resolution of a triangulation (including the boundary triangulation) while preserving its underlying geometric and topological structure.~Through subdivisions, it is possible to equip the set $\mathcal{T}_M$ of triangulations of $M$ with a more suitable partial order, which has been extensively studied in the mathematical literature.~This refined ordering better captures the idea of progressively resolving finer degrees of freedom and will be essential for the analysis that follows.
\begin{definition}[Subdivision]
\label{def:subdiv}
    A triangulation $\Delta$ is a subdivision of $\Delta'$ if $|\Delta|=|\Delta'|$ and every simplex of $\Delta$ is a subset of some simplex of $\Delta'$ or, equivalently, every simplex of $\Delta'$ is a union of simplices of $\Delta$.
\end{definition}
In particular, in this work, we are taking in consideration the stellar subdivisions~\cite{Hudson_1969}, already present in the literature of Loop Quantum Gravity, and known as refining Alexander moves~\cite{Dittrich_Geiller_2015,Bahr_Dittrich_Geiller_2021}.
\begin{definition}[Refinement]
    Let $\Delta_1,\Delta_2\in\mathcal{T}_M$, $\Delta_1$ is a \textit{refinement} of $\Delta_2$ if $\Delta_1$ is a stellar subdivision of $\Delta_2$ (namely, if it can be obtained via a sequence of refining Alexander moves).
\end{definition}

\begin{figure}[t!]
    \centering
    \includegraphics[width=0.9\linewidth]{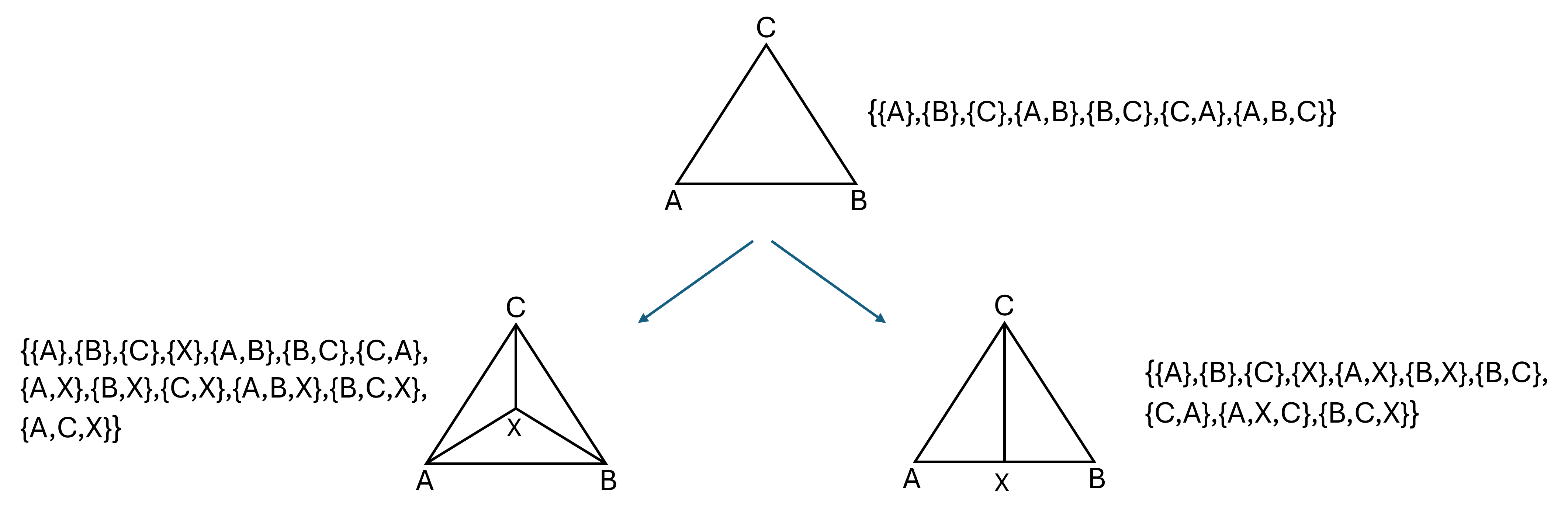}
    \caption{Two possible subdivisions by starring of a 2-simplex for different positions of an arbitrary point $\textsf X$.~These two moves are the refining Alexander moves for the 2-simplex.~The abstract form of the simplices is reported together with their geometric realisation.}
    \label{fig:alex}
\end{figure}
\noindent
A simple example of stellar subdivision is illustrated in Fig.~\ref{fig:alex}.~Notice that Alexander moves may alter the boundary triangulation~\cite{Alexander_1926,Alexander_1930}.

The notion of refinement induces an order on $\mathcal{T}_M$, which we refer to as \emph{strong order}, and equips it with the structure of a directed poset.
\begin{definition}[Strong order]
\label{refin}
    The strong order (or refinement) $\preceq$ on $\mathcal{T}_M$ is $\Delta_1\preceq\Delta_2$ if and only if $\Delta_2$ is a refinement of $\Delta_1$.
\end{definition}
\begin{proposition}\label{refin:poset}
    The couple $(\mathcal{T}_M,\preceq)$ is a directed partially ordered set.
\end{proposition}
\begin{proof}
    The order $\preceq$ is clearly reflexive and transitive. Moreover, it is antisymmetric, indeed if $\Delta_1\preceq\Delta_2$ and $\Delta_2\preceq\Delta_1$, then $\Delta_1=\Delta_2$, for all $\Delta_1,\Delta_2\in\mathcal{T}_M$. The existence of an upper bound follows from the former Alexander conjecture (originally proposed by H.~Tietze in~\cite{Tietze_1908}): if two complexes are piecewise linearly homeomorphic then they admit isomorphic stellar subdivisions. This conjecture was recently proved in~\cite{Adiprasito_Pak_2024}
\end{proof}

\noindent
We close this appendix by enouncing two properties, which are a direct consequence of Lemma 1.3 in~\cite{Hudson_1969} and will turn useful at various points in proving some of the results in the main text.~The first property characterises the refinement of triangulations glued along a common boundary, while the second property relates boundary and bulk refinements.
\begin{lemma}
\label{lem:guess}
    Let $\Delta\in\mathcal{T}_M$, then
    \begin{enumerate}
        \item for $M=M_1\cup_{\Sigma}M_2$ and $\Delta$ given by the union $\Delta_1\cup_{\delta}\Delta_2$ of two triangulations $\Delta_1\in\mathcal{T}_{M_1},\,\Delta_2\in\mathcal{T}_{M_2},$ along a common triangulated boundary $\delta\in\mathcal{T}_{\Sigma}$, each refinement $\Delta'$ of $\Delta$ is of the form $\Delta_1'\cup_{\delta'}\Delta'_2$, for some $\Delta_1'\succeq\Delta_1\in\mathcal{T}_{M_1},\,\Delta_2'\succeq\Delta_2\in\mathcal{T}_{M_2},$ and $\delta'\succeq\delta\in\mathcal{T}_{\Sigma}$.
        \item each refinement $\delta'\in\mathcal{T}_{\del M}$ of the boundary triangulation $\delta=\del\Delta\in\mathcal{T}_{\del M}$ is the boundary of a refined triangulation $\Delta'\succeq\Delta\in\mathcal{T}_M$.
    \end{enumerate}
\end{lemma}

\section{Proofs of Sec.~\ref{Sec:Continuum1}}\label{App:Proofs1}
In this appendix we collect the proofs of the theorems and lemmas of Sec.~\ref{Sec:Continuum1}, reported in the same order as they appear in the main text.

\noindent\textbf{Lemma~\ref{lem:subnet}.}\emph{
    Given $ M\simeq N_1\cup_{\Sigma}N_2$ with $\partial N_1=\Sigma_1^*\sqcup\Sigma$, $N_2=\Sigma^*\sqcup\Sigma_2$, and $\partial  M=\Sigma_1^*\sqcup \Sigma_2$. The subset of $\mathcal{T}_{M}^{\delta_1^* \sqcup \delta_2}$ composed by triangulations that can be seen as the gluing of two triangulations of $N_1$ and $N_2$ along a triangulation $\delta$ of the boundary component $\Sigma$ is the image of the map
    \[
    \begin{matrix}
        \gamma: & \mathcal{T}_{N_1}^{\delta_1^* \sqcup \delta}\times \mathcal{T}_{N_2}^{\delta^* \sqcup \delta_2} & \to & \mathcal{T}_{M}^{\delta_1^* \sqcup \delta_2};\\
        & (\Delta_1,\Delta_2) & \mapsto & \Delta_1\cup_\delta \Delta_2.
    \end{matrix}
    \]
    The map $\gamma$ is order-preserving and cofinal. As a consequence, $Z_{\gamma(\bullet)}$ is a subnet of $Z_{\bullet}$.
}
\begin{proof}
    The set $\mathcal{T}_{N_1}^{\delta_1^* \sqcup \delta}\times \mathcal{T}_{N_2}^{\delta^* \sqcup \delta_2}$ is a directed product set, and it is naturally equipped with a binary relation defined as follows: $(\Delta_1,\Delta_2) \leq (\Delta_1', \Delta_2')$ if $\Delta_1 \leq \Delta_1'$ and $\Delta_2 \leq \Delta_2'$.~The gluing map $\gamma:\mathcal{T}_{N_1}^{\delta_1^* \sqcup \delta}\times \mathcal{T}_{N_2}^{\delta^* \sqcup \delta_2}\to \mathcal{T}_M^{\delta_1^*\sqcup\delta_2}$ satisfies the properties of Definition~\ref{def:subnet}.~In fact, the property \textit{1.} can be verified straight forwardly:~if $(\Delta_1,\Delta_2) \leq (\Delta_1', \Delta_2')$ then $N_d(\Delta_1) \leq N_d(\Delta_1')$ and $N_d(\Delta_2) \leq N_d(\Delta_2')$.~Since $N_d(\Delta_1\cup_\delta\Delta_2)=N_d(\Delta_1)+N_d(\Delta_2)$, we obtain $N_d(\Delta_1\cup_\delta\Delta_2)\leq N_d(\Delta_1'\cup_\delta\Delta_2')$.~Thus, the map $\gamma$ is order-preserving. The property \textit{2.} can be proved as follows.~Let $\Delta \in  \mathcal{T}_M^{\delta_1^*\sqcup\delta_2}$ and suppose that $N_d(\Delta)=N$.~Since the image of $\gamma$ is non-empty, there exists a triangulation $\Delta''=\Delta_1\cup_\delta\Delta_2\in {\rm Ran}(\gamma)\subset  \mathcal{T}_M^{\delta_1^*\sqcup\delta_2}$, whose $N_d(\Delta'')=N''$.~If $N''\geq N$, then $\Delta'' \geq \Delta$.~If $N''< N$, from $\Delta''$, via a sequence $k$ of Pachner moves $1\to d+1$ on $\Delta_1$, such that $kd>N-N''$, we obtain a $\Delta'\in {\rm Ran}(\gamma)$ such that $\Delta \leq \Delta'$. Hence, ${\rm Ran}(\gamma)$ is cofinal in $ \mathcal{T}_M^{\delta_1^*\sqcup\delta_2}$.
\end{proof}
\noindent
\textbf{Proposition~\ref{prop:w-gluing}.}\emph{
    Let $M=N_1\cup_\Sigma N_2$, with $N_1=\Sigma_1^*\sqcup\Sigma$ and $N_1=\Sigma^*\sqcup\Sigma_2$. Given $\delta,\delta_1,\delta_2$ triangulations of $\Sigma,\Sigma_1,\Sigma_2$, respectively; then, the following holds:
    \begin{equation*}
        Z_{N_1\cup_\Sigma N_2; \delta_1^* \sqcup \delta_2}= \langle Z_{N_1; \delta_1^* \sqcup \delta}, Z_{N_2; \delta^* \sqcup \delta_2} \rangle_\delta\,.
    \end{equation*}
}
\begin{proof}
    We want to show that, for any triangulation $\delta$ of $\Sigma$,
\begin{equation}
\label{eq:Agluing}
    Z_{N_1\cup_\Sigma N_2; \delta_1^* \sqcup \delta_2}= \langle Z_{N_1; \delta_1^* \sqcup \delta}, Z_{N_2; \delta^* \sqcup \delta_2} \rangle_\delta\;.
\end{equation}
By definition,
\begin{equation}
    \begin{split}
        Z_{N_1\cup_\Sigma N_2; \delta_1^* \sqcup \delta_2} \coloneqq & \lim_{\Delta \in \mathcal{T}_{N_1\cup_\Sigma N_2}^{\delta_1^* \sqcup \delta_2}} Z_{\Delta} \\ =& \lim_{(\Delta_1,\Delta_2) \in \mathcal{T}_{N_1}^{\delta_1^* \sqcup \delta}\times \mathcal{T}_{N_2}^{\delta^* \sqcup \delta_2}} Z_{\Delta_1\cup_\delta\Delta_2} \\ =& \lim_{(\Delta_1,\Delta_2) \in \mathcal{T}_{N_1}^{\delta_1^* \sqcup \delta}\times \mathcal{T}_{N_2}^{\delta^* \sqcup \delta_2}} \langle Z_{\Delta_1},Z_{\Delta_2}\rangle_\delta\;,
    \end{split}
\end{equation}
where the second equality uses Proposition~\ref{lem:subnet}, and the third follows from the gluing axiom $Z_{\Delta_1\cup_\delta\Delta_2}=\langle Z_{\Delta_1},Z_{\Delta_2}\rangle_\delta$ (cfr.~property \emph{4.} in Def.~\ref{def:truncSF}). Thus, to establish~\eqref{eq:Agluing} it suffices to prove that 
\begin{equation}
    \lim_{(\Delta_1,\Delta_2) \in \mathcal{T}_{N_1}^{\delta_1^* \sqcup \delta}\times \mathcal{T}_{N_2}^{\delta^* \sqcup \delta_2}} \langle Z_{\Delta_1},Z_{\Delta_2}\rangle_\delta= \langle Z_{N_1; \delta_1^* \sqcup \delta}, Z_{N_2; \delta^* \sqcup \delta_2} \rangle_\delta\;.
\end{equation}
This follows from Lemma~\ref{lem:contProd} which states that, for a continuous map of two nets, the action of the map commutes with taking the limit along the product of their directed sets; in our case, the continuous map is the pairing $\langle\cdot,\cdot\rangle_\delta$, the two nets $Z_{\bullet}$ are indexed by $\mathcal{T}_{N_1}^{\delta_1^* \sqcup \delta}$ and $\mathcal{T}_{N_2}^{\delta^* \sqcup \delta_2}$, and the limit is taken along the product of these directed sets, namely $\mathcal{T}_{N_1}^{\delta_1^* \sqcup \delta}\times \mathcal{T}_{N_2}^{\delta^* \sqcup \delta_2}$.
\end{proof}
\begin{remark}\label{rmk:loc-ref}
    The previous result can be obtained also performing the limit on just one side of the splitting. Define
\[
s^2_{\bullet}\coloneqq \braket{Z_{\Delta_1},Z_{\bullet}}_\delta:
\mathcal{T}_{N_2}^{\delta^* \sqcup \delta_2}\to \mathcal{H}_{\delta_1^*\sqcup\delta_2},
\]
for a fixed $\Delta_1\in\mathcal{T}_{N_1}^{\delta_1^*\sqcup\delta}$, and consider the map
\begin{equation}
    \begin{matrix}
         h_{\Delta_1}:&\mathcal{T}_{N_2}^{\delta^* \sqcup \delta_2}&\to&\mathcal{T}_{N_1}^{\delta_1^* \sqcup \delta}\times \mathcal{T}_{N_2}^{\delta^* \sqcup \delta_2}\,;\\
    &\Delta &\mapsto &(\Delta_1,\Delta)\,.
    \end{matrix}
\end{equation}
Clearly, $s^2_{\bullet}=Z_{\gamma\circ h_{\Delta_1}(\bullet)}$, and so $s^2_{\bullet}$ is a subnet of $Z_{\bullet}$. Hence, by existence and uniqueness of the limit,
\begin{equation}
    \lim_{\Delta_2\in \mathcal{T}_{N_2}^{\delta^* \sqcup \delta_2}} \langle Z_{\Delta_1},Z_{\Delta_2}\rangle_\delta=Z_{N_1\cup_\Sigma N_2; \delta_1^* \sqcup \delta_2}\,.
\end{equation}
Here, the result of the limit does not depend on $\Delta_1$. Moreover, changing the limit, the result does not change. Indeed, we can define \[s^1_{\bullet}\coloneqq \braket{Z_{\bullet},Z_{\Delta_2}}_\delta:\mathcal{T}_{N_1}^{\delta_1^* \sqcup \delta}\to \mathcal{H}_{\delta_1^*\sqcup\delta_2}\] and, given $\Delta_2\in\mathcal{T}_{N_2}^{\delta^* \sqcup \delta_2}$, the map 
\begin{equation}
    \begin{matrix}
        h_{\Delta_2}:&\mathcal{T}_{N_1}^{\delta_1^* \sqcup \delta}&\to&\mathcal{T}_{N_1}^{\delta_1^* \sqcup \delta}\times \mathcal{T}_{N_2}^{\delta^* \sqcup \delta_2}\,;\\
    &\Delta &\mapsto &(\Delta,\Delta_2)\,.
    \end{matrix}
\end{equation}
Again, $s^1_{\bullet}=Z_{\gamma \circ h_{\Delta_2}(\bullet)}$, and so $s^1_{\bullet}$ is a subnet of $Z_{\bullet}$. Hence,
\begin{equation}
    \lim_{\Delta_1\in \mathcal{T}_{N_1}^{\delta_1^* \sqcup \delta}} \langle Z_{\Delta_1},Z_{\Delta_2}\rangle_\delta=Z_{N_1\cup_\Sigma N_2; \delta_1^* \sqcup \delta_2}
\end{equation}
As before, the result of the limit does not depend on $\Delta_2$, so taking the second limit does not change the result.
\end{remark}
\noindent
\textbf{Proposition~\ref{prop:w-range}} \emph{The ranges of the operators $Z_{\Sigma\times I;\delta^* \sqcup \delta}:\mathcal{H}_{\delta}\to\mathcal{H}_{\delta}$ for different triangulations $\delta$ of $\Sigma$ are all isomorphic.~That is,  
   \begin{equation}
   {\rm Ran}(Z_{\Sigma\times I;\delta^* \sqcup \delta})\simeq {\rm Ran}(Z_{\Sigma\times I;\delta'^* \sqcup \delta'})\;,
   \end{equation}
   for every pair $\delta,\delta'$ of triangulations of $\Sigma$.}
\begin{proof}
    Let $V\coloneqq \mathcal{H}_{\delta}$,~$W\coloneqq \mathcal{H}_{\delta'}$,~$P\coloneqq Z_{\Sigma\times I;\delta^* \sqcup \delta}:V\to V$,~and $P'\coloneqq Z_{\Sigma\times I;\delta'^* \sqcup \delta'}:W\to W$. Notice that $P$ and $P'$ are projections, i.e. $P\circ P=P$, because of Corollary~\ref{cor:w-proj}. \\
    Define $F\coloneqq Z_{\Sigma\times I;\delta^* \sqcup \delta'}:V\to W$ and $H\coloneqq Z_{\Sigma\times I;\delta'^* \sqcup \delta}:W\to V$.~Then, by Prop.~\ref{prop:w-gluing}, the following diagram
    \begin{equation}
        \begin{tikzcd}
    V \arrow[r, "P"] \arrow[d, "F" left] & V \arrow[d, "F"]\\
    W \arrow[ur, "H"] \arrow[r,"P'"] & W
\end{tikzcd}
    \end{equation}
    commutes.~This allows us to prove that ${\rm Ran}(P)\subset V$ and ${\rm Ran}(P')\subset W$ are isomorphic via $F$ as follows.
    \begin{itemize}
        \item Injectivity: Let $v_1,v_2 \in {\rm Ran}(P)$ be such that $F(v_1)=F(v_2)$.~Then, we have
        \[H\circ F(v_1)=H\circ F(v_2)\quad\implies\quad Pv_1=Pv_2\quad\implies\quad v_1=v_2\,.\]
        \item Surjectivity: Let $w\in {\rm Ran}(P')$.~Then, there always exists $v\in {\rm Ran}(P)$ such that $F(v)=w$.~This is defined by $v=H(w)$.~Indeed, we have \[F(v)=F\circ H(w)=P'w=w\,.\]
    \end{itemize}
\end{proof}

\section{Proofs of Sec.~\ref{Sec:Continuum2}}\label{App:Proofs2}

In this appendix we collect the proofs of the theorems and lemmas of Sec.~\ref{Sec:Continuum2}, reported in the same order as they appear (or are referred to) in the main text.

\begin{lemma}\label{rm:lim} The inductive limit $\h_{\Sigma}^{\rm ind}=\lim_{\delta\in\mathcal{T}_{\Sigma}}\h_{\delta}$ is a Hilbert space.
\end{lemma}
\begin{proof}
    Recalling the notation $[\psi_\delta,\delta]$ for the equivalence class of $\psi_\delta\in\h_\delta$ (cfr.~Def.~\ref{def:inductive}), we define the linear combinations as
    \[\alpha [\psi_\delta,\delta]+\beta[\phi_{\delta'},\delta']=[\alpha \iota_{\delta''\delta}(\psi_{\delta})+\beta \iota_{\delta''\delta'}(\phi_{\delta'}),\delta''],\]
    where $\delta''\succeq\delta,\delta'$ exists because $\mathcal{T}_{\Sigma}$ is a directed poset. This definition does not depend on the choice of $\delta''$. Indeed, consider $\delta'''\succeq\delta''$, then
    \[[\alpha \iota_{\delta'''\delta}(\psi_{\delta})+\beta \iota_{\delta'''\delta'}(\phi_{\delta'}),\delta''']=[\iota_{\delta'''\delta''}(\alpha \iota_{\delta''\delta}(\psi_{\delta})+\beta \iota_{\delta''\delta'}(\phi_{\delta'})),\delta''']=[\alpha \iota_{\delta''\delta}(\psi_{\delta})+\beta \iota_{\delta''\delta'}(\phi_{\delta'}),\delta''].\]
    Similarly, it does not depend on the choice of the representatives of the equivalence classes.\\
    If the injection maps are also isometries, then there exists a canonical scalar product
    \[\braket{[\psi_{\delta},\delta]|[\phi_{\delta'},\delta']}_{\Hind}=\braket{\iota_{\delta''\delta}(\psi_{\delta})|\iota_{\delta''\delta'}(\phi_{\delta'})}_{\Hil_{\delta''}}\;,\]
    which does not depend either on the choice of $\delta''$ or on the representatives.
\end{proof}

\noindent To prove Prop.~\ref{prop:factor}, we first need to construct the inductive family in the case in which $\Sigma$ has finite connected components. This is given by the following lemma.
\begin{lemma}
\label{connIndF}
    Let $\Sigma=\Sigma_1\sqcup\Sigma_2$ with $\Sigma_{1,2}$ connected. Then, $\mathcal{T}_{\Sigma}$ is a directed poset equivalent to $\mathcal{T}_{\Sigma_1}\times\mathcal{T}_{\Sigma_2}$. Moreover, \[(\{\h_{\delta_1\sqcup\delta_2}=\h_{\delta_1}\otimes\h_{\delta_2}\}_{\delta_1\sqcup\delta_2\in \mathcal{T}_{\Sigma}},\{\iota_{\delta_1'\delta_2',\delta_1\delta_2}=\iota_{\delta_1'\delta_1}\otimes\iota_{\delta'_2\delta_2}\}_{\delta'_1\sqcup\delta'_2\succeq\delta_1\sqcup\delta_2})\] with $(\delta_1,\delta_2)\in\mathcal{T}_{\Sigma_1}\times\mathcal{T}_{\Sigma_2}$ is an inductive family.
\end{lemma}
\begin{proof}
    The equivalence between $\mathcal T_{\Sigma}$ and $\mathcal{T}_{\Sigma_1}\times\mathcal{T}_{\Sigma_2}$ as posets is a direct consequence of the two components being disconnected.~Indeed, we have $\delta'_1\sqcup\delta'_2\succeq\delta_1\sqcup\delta_2$ iff $\delta'_1\succeq\delta_1$ and $\delta'_2\succeq\delta_2$.~The mapping $\delta_1\sqcup\delta_2\mapsto (\delta_1,\delta_2)$ is thus a bijective map between posets and preserve the ordering.\\
    The second part of the claim can be proved by direct inspection.~Specifically, let us verify the two properties in the definition \eqref{def:inductive} of an inductive family with $\delta_1\preceq\delta'_1\preceq\delta''_1\in\mathcal{T}_{\Sigma_1}$ and $\delta_2\preceq\delta'_2\preceq\delta''_2\in\mathcal{T}_{\Sigma_2}$:
    \begin{align*}
        &1)\ \ \iota_{\delta_1\delta_1}\otimes\iota_{\delta_2\delta_2}=\mathbbm{1}_{\h_{\delta_1}}\otimes\mathbbm{1}_{\h_{\delta_2}}=\mathbbm{1}_{\h_{\delta_1}\otimes\h_{\delta_2}};\\
        &2)\ \ (\iota_{\delta''_1\delta'_1}\otimes\iota_{\delta''_2\delta'_2})\circ(\iota_{\delta'_1\delta_1}\otimes\iota_{\delta'_2\delta_2})=(\iota_{\delta''_1\delta'_1}\circ\iota_{\delta'_1\delta_1})\otimes(\iota_{\delta''_2\delta'_2}\circ\iota_{\delta'_2\delta_2})=(\iota_{\delta''_1\delta_1})\otimes(\iota_{\delta''_2\delta_2}).
    \end{align*}
\end{proof}

\noindent\textbf{Proposition~\ref{prop:factor}.}~\emph{
    $\h_{\Sigma_1\sqcup\Sigma_2}^{\rm ind}$ and $\h_{\Sigma_1}^{\rm ind}\otimes\h_{\Sigma_2}^{\rm ind}$ are isomorphic as Hilbert spaces.}
\begin{proof}
    Let us first notice that the span of vectors of the type $[\psi_{\delta_1}\otimes\phi_{\delta_2} ,\delta_1\sqcup \delta_2]$ is dense in $\h_{\Sigma_1\sqcup\Sigma_2}^{\rm ind}$. The inner product on this kind of vectors is given by
    \begin{align*}
        &\braket{[\psi_{\delta_1}\otimes\phi_{\delta_2},\delta_1\sqcup\delta_2]|[\psi_{\delta'_1}\otimes\phi_{\delta'_2},\delta'_1\sqcup\delta'_2]}_{\h_{\Sigma_1\sqcup\Sigma_2}^{\rm ind}}\\
        &\qquad\qquad\qquad\qquad=\braket{[\iota_{\delta''_1\delta_1}(\psi_{\delta_1})\otimes\iota_{\delta''_2\delta_2}(\phi_{\delta_2}),\delta''_1\sqcup\delta''_2]|[\iota_{\delta''_1\delta'_1}(\psi_{\delta'_1})\otimes\iota_{\delta''_2\delta'_2}(\phi_{\delta'_2}),\delta''_1\sqcup\delta''_2]}_{\h_{\Sigma_1\sqcup\Sigma_2}^{\rm ind}}\\
        &\qquad\qquad\qquad\qquad=\braket{\iota_{\delta''_1\delta_1}(\psi_{\delta_1})\otimes\iota_{\delta''_2\delta_2}(\phi_{\delta_2})|\iota_{\delta''_1\delta'_1}(\psi_{\delta'_1})\otimes\iota_{\delta''_2\delta'_2}(\phi_{\delta'_2})}_{\h_{\delta''_1\sqcup\delta''_2}}\\
        &\qquad\qquad\qquad\qquad=\braket{\iota_{\delta''_1\delta_1}(\psi_{\delta_1}),\iota_{\delta''_1\delta'_1}(\phi_{\delta'_1})}_{\h_{\delta''_1}}\braket{\iota_{\delta''_2\delta_2}(\phi_{\delta_2}),\iota_{\delta''_2\delta'_2}(\phi_{\delta'_2})}_{\h_{\delta''_2}},
    \end{align*}
    where $\delta''_1\succeq\delta_1,\delta_1'$ and $\delta''_2\succeq\delta_2,\delta_2'$.

Let now $L$ be the linear map from the span of decomposable tensors in $\h_{\Sigma_1}^{\rm ind}\otimes\h_{\Sigma_2}^{\rm ind}$ to $\h_{\Sigma_1\sqcup\Sigma_2}^{\rm ind}$ given by
$$
L([\psi_{\delta_1},\delta_1]\otimes[\phi_{\delta_2},\delta_2])=[\psi_{\delta_1}\otimes\phi_{\delta_2} ,\delta_1\sqcup \delta_2]\;.
$$
This clearly does not depend on the representatives.~Moreover, this map is an isometry.~Indeed, the scalar product on the decomposable tensors in $\h_{\Sigma_1}^{\rm ind}\otimes\h_{\Sigma_2}^{\rm ind}$ is
    \begin{align*}
        \braket{[\psi_{\delta_1},\delta_1]\otimes[\phi_{\delta_2},\delta_2],[\psi_{\delta'_1},\delta'_1]\otimes[\phi_{\delta'_2},\delta'_2]}_{\h_{\Sigma_1}^{\rm ind}\otimes\h_{\Sigma_2}^{\rm ind}}&=\braket{[\psi_{\delta_1},\delta_1],[\psi_{\delta'_1},\delta'_1]}_{\h_{\Sigma_1}^{\rm ind}}\braket{[\phi_{\delta_2},\delta_2],[\phi_{\delta'_2},\delta'_2]}_{\h_{\Sigma_2}^{\rm ind}}\\
        &=\braket{\iota_{\delta''_1\delta_1}(\psi_{\delta_1}),\iota_{\delta''_1\delta'_1}(\phi_{\delta'_1})}_{\h_{\delta''_1}}\braket{\iota_{\delta''_2\delta_2}(\phi_{\delta_2}),\iota_{\delta''_2\delta'_2}(\phi_{\delta'_2})}_{\h_{\delta''_2}},
    \end{align*}
    here $\delta''_1,\delta''_2$ are the same as before. From this, it follows that
    \begin{align*}
        \braket{[\psi_{\delta_1},\delta_1]\otimes[\phi_{\delta_2},\delta_2],[\psi_{\delta'_1},\delta'_1]\otimes[\phi_{\delta'_2},\delta'_2]}_{\h_{\Sigma_1}^{\rm ind}\otimes\h_{\Sigma_2}^{\rm ind}}&=\braket{[\psi_{\delta_1}\otimes\phi_{\delta_2},\delta_1\sqcup\delta_2]|[\psi_{\delta'_1}\otimes\phi_{\delta'_2},\delta'_1\sqcup\delta'_2]}_{\h_{\Sigma_1\sqcup\Sigma_2}^{\rm ind}}\\
        &=\braket{L([\psi_{\delta_1},\delta_1]\otimes[\phi_{\delta_2},\delta_2]),L([\psi_{\delta'_1},\delta'_1]\otimes[\phi_{\delta'_2},\delta'_2])}_{\h_{\Sigma_1\sqcup\Sigma_2}^{\rm ind}}.
    \end{align*}
    Hence, $L$ is a bounded injective linear map, admitting a unique extension on $\h_{\Sigma_1}^{\rm ind}\otimes\h_{\Sigma_2}^{\rm ind}$.~Its range coincides with the span of vectors of the type $[\psi_{\delta_1}\otimes\phi_{\delta_2} ,\delta_1\sqcup \delta_2]$ because the injective maps that define the inductive family factorise.~The inverse of $L$ is then also well-defined and is an isometry.~Thus, the map $L$ is unitary.
\end{proof}

\noindent\textbf{Proposition~\ref{prop:dual}.}~\emph{$\h_{\Sigma^*}^{\rm ind}= (\h_{\Sigma}^{\rm ind})^*$, and the following diagram}
 \begin{equation}
    \label{app:dualcd}
    \begin{tikzcd} 
      & \h_{\delta_1} \arrow{ld}[swap]{\iota_{\delta_1}} \arrow[r, "R_{\delta_1}"] &  \h_{\delta_1}^* \arrow[rd, "\iota_{\delta_1^*}"] & \\ 
      \h_{\Sigma}^{\rm ind} \arrow{rrr}[swap]{R_{\Sigma}}& & & (\h_{\Sigma}^{\rm ind})^*\\
       & \h_{\delta_2} \arrow[from=uu, crossing over, "\iota_{\delta_2\delta_1}" near start]\arrow[lu, "\iota_{\delta_2}"]\arrow{r}[swap]{R_{\delta_2}} &  \h_{\delta_2}^* \arrow{ru}[swap]{\iota_{\delta_2^*}} \arrow[from=uu, crossing over, "\iota_{\delta_2^*\delta_1^*}" near start]& 
    \end{tikzcd}
\end{equation}    
    \emph{commutes. Furthermore, $\h_{\varnothing}^{\rm ind}\simeq \C$.}
\begin{proof}
    The injective maps $\iota_{\delta_2\delta_1}$ are such that the following diagram
\[\begin{tikzcd} 
     \h_{\delta_1} \arrow[r, "\iota_{\delta_2\delta_1}"] \arrow{d}[swap]{R_{\delta_1}} & \h_{\delta_2} \arrow[d, "R_{\delta_2}"]\\ 
      \h_{\delta_1}^* \arrow{r}[swap]{\iota_{\delta_2^*\delta_1^*}} &  \h_{\delta_2}^*
    \end{tikzcd}\]
commutes.~This ensures that $\h_{\Sigma^*}^{\rm ind}= (\h_{\Sigma}^{\rm ind})^*$. Indeed, let $[f,\delta^*]\in\h_{\Sigma^*}^{\rm ind}$, with $f=R_\delta\varphi$, for some $\varphi\in\h_\delta$. Then, a pairing with $[\psi,\delta']\in \h_{\Sigma}^{\rm ind}$ can be defined as follows ($\delta_1\succeq\delta,\delta'$)
\begin{align*}
    [f,\delta^*][[\psi,\delta']]&=\iota_{\delta_1^*\delta^*}(f)[\iota_{\delta_1\delta'}(\psi)]\\
    &=R_{\delta_1}\circ\iota_{\delta_1\delta}(\varphi)[\iota_{\delta_1\delta'}(\psi)]\\
    &=\braket{\iota_{\delta_1\delta}(\varphi)|\iota_{\delta_1\delta'}(\psi)}_{\h_{\delta_1}}\\
    &=\braket{[\varphi,\delta]|[\psi,\delta']}_{\h_{\Sigma}^{\rm ind}}\;.
\end{align*}
This recovers the scalar product and the Riesz map on $\h_{\Sigma}^{\rm ind}$, hence $[f,\delta^*]\in(\Hind)^*$. As a consequence, the diagram \ref{app:dualcd} commutes.~Finally, the property $\h_{\varnothing}^{\rm ind}\simeq \C$ simply follows from $\mathcal{T}_{\varnothing}=\varnothing$.
\end{proof}
\noindent
The following lemma and its corollary are needed to prove Theorem~\ref{Thm:no-go}.
\noindent
\begin{lemma}
    \label{lem:cutsubnet}
    Let $M=M_1\cup_{\Sigma}M_2$, with $\partial M_1=\Sigma^*_1\sqcup\Sigma,\,\del M_2=\Sigma^*\sqcup\Sigma_2$, and define
    $$
    S_{M_1,M_2,\Sigma}=\{(\Delta_1,\Delta_2)\in\mathcal{T}_{M_1}\times\mathcal{T}_{M_2}\ |\ \partial\Delta_1=\delta_1^*\sqcup\delta,\,\partial\Delta_2=\delta^*\sqcup\delta_2\  \text{for some}\ \delta_1\in\mathcal{T}_{\Sigma_1},\delta_2\in\mathcal{T}_{\Sigma_2},\delta\in\mathcal{T}_{\Sigma}\}.
    $$
    The subsposet $S_{M_1,M_2,\Sigma}$ is directed.~Furthermore,
    \begin{enumerate}
        \item the map $\gamma_{M_1,M_2,\Sigma}:S_{M_1,M_2,\Sigma}\to\mathcal{T}_{M_1\cup_{\Sigma}M_2}$ by $(\Delta_1,\Delta_2)\mapsto\Delta_1\cup_{\delta}\Delta_2$ is order-preserving and cofinal;
        \item the inclusion map $S_{M_1,M_2,\Sigma}\hookrightarrow\mathcal{T}_{M_1}\times\mathcal{T}_{M_2}$ is order-preserving and cofinal.
    \end{enumerate}
\end{lemma}
\begin{proof}
The set $S_{M_1,M_2,\Sigma}$ is a poset because it inherits the product order of $\mathcal{T}_{M_1}\times\mathcal{T}_{M_2}$. Let now $(\Delta_1,\Delta_2)$ and $(\Delta'_1,\Delta'_2)$ in $\mathcal{T}_{M_1}\times\mathcal{T}_{M_2}$ be such that
\[\partial\Delta_1=\delta_1^*\sqcup\delta\,,\,\partial\Delta_2=\delta^*\sqcup\delta_2\,,\,\partial\Delta'_1=(\delta'_1)^*\sqcup\delta'\,,\,\partial\Delta'_2=(\delta')^*\sqcup\delta'_2\,,\]
and let us consider the common refinement $\Delta''_1\succeq\Delta_1,\Delta'_1$ and $\Delta''_2\succeq\Delta_2,\Delta'_2$ with $\partial\Delta''_1=(\delta''_1)^*\sqcup\delta_a,\,\partial\Delta_2=(\delta_b)^*\sqcup\delta''_2$. By Lemma~\ref{lem:guess}, $\delta_a\succeq\delta$ and $\delta_b\succeq \delta$ both in $\mathcal{T}_{\Sigma}$. Since $\mathcal{T}_{\Sigma}$ is a directed set, there exists $\delta''\succeq \delta_a,\delta_b$. So now, by Lemma~\ref{lem:guess}, there exist $\Delta'''_1\succeq\Delta''_1$ and $\Delta'''_2\succeq\Delta''_2$ with $\partial\Delta'''_1=(\delta''_1)^*\sqcup\delta'',\,\partial\Delta'''_2=(\delta'')^*\sqcup\delta''_2$. Then, $(\Delta'''_1,\Delta'''_2)$ is the majorant of $(\Delta_1,\Delta_2)$ and $(\Delta'_1,\Delta'_2)$ in $S_{M_1,M_2,\Sigma}$. Therefore, $S_{M_1,M_2,\Sigma}$ is directed.
\smallskip

\noindent
We now prove statements \emph{1} and \emph{2}.
\begin{itemize}
    \item[\emph{1.}] The map $\gamma_{M_1,M_2,\Sigma}$ is order-preserving since, if $(\Delta'_1,\Delta'_2)\succeq(\Delta_1,\Delta_2)\in S_{M_1,M_2,\Sigma}$, then $\Delta'_1\cup_{\delta'}\Delta'_2\succeq\Delta_1\cup_{\delta}\Delta_2\in\mathcal{T}_M$.~The map is cofinal because of Lemma~\ref{lem:guess}.1.~In fact, let $\Delta\in\mathcal{T}_M$, and let us consider a triangulation $\gamma_{M_1,M_2,\Sigma}(\Delta_1,\Delta_2)\in\mathcal{T}_M$ image of a element in $S_{M_1,M_2,\Sigma}$.~There exists a common refinement $\Delta'\succeq\Delta,\gamma_{M_1,M_2,\Sigma}(\Delta_1,\Delta_2)$, but, by Lemma~\ref{lem:guess}.1, $\Delta'\in S_{M_1,M_2,\Sigma}$.
    \item[\emph{2.}] The map is order-preserving since the order on $S_{M_1,M_2,\Sigma}$ is induced by the product order so that $S_{M_1,M_2,\Sigma}$ is a proper subposet of $\mathcal{T}_{M_1}\times\mathcal{T}_{M_2}$.~It is cofinal because of Lemma~\ref{lem:guess}.2.~In fact, let $\Delta_1$ and $\Delta_2$ have on the common boundary $\Sigma$ two different triangulations $\tilde\delta_1$ and $\tilde\delta_2$.~Considering then a refinement $\delta \succeq\tilde\delta_1,\tilde\delta_2$ of the boundary, by Lemma~\ref{lem:guess}.2, it is associated to refinements $\Delta_1'\succeq\Delta_1$ and $\Delta_2'\succeq\Delta_2$, but $(\Delta_1',\Delta_2')\in S_{M_1,M_2,\Sigma}$ by definition.
\end{itemize}
\end{proof}

\begin{corollary}
    \label{diagonalcut}
    Considering the diagonal subset $D_M=\{(\Delta,\Delta^*)\in\mathcal{T}_{M}\times\mathcal{T}_{M^*}\}$, with $\del M=\Sigma_1^*\sqcup\Sigma$.~Then, the following statements hold true.
    \begin{enumerate}
        \item The map $\gamma_{M,\Sigma}:D_{M}\to\mathcal{T}_{M\cup_{\Sigma}M^*}$ by $(\Delta,\Delta^*)\mapsto\Delta\cup_{\delta}\Delta^*$, with $\delta$ the triangulation on $\Sigma$ induced by $\Delta$, is order-preserving and cofinal.
        \item The inclusion map $D_{M}\hookrightarrow\mathcal{T}_{M}\times\mathcal{T}_{M^*}$ is order-preserving and cofinal.
    \end{enumerate}
\end{corollary}
\begin{proof}
We first notice that $\gamma_{M,\Sigma}:D_M\hookrightarrow S_{M,M^*,\Sigma}\xrightarrow{\gamma_{M,M^*,\Sigma}}\mathcal{T}_{M\cup_{\Sigma}M^*}$ and also the inclusion map is a composition of immersions $D_M\hookrightarrow S_{M,M^*,\Sigma}\hookrightarrow\mathcal{T}_{M}\times\mathcal{T}_{M^*}.$ The thesis follows by the fact that composition of order-preserving and cofinal maps is an order-preserving and cofinal map.\\
We just need to show that $D_M\hookrightarrow S_{M,M^*,\Sigma}$ is a order-preserving and cofinal map. The map is order-preserving since the order on $D_{M}$ is induced by the product order, thus $D_{M}$ is a proper subposet of $S_{M,M^*,\Sigma}$. It is cofinal because, for every couple $(\Delta_1,\Delta_2^*)\in S_{M,M^*,\Sigma}$, there exists a common refinement $\Delta\succeq\Delta_1,\Delta_2\in\mathcal{T}_{M}$. Then, $(\Delta,\Delta^*)\succeq(\Delta_1,\Delta_2^*)$.
\end{proof}
\noindent
We are now ready to prove Theorem~\ref{Thm:no-go}.
\smallskip

\noindent
\textbf{Theorem~\ref{Thm:no-go}}~\textbf{(No-go theorem).}~\emph{Suppose that, for each $M$, the net $\tilde{Z}_\bullet: \mathcal{T}_M\to \Hil^{\rm ind}_{\del M}$ converges to an element
    \begin{equation}
        Z(M)=\lim_{\Delta\in\mathcal{T}_M}\tilde{Z}_{\Delta}\in\Hil^{\rm ind}_{\del M}.
    \end{equation}
    Then, the map 
    \begin{equation*}
        \begin{matrix}
            \Sigma\mapsto\Hind,\\
            M\mapsto Z(M)
        \end{matrix}
    \end{equation*}
    defines a Topological Quantum Field Theory.}
    \begin{proof}
        We have to check that the map satisfies the Atiyah's axioms as given in Def~\ref{def:TQFT}. Propositions \ref{prop:factor} and \ref{prop:dual} ensure that $\Sigma\mapsto\Hind$ satisfies Involution, Factorisation, and Normalisation. Hence, we just need to verify the properties of the limit $Z(M)$.
        Normalisation is immediate:
        \[Z(\varnothing)=\lim_{\Delta\in\mathcal{T}_{\varnothing}}\tilde{Z}_{\Delta}=\tilde{Z}_{\varnothing}=\iota_{\varnothing}(Z(\varnothing))=1.\]
        The Dagger property follows from the continuity of the Riesz map and the Dagger property of $\tilde{Z}$, namely
        \[Z(M^*)=\lim_{\Delta\in\mathcal{T}_{M^*}}\tilde{Z}_{\Delta}=\lim_{\Delta\in\mathcal{T}_{M}}\tilde{Z}_{\Delta^*}=\lim_{\Delta\in\mathcal{T}_{M}}R_{\Sigma}(\tilde{Z}_{\Delta})=R_{\Sigma}(\lim_{\Delta\in\mathcal{T}_{M}}\tilde{Z}_{\Delta})=R_{\Sigma}(Z(M)).\]
        The gluing property requires a bit more discussion. Consider $M=M_1\cup_{\Sigma}M_2$, with $\del M_1=\Sigma_1^*\sqcup\Sigma$ and $\del M_2=\Sigma^*\sqcup\Sigma_2$, then
        \begin{align*}
    Z(M)=\lim_{\Delta\in\mathcal{T}_M}\tilde{Z}_{\Delta}=\lim_{(\Delta_1,\Delta_2)\in S_{M_1,M_2,\Sigma}}\tilde{Z}_{\Delta_1\cup_{\delta}\Delta_2}=\lim_{(\Delta_1,\Delta_2)\in S_{M_1,M_2,\Sigma}}\braket{\tilde{Z}_{\Delta_1},\tilde{Z}_{\Delta_2}}_{\Sigma}\;,
\end{align*}
where, in the second equality, we used Lemma~\ref{lem:cutsubnet} together with the equality of the limits of a net and of all of its subnets; while, in the last equality, we used the gluing property of $\tilde Z_{\bullet}$ from the main text (cfr.~property 3 below Eq.~\eqref{eq:ind-net}).
    Let us now look at the product $\braket{Z(M_1),Z(M_2)}_{\Sigma}$. Using Lemma~\ref{lem:contProd}, with the continuous function given the natural pairing $\braket{\cdot,\cdot}_{\Sigma}$, and Lemma~\ref{lem:cutsubnet}, we have
    \begin{align*}
\braket{Z(M_1),Z(M_2)}_{\Sigma}&=\left\langle\lim_{\Delta_1\in\mathcal{T}_{M_1}}\tilde{Z}_{\Delta_1},\lim_{\Delta_2\in\mathcal{T}_{M_2}}\tilde{Z}_{\Delta_2}\right\rangle_{\Sigma}\\
    &=\lim_{(\Delta_1,\Delta_2)\in\mathcal{T}_{M_1}\times\mathcal{T}_{M_2}}\braket{\tilde{Z}_{\Delta_1},\tilde{Z}_{\Delta_2}}_{\Sigma}\\
    &=\lim_{(\Delta_1,\Delta_2)\in S_{M_1,M_2,\Sigma}}\braket{\tilde{Z}_{\Delta_1},\tilde{Z}_{\Delta_2}}_{\Sigma}
\end{align*}
Then, from the unicity of the limit, we conclude that
\[Z(M_1\cup_{\Sigma}M_2)=\braket{Z(M_1),Z(M_2)}_{\Sigma}.\]
This holds true also for the cylinder $M_{1,2}=\Sigma\times I$, from which it follows that $Z(\Sigma\times I)$ is a projector: 
\[Z(\Sigma\times I)=\braket{Z(\Sigma\times I),Z(\Sigma\times I)}_{\Sigma}.\]
Thus, the theory specified by the map under consideration defines a semi-functor for $\mathsf{Cob}_{\rm d}^{\rm or}$ to $\mathsf{Hilb}$. Turning the semi-functor into functor by the standard procedure, we can just consider as Hilbert space $\mathrm{Ran}(Z(\Sigma\times I))$ and get also the Identity property of a TQFT.
\end{proof}

\noindent\textbf{Theorem~\ref{Thm:riggingmap}}~\textbf{(Rigging map).}~\emph{Let $\Sigma$ be a closed, oriented $(d-1)$-dimensional manifold.~The map $P_{\Sigma}:\mathcal D_{\Sigma}\to\mathcal D_{\Sigma}'$ given by}
\be\label{app:PSigma}
    P_{\Sigma}(\phi)[\psi]\coloneqq Z(\Sigma\times I)[R_{\Sigma}(\phi)\otimes\psi]\;,\qquad \phi,\psi\in\mathcal D_{\Sigma}\;,
\ee
\emph{defines a \emph{rigging map}.~That is, it satisfies the following properties:}
\begin{enumerate}
    \item[\emph{i.}] Anti-linearity: $P_{\Sigma}(\alpha\phi_1+\beta\phi_2)[\psi]=\left(\overline{\alpha}P_{\Sigma}(\phi_1)+\overline{\beta}P_{\Sigma}(\phi_2)\right)[\psi]$\emph{, for any $\phi_1,\phi_2,\psi\in\mathcal D_{\Sigma}$ and $\alpha,\beta\in\mathbb C$};
    \item[\emph{ii.}] Reality: $P_{\Sigma}(\phi)[\psi]=\overline{P_{\Sigma}(\psi)[\phi]}$\emph{, for any $\phi,\psi\in\mathcal D_{\Sigma}$};
    \item[\emph{iii.}]Positive semi-definiteness: $P_{\Sigma}(\psi)[\psi]\geq0$\emph{, for any $\psi\in\mathcal D_{\Sigma}$}.
\end{enumerate}

\begin{proof}
    Property $i.$ straightforwardly follows from the anti-linearity of the Riesz map $R_{\Sigma}$ and the linearity of $Z(\Sigma\times I)$.~As for property $ii.$, recalling the definition~\eqref{dLim} of $Z(M)$ and invoking the Riesz representation theorem, we have that
    \be
    P_{\Sigma}(\phi)[\psi]=\lim_{\Delta\in\mathcal{T}_{\Sigma\times I}}\braket{\phi|\tilde{Z}_{\Delta}\psi}_{\Hind}\;,
    \ee
    where, for $\del\Delta=\delta_1^*\sqcup\delta_2$, $\tilde{Z}_{\Delta}=\iota_{\delta_1^*\sqcup\delta_2}Z(\Delta)\in\Hil_{\Sigma^*\sqcup\Sigma}^{\rm ind}\simeq(\Hind)^*\otimes\Hind\simeq\mathcal B_{\rm HS}(\Hind,\Hind)$.~Therefore, proving reality of $P_{\Sigma}$ amounts to show that
    \be\label{eq:reality}
    \lim_{\Delta\in\mathcal{T}_{\Sigma\times I}}\braket{\phi|\tilde{Z}_{\Delta}\psi}_{\Hind}=\lim_{\Delta\in\mathcal{T}_{\Sigma\times I}}\overline{\braket{\psi|\tilde{Z}_{\Delta}\phi}}_{\Hind}\;,
    \ee
for any $\phi,\psi\in\mathcal \cD_{\Sigma}$.

Our strategy will be to show that there exists a subnet for which~\eqref{eq:reality} holds true so that, owing to the uniqueness of the limit of a net for Hausdorff spaces and the fact that each of its subnets admits the same limit (cfr.~App.~\ref{App:Triang}), property \textit{ii.} is proved.

To this aim, let us consider the subset $\tilde{Z}_{\gamma_{\Sigma\times I}(\bullet)}:\cD_{\Sigma\times I}\to\Hind$.~By Corollary~\ref{diagonalcut}, it is a subnet of $\tilde{Z}_{\bullet}:\mathcal T_{\Sigma\times I}\to\Hind$ and, since $\Hind$ is a Hilbert space, hence Hausdorff, the limit of the net is unique and all the subnets converge to the same limit.~Therefore, the LHS of Eq.~\eqref{eq:reality} gives
\be\label{reality:subnet}
\lim_{\Delta\in\mathcal{T}_{\Sigma\times I}}\braket{\phi|\tilde{Z}_{\Delta}\psi}_{\Hind}=\lim_{(\Delta,\Delta^*)\in D_{\Sigma\times I}}\braket{\phi|\tilde{Z}_{\\\gamma_{\Sigma\times I,\Sigma}(\Delta,\Delta^*)}\psi}_{\Hind}=\lim_{\Delta\cup_{\delta}\Delta^*\in\mathcal{T}_{\Sigma\times I}}\braket{\phi|\tilde{Z}_{\Delta\cup_{\delta}\Delta^*}\psi}_{\Hind}\;.
\ee
Let us now look at the quantity in the limit on the RHS.~We have
\begin{align}
\braket{\phi|\tilde{Z}_{\Delta\cup_{\delta}\Delta^*}\psi}_{\Hind}&=\braket{\phi|(\iota_{\delta_1^*\sqcup\delta_1}Z(\Delta\cup_{\delta}\Delta^*))\psi}_{\Hind}\nonumber\\
&=\braket{\phi|\iota_{\delta_1}\circ Z(\Delta\cup_{\delta}\Delta^*)\circ\iota_{\delta_1}^{\dagger}(\psi)}_{\Hind}\nonumber\\
&=\braket{\phi_{\delta_1}|Z(\Delta\cup_{\delta}\Delta^*)\psi_{\delta_1}}_{\Hil_{\delta_1}}\nonumber\\
&=\overline{\braket{\psi_{\delta_1}|Z(\Delta\cup_{\delta}\Delta^*)\phi_{\delta_1}}}_{\Hil_{\delta_1}}\nonumber\\
&=\overline{\braket{\psi|\tilde{Z}_{\Delta\cup_{\delta}\Delta^*}\phi}}_{\Hind}\;,\label{reality:intsteps}
\end{align}
where, in the first line, $\del(\Delta\cup_{\delta}\Delta^*)=\delta_1^*\sqcup\delta_1$ and we used the property $\iota_{\delta_1^*\sqcup\delta_1}=\iota_{\delta_1^*}\otimes\iota_{\delta_1}$ of the injective maps $\iota_{\delta_1^*\sqcup\delta_1}:\Hil_{\delta_1^*\sqcup\delta_1}\to\Hil_{\Sigma^*\sqcup\Sigma}^{\rm ind}$ (cfr.~Lemma~\ref{connIndF});~in the second line, we used the canonical isomorphism $\Hil_{\delta_1}^*\otimes\Hil_{\delta_1}\simeq\mathcal B_{\rm HS}(\Hil_{\delta_1},\Hil_{\delta_1})$ to regard $Z(\Delta\cup_{\delta}\Delta^*)$ as an operator on $\Hil_{\delta_1}$, together with the compatibility of the maps $\iota_{\delta_1}$ and the Riesz maps at the truncated and inductive level (cfr. Prop.~\ref{prop:dual});~in the third line, $\psi_{\delta_1}=\iota_{\delta_1}^{\dagger}\psi,\,\phi_{\delta_1}=\iota_{\delta_1}^{\dagger}\phi\in\Hil_{\delta_1}$, and we used the isometry property of the injective maps $\iota_{\delta_1}:\Hil_{\delta_1}\to\Hind$;~in the fourth line, we used the duality axiom of spin foams (cfr.~Sec.~\ref{Sec:truncatedaxioms}) to write
\be
Z\left(\Delta\cup_{\delta}\Delta^*\right)=Z\left((\Delta\cup_{\delta}\Delta^*)^*\right)=Z(\Delta\cup_{\delta}\Delta^*)^{\dagger}\;;
\ee
and the last equality follows from repeating the same steps in reverse order.~The result~\eqref{eq:reality} then follows from inserting~\eqref{reality:intsteps} back into~\eqref{reality:subnet} and invoking again the uniqueness of the limit of all subnets of a convergent net for the RHS.

Lastly, as for the property \emph{iii.} of $P_{\Sigma}$, by similar arguments to those leading to the second to last line of Eq.~\eqref{reality:intsteps}, we have
\be
P_{\Sigma}(\psi)[\psi]=\lim_{\Delta\cup_{\delta}\Delta^*\in\mathcal{T}_{\Sigma\times I}}\overline{\braket{\psi_{\delta_1}|Z(\Delta\cup_{\delta}\Delta^*)\psi_{\delta_1}}}_{\Hil_{\delta_1}}\;,
\ee
from which, recalling that $Z(\Delta\cup_{\delta}\Delta^*)=\braket{Z(\Delta),Z(\Delta^*)}_{\delta}\in\Hil_{\delta_1}^*\otimes\Hil_{\delta_1}$, with $\braket{\cdot,\cdot}_{\delta}$ the dual pairing between $\Hil_{\delta}$ and $\Hil_{\delta}^*$, or equivalently $Z(\Delta\cup_{\delta}\Delta^*)=Z(\Delta)^{\dagger}Z(\Delta)$ as an operator on $\Hil_{\delta_1}$, it follows that
\be\label{rigging:semipos}
P_{\Sigma}(\psi)[\psi]=\lim_{\Delta\in\mathcal{T}_{\Sigma\times I}}\overline{\braket{Z(\Delta)\psi_{\delta_1}|Z(\Delta)\psi_{\delta_1}}}_{\Hil_{\delta_1}}=\lim_{\Delta\in\mathcal{T}_{\Sigma\times I}}\|Z(\Delta)\psi_{\delta_1}\|^2_{\Hil_{\delta_1}}\;.
\ee
As $\|Z(\Delta)\psi_{\delta_1}\|^2_{\Hil_{\delta_1}}\geq0$, for any $\psi_{\delta_1}$ and $\Delta$, the RHS of~\eqref{rigging:semipos} is non-negative, i.e.~$P_{\Sigma}$ is positive semi-definite.
\end{proof}

\begin{lemma}\label{lem:Psigmastar}
    Let $\Sigma$ be a closed oriented $(d-1)$-manifold and $\Sigma^*$ the manifold with opposite orientation.~Then, the following statements hold true:
    \begin{itemize}
    \item[(a)] Denoting by $\overline{\cD_{\Sigma}'}$ the set of anti-linear functionals on $\cD_{\Sigma}$, the mapping $R_{\Sigma}':\cD_{\Sigma^*}'\to \overline{\cD_{\Sigma}'}$ defined by
    \be\label{eq:Rprime}
    R'_{\Sigma}(F)=F\circ R_{\Sigma}\;,\qquad F\in\cD_{\Sigma^*}'
    \ee
    is an anti-linear isomorphism.
    \item[(b)] The map $P_{\Sigma^*}:\cD_{\Sigma^*}\to\cD_{\Sigma^*}'$ satisfies the property
    \be\label{eq:Psigmastar}
    P_{\Sigma}(\phi)[\psi]=\overline{P_{\Sigma^*}\left(R_{\Sigma}(\phi)\right)\left[R_{\Sigma}(\psi)\right]}\;,
    \ee
    for any $\phi,\psi\in \mathcal{D}_{\Sigma}$.~That is, the following diagram
    \begin{equation}\label{diag:PSigmaPSigmastar}
        \begin{tikzcd} 
     \mathcal{D}_{\Sigma} \arrow[rr, "P_{\Sigma}"] \arrow{d}[swap]{R_{\Sigma}} & & \mathcal{D}_{\Sigma}' \arrow[d, "\bar{\,\cdot\,}"] & \\ 
      \mathcal{D}_{\Sigma^*} \arrow{r}[swap]{P_{\Sigma^*}} &  \mathcal{D}_{\Sigma^*}' \arrow{r}[swap]{R_{\Sigma}'}& \overline{\mathcal{D}_{\Sigma}'}
    \end{tikzcd}
    \end{equation}
    commutes.
    \end{itemize}
\end{lemma}

\begin{proof}
    As for statement (a), let us start by noticing that the set $\overline{\mathcal{D}_{\Sigma}'}$ of anti-linear functionals on $\cD_{\Sigma}$ is anti-linearly isomorphic to $\mathcal{D}_{\Sigma}'$ via complex conjugation.~Namely, for all $f\in\mathcal{D}_{\Sigma}'$, we can define $\bar{f}\in\overline{\mathcal{D}_{\Sigma}'}$ as $\bar{f}[\psi]=\overline{f[\psi]}$, where $\bar{\cdot}$ on the RHS is the complex conjugation.~Recalling that $\cD_{\Sigma^*}=R_{\Sigma}(\cD_{\Sigma})$  (cfr.~Sec.~\ref{Sec:Continuum2}, assumption \ref{Assumption3}), the map $R'_{\Sigma}:\cD_{\Sigma^*}'\to \overline{\cD_{\Sigma}'}$ is defined by the pullback via Riesz map $R_{\Sigma}$ and, owing to the properties of the Riesz map, it is then anti-linear and invertible.~Explicitly, we have $(R_{\Sigma}')^{-1}\bar{f}[\psi]=\bar{f}[R_{\Sigma}^{-1}(\psi)]$, for any $\psi\in\cD_{\Sigma}$, that is $(R_{\Sigma}')^{-1}\bar{f}=\bar{f}\circ R_{\Sigma}^{-1}$.

\noindent
Let us now prove the statement (b).~Using the definition of $P_{\Sigma}$ given in Theorem~\ref{Thm:riggingmap}, together with the properties of the Riesz map $R_{\Sigma}$ (and its inverse $(R_{\Sigma})^{-1}=R_{\Sigma^*}$), we have
\begin{equation}\label{steps:Psigmastar}
    \begin{split}
        P_{\Sigma^*}(R_{\Sigma}(\phi))[R_{\Sigma}(\psi)&]\,=Z(\Sigma^*\times I)(R_{\Sigma^*}(R_{\Sigma}(\phi))\otimes R_{\Sigma}(\psi))\\
&\;\,=\lim_{\Delta\in\mathcal{T}_{\Sigma^*\times I}}\braket{\tilde{Z}_{\Delta}|\phi\otimes R_{\Sigma}(\psi)}_{\Hind\otimes(\Hind)^*}\\
&\;\,=\lim_{\Delta\in\mathcal{T}_{(\Sigma\times I)^*}}\braket{\tilde{Z}_{\Delta}|R_{\Sigma}(\psi)\otimes \phi}_{(\Hind)^*\otimes\Hind}\\
&\;\,=\lim_{\Delta\in\mathcal{T}_{(\Sigma\times I)^*}}\overline{\braket{R_{\Sigma^*\sqcup \Sigma}(\tilde{Z}_{\Delta})|R_{\Sigma^*\sqcup \Sigma}(R_{\Sigma}(\psi)\otimes \phi)}}_{(\Hind)^*\otimes\Hind}\\
&\;\,=\lim_{\Delta\in\mathcal{T}_{\Sigma\times I}}\overline{\braket{\tilde{Z}_{\Delta}|R_{\Sigma}(\phi)\otimes \psi}}_{(\Hind)^*\otimes\Hind}\\
&\;\,=\overline{P_{\Sigma}(\phi)[\psi]}\;.
    \end{split}
\end{equation}
The third step is non-trivial and is a consequence of how the orientation is induced on the cylinder.~Specifically, according to our convention, we have that $\Sigma^*\times I=(\Sigma\times I)^*$ topologically.~Even though, the corresponding states live in two different Hilbert spaces, these are isomorphic via the symmetry $\sigma_{\Hil_1,\Hil_2}:\Hil_1\otimes\Hil_2\xrightarrow{\sim}\Hil_2\otimes\Hil_1$ of the monoidal category $\mathsf{Hilb}$.~That is, $Z((\Sigma\times I)^*)=\sigma_{\Hind,\Hil^{\mathrm{ind}}_{\Sigma^*}}(Z(\Sigma^*\times I))$.

Finally, the commutativity of the diagram~\eqref{diag:PSigmaPSigmastar} readily follows from combining Eqs.~\eqref{eq:Psigmastar} and~\eqref{eq:Rprime} with $F=P_{\Sigma^*}(R_{\Sigma}(\phi))$. 
\end{proof}

\section{Proofs of Sec.~\ref{Sec:HphysCLQG}}\label{App:Proofs3}

In this appendix we collect the proofs of the propositions and lemmas of Sec.~\ref{Sec:HphysCLQG}, reported in the same order as they appear in the main text.

\noindent
\textbf{Proposition~\ref{prop:physdual} (Involution).}~\emph{Let $\Sigma$ be a closed oriented $(d-1)$-manifold, $\Sigma^*$ the manifold with opposite orientation, and $\Hphys$ as defined in~\eqref{eq:Hphys}.~Then, $\Hil^{\rm phys}_{\Sigma^*}\simeq(\Hphys)^*$.}
\begin{proof}
To prove the claim, we are going to construct an invertible isometry $B_{\Sigma}:\Hil^{\rm phys}_{\Sigma^*}\to(\Hphys)^*$.~To this aim, let us start by considering a pairing $b_{\Sigma}:\Hil^{\rm phys}_{\Sigma^*}\times \Hphys \to \C$ defined as
\begin{equation}
    b_{\Sigma}(\nu_{\rm phys},\phi_{\rm phys})\coloneqq P_{\Sigma^*}(\nu)[R_{\Sigma}(\phi)],
\end{equation}
with $\nu\in\mathcal{D}_{\Sigma^*},\,\varphi\in\mathcal{D}_{\Sigma}$, and $\nu_{\rm phys},\,\phi_{\rm phys}$ their respective equivalence classes in $\Hil^{\rm phys}_{\Sigma^*}$ and $ \Hphys$.~By Lemma~\ref{lem:Psigmastar}, we have
\begin{equation}
    b_{\Sigma}(\nu_{\rm phys},\phi_{\rm phys})=P_{\Sigma}(R_{\Sigma}^{-1}(\nu))[\phi]=\braket{R_{\Sigma}^{-1}(\nu)_{\rm phys}|\phi_{\rm phys}}_{\Hphys}\;,
\end{equation}
from which it follows that $b_{\Sigma}(\nu_{\rm phys},\cdot)\in(\Hphys)^*$.~Let now $B_{\Sigma}$ be the map given by
\begin{equation}
    B_{\Sigma}(\nu_{\rm phys})[\cdot]=b_{\Sigma}(\nu_{\rm phys},\cdot)=\mathfrak{R}_{\Sigma}(R_{\Sigma}^{-1}(\nu)_{\rm phys})[\cdot]\;,
\end{equation}
where $\mathfrak{R}_{\Sigma}:\Hphys\to(\Hphys)^*$ denotes the Riesz map on the physical Hilbert space.~Notice that, by~\eqref{eq:Psigmastar}, $\mathrm{ker}(P_\Sigma)$ and $\mathrm{ker}(P_{\Sigma^*})$ are isomorphic through $R_{\Sigma}$, so that the previous equation is well-defined.\\
As every element of $(\Hphys)^*$ can be written as $\mathfrak{R}_{\Sigma}(\phi_{\rm phys})$, for some $\phi\in\mathcal{D}_{\Sigma}$, the map $B_{\Sigma}$ is invertible, and we have:
\begin{equation}
    B_{\Sigma}^{-1}(\mathfrak{R}_{\Sigma}(\phi_{\rm phys}))=R_{\Sigma}(\phi)_{\rm phys}\;.
\end{equation}
Lastly, let $\nu,\mu$ be in $\cD_{\Sigma^*}$, then
\be
\begin{aligned}
        \braket{B_{\Sigma}(\nu_{\rm phys})|B_{\Sigma}(\mu_{\rm phys})}_{\Hil^{\rm phys}_{\Sigma^*}}&=\braket{\mathfrak{R}_{\Sigma}(R_{\Sigma}^{-1}(\nu)_{\rm phys})|\mathfrak{R}_{\Sigma}(R_{\Sigma}^{-1}(\mu)_{\rm phys})}_{(\Hphys)^*}\\
    &=\overline{\braket{R_{\Sigma}^{-1}(\nu)_{\rm phys}|R_{\Sigma}^{-1}(\mu)_{\rm phys}}}_{\Hphys}\\
    &=\overline{P_{\Sigma}(R_{\Sigma}^{-1}\nu)[R_{\Sigma}^{-1}\mu]}\\
    &=P_{\Sigma^*}(\nu)[\mu]\\
    &=\braket{\nu_{\rm phys}|\mu_{\rm phys}}_{\Hil^{\rm phys}_{\Sigma^*}}\;,
    \end{aligned}
\ee
i.e., the map $B_{\Sigma}$ is an isometry.
\end{proof}

\textbf{Lemma~\ref{lem:factphys}.}~\emph{
    For a canonical $d$-dimensional spin foam model as given in \emph{Def.~\ref{def:split}}, the rigging map, as constructed in \eqref{def:PSigma}, decomposes as
    \begin{equation}
        P_{\Sigma_1\sqcup\Sigma_2}(\phi)[\psi]=P_{\Sigma_1}(\phi_1)[\psi_1]P_{\Sigma_2}(\phi_2)[\psi_2]\;,
    \end{equation}
    for any closed $(d-1)$-manifolds $\Sigma_1$ and $\Sigma_2$, and any decomposable tensors $\phi=\phi_1\otimes\phi_2,\,\psi=\psi_1\otimes\psi_2\in\cD_{\Sigma_1\sqcup\Sigma_2}=\cD_{\Sigma_1}\otimes\cD_{\Sigma_2}$.
}
\begin{proof}
    Recalling that $R_{\Sigma_1\sqcup\Sigma_2}=R_{\Sigma_1}\otimes R_{\Sigma_2}$, and the definition \eqref{def:PSigma}, we have
        \begin{align}
            P_{\Sigma_1\sqcup\Sigma_2}(\phi)[\psi]&=Z((\Sigma_1\sqcup\Sigma_2)\times I)[R_{\Sigma_1\sqcup\Sigma_2}(\phi)\otimes \psi]\nonumber\\
            &=Z((\Sigma_1\times I)\sqcup(\Sigma_2\times I))[R_{\Sigma_1\sqcup\Sigma_2}(\phi)\otimes \psi]\nonumber\\
            &=\lim_{\Delta\in\mathcal{T}_{(\Sigma_1\times I)\sqcup(\Sigma_2\times I)}}\braket{\tilde{Z}_\Delta|R_{\Sigma_1\sqcup\Sigma_2}(\phi)\otimes\psi}_{\h^{\rm ind}_{\Sigma_1^*\sqcup\Sigma_2^*\sqcup\Sigma_1\sqcup\Sigma_2}}\nonumber\\
            &=\lim_{(\Delta_1,\Delta_2)\in\mathcal{T}_{(\Sigma_1\times I)}\times\mathcal{T}_{(\Sigma_2\times I)}}\braket{\phi|\tilde{Z}_{\Delta_1\sqcup\Delta_2}\psi}_{\h^{\rm ind}_{\Sigma_1\sqcup\Sigma_2}}\;,\label{eq:PSigmacanSF}
        \end{align}
    where, in the last line, we used the fact that, for any manifolds $M$ and $N$, $\mathcal{T}_{M\sqcup N}=\mathcal{T}_M\times\mathcal{T}_N$ as posets.
    
    By Lemma~\ref{connIndF}, the requirement of the spin foam model to be canonical (cfr.~Eq.~\eqref{eq:canSF} in Def.~\ref{def:split}) implies that
     \begin{equation}
     \begin{split}
        \tilde{Z}_{\Delta_1\sqcup\Delta_2}&=\iota_{\partial\Delta_1\sqcup\partial\Delta_2}(Z(\Delta_1\sqcup\Delta_2))\\
        &=\iota_{\partial\Delta_1}\otimes\iota_{\partial\Delta_2}(Z(\Delta_1)\otimes Z(\Delta_2))\\
        &=\iota_{\partial\Delta_1}(Z(\Delta_1))\otimes\iota_{\partial\Delta_2}(Z(\Delta_2))\\
        &=\tilde Z_{\Delta_1}\otimes\tilde Z_{\Delta_2}\;.
     \end{split}
    \end{equation}
    Thus, by the hypothesis that $\phi$ and $\psi$ are decomposable tensors, together with Prop.~\ref{prop:factor}, the quantity inside the limit in \eqref{eq:PSigmacanSF} reads as
    \begin{equation}\label{eq:PSigmacanSF2}
        \braket{\phi|\tilde{Z}_{\Delta_1\sqcup\Delta_2}\psi}_{\h^{\rm ind}_{\Sigma_1\sqcup\Sigma_2}}=\braket{\phi_1|\tilde{Z}_{\Delta_1}\psi_1}_{\h^{\rm ind}_{\Sigma_1}}\braket{\phi_2|\tilde{Z}_{\Delta_2}\psi_2}_{\h^{\rm ind}_{\Sigma_2}}.
    \end{equation}
    Finally, plugging \eqref{eq:PSigmacanSF2} back into \eqref{eq:PSigmacanSF}, and invoking Lemma~\ref{lem:contProd} for the nets
    \[\braket{\phi_1|\tilde{Z}_{\bullet}\psi_1}_{\h^{\rm ind}_{\Sigma_1}}:\mathcal{T}_{\Sigma_1\times I}\to\C\quad,\quad
        \braket{\phi_2|\tilde{Z}_{\bullet}\psi_2}_{\h^{\rm ind}_{\Sigma_2}}:\mathcal{T}_{\Sigma_2\times I}\to\C\,,\]
    which converge by assumption, we obtain
    \begin{equation}
    \begin{split}
        P_{\Sigma_1\sqcup\Sigma_2}(\phi)[\psi]&=\lim_{(\Delta_1,\Delta_2)\in\mathcal{T}_{(\Sigma_1\times I)}\times\mathcal{T}_{(\Sigma_2\times I)}}\braket{\phi_1|\tilde{Z}_{\Delta_1}\psi_1}_{\h^{\rm ind}_{\Sigma_1}}\braket{\phi_2|\tilde{Z}_{\Delta_2}\psi_2}_{\h^{\rm ind}_{\Sigma_2}}\\
        &=\lim_{\Delta_1\in\mathcal{T}_{\Sigma_1\times I}}\braket{\phi_1|\tilde{Z}_{\Delta_1}\psi_1}_{\h^{\rm ind}_{\Sigma_1}}\lim_{\Delta_1\in\mathcal{T}_{\Sigma_1\times I}}\braket{\phi_2|\tilde{Z}_{\Delta_2}\psi_2}_{\h^{\rm ind}_{\Sigma_2}}\\
        &=Z(\Sigma_1\times I)[R_{\Sigma_1}(\phi_1)\otimes\psi_1]Z(\Sigma_2\times I)[R_{\Sigma_2}(\phi_2)\otimes\psi_2]\\
        &=P_{\Sigma_1}(\phi_1)[\psi_1]P_{\Sigma_2}(\phi_2)[\psi_2]\;.
    \end{split}
    \end{equation}
\end{proof}
\noindent
\textbf{Proposition~\ref{prop:factphys} (Factorisation).}~\emph{The physical Hilbert space resulting from a canonical $d$-dimensional spin foam model factorises
    \begin{equation*}
        \h^{\rm phys}_{\Sigma_1\sqcup\Sigma_2}\simeq\h^{\rm phys}_{\Sigma_1}\otimes \h^{\rm phys}_{\Sigma_2}\;,
    \end{equation*}
    for any closed $(d-1)$-manifolds $\Sigma_1,\Sigma_2$.}
\begin{proof}
    Recall that decomposable tensors, i.e.~states of the kind $\psi=\psi_1\otimes\psi_2$ with $\psi_1\in\cD_{\Sigma_1}$ and $\psi_2\in\cD_{\Sigma_2}$, generate $\cD_{\Sigma_1\sqcup\Sigma_2}$.~Hence, $\cD_{\Sigma_1\sqcup\Sigma_2}/\ker(P_{\Sigma_1\sqcup\Sigma_2})$ is generated by states of the form $(\psi_1\otimes\psi_2)_{\rm phys}$.~Let then $U:\cD_{\Sigma_1\sqcup\Sigma_2}/\ker(P_{\Sigma_1\sqcup\Sigma_2})\subset \h^{\rm phys}_{\Sigma_1\sqcup\Sigma_2}\to\h^{\rm phys}_{\Sigma_1}\otimes \h^{\rm phys}_{\Sigma_2}$ be the map acting on these states as
    \begin{equation}
        U((\psi_1\otimes\psi_2)_{\rm phys})=(\psi_1)_{\rm phys}\otimes(\psi_2)_{\rm phys},
    \end{equation}
    and extended by linearity.~By Lemma~\ref{lem:factphys}, this map is well defined.~Indeed, Lemma~\ref{lem:factphys} implies that $\ker P_{\Sigma_1\sqcup\Sigma_2}=\ker P_{\Sigma_1}\otimes \cD_{\Sigma_2}+\cD_{\Sigma_1}\otimes\ker P_{\Sigma_2}$ so that, if $\psi_1 \otimes \psi_2$ lies in $\ker P_{\Sigma_1 \sqcup \Sigma_2}$, then at least one among $\psi_1$ and $\psi_2$ must lie in the kernel of $P_{\Sigma_1}$ or $P_{\Sigma_2}$, respectively.~Furthermore, the map $U$ is an isometry. Indeed, from the definition \eqref{def:physinnprod} of the physical inner product, combined with Lemma~\ref{lem:factphys}, it follows that
    \begin{equation}
    \begin{split}
        \braket{(\phi_1\otimes\phi_2)_{\rm phys}|(\psi_1\otimes\psi_2)_{\rm phys}}_{\h^{\rm phys}_{\Sigma_1\sqcup\Sigma_2}}&=P_{\Sigma_1\sqcup\Sigma_2}(\phi_1\otimes\phi_2)[\psi_1\otimes\psi_2]\\
        &=P_{\Sigma_1}(\phi_1)[\psi_1]P_{\Sigma_2}(\phi_2)[\psi_2]\\
        &=\braket{(\phi_1)_{\rm phys}|(\psi_1)_{\rm phys}}_{\h^{\rm phys}_{\Sigma_1}}\braket{(\phi_2)_{\rm phys}|(\psi_2)_{\rm phys}}_{\h^{\rm phys}_{\Sigma_2}}\\
        &=\braket{(\phi_1)_{\rm phys}\otimes(\phi_2)_{\rm phys}|(\psi_1)_{\rm phys}\otimes(\psi_2)_{\rm phys}}_{\h^{\rm phys}_{\Sigma_1}\otimes \h^{\rm phys}_{\Sigma_2}}.
    \end{split}
    \end{equation}
    Since it is an isometry, it is bounded and continuous, so it can be extended to the completion 
    \begin{equation}
        U:\h^{\rm phys}_{\Sigma_1\sqcup\Sigma_2}\to \h^{\rm phys}_{\Sigma_1}\otimes \h^{\rm phys}_{\Sigma_2}.
    \end{equation}
    The last thing we have to show is that the map $U$ is also invertible, for which we follow a similar path.~Let us define a map on $\cD_{\Sigma_1}/\ker(P_{\Sigma_1})\otimes \cD_{\Sigma_2}/\ker(P_{\Sigma_2})$ by its action on decomposable tensors
    \begin{equation}
        \hat{U}((\psi_1)_{\rm phys}\otimes(\psi_2)_{\rm phys})=(\psi_1\otimes\psi_2)_{\rm phys}.
    \end{equation}
    The map is well defined: let $\phi_1$ be in $\ker P_{\Sigma_1}$, then $\hat{U}((\psi_1+\phi_1)_{\rm phys}\otimes(\psi_2)_{\rm phys})=(\psi_1\otimes\psi_2+\phi_1\otimes \psi_2)_{\rm phys}$, but $\phi_1\otimes \psi_2\in\ker P_{\Sigma_1\sqcup\Sigma_2}$ because of Lemma~\ref{lem:factphys}. The same holds for a $\phi_2\in\ker P_{\Sigma_2}$. Moreover, the map $\hat{U}$ is an isometry for the same reason of $U$, and it is clearly its inverse. Extending it to the completion of $\cD_{\Sigma_1}/\ker(P_{\Sigma_1})\otimes \cD_{\Sigma_2}/\ker(P_{\Sigma_2})$, which is nothing but $\h^{\rm phys}_{\Sigma_1}\otimes \h^{\rm phys}_{\Sigma_2}$, completes then the proof of the isomorphism $\h^{\rm phys}_{\Sigma_1\sqcup\Sigma_2}\simeq\h^{\rm phys}_{\Sigma_1}\otimes \h^{\rm phys}_{\Sigma_2}$.
\end{proof}

\subsection{Discussing the Gluing Proposal}\label{AppSec:gluing}

Let us discuss in more detail the Conjecture~\ref{conj:gluing} and present a possible justification based on the ability to perform a key technical step, namely the interchange of a limit with a sum.~This requirement may impose additional conditions on spin foam models, potentially excluding certain classes of models or favouring specific treatments of divergences.
\medskip

\noindent
First of all, let us recall the statement of the proposal:
\medskip

\noindent
\textbf{Conjecture~\ref{conj:gluing} (Gluing).}\emph{
    For any d-dimensional manifolds $M,N$ such that $\partial M=\Sigma_1^*\sqcup \Sigma$ and $\partial N=\Sigma^*\sqcup \Sigma_2$; and, for any decomposable vector $\psi_1\otimes \psi_2\in\mathcal{D}_{\Sigma_1^*\sqcup\Sigma_2}$, the following expression holds
    \begin{equation}\label{eq:gluingconv}
        Z(M\cup_\Sigma N)[\psi_1\otimes\psi_2]=\sum_{\alpha,\beta\in\mathscr I}Z(M)[\psi_1\otimes \varphi_\alpha]P_{\Sigma}(\varphi_\alpha)[\varphi_\beta]Z(N)[R_\Sigma(\varphi_\beta)\otimes\psi_2]\;,
    \end{equation}
    where $\{\varphi_\alpha\}_{\alpha\in\mathscr I}$ is an orthonormal basis of $\mathcal{D}_{\Sigma}$. Moreover, for every $(d-1)$-dimensional closed manifold $\Sigma$, the rigging map \eqref{def:PSigma} satisfies a convolution property
    \begin{equation}\label{eq:Pconvolution}
        P_{\Sigma}(\phi)[\psi]=\sum_{\alpha\in\mathscr I}P_{\Sigma}(\phi)[\varphi_\alpha]P_{\Sigma}(\varphi_\alpha)[\psi]\;,
    \end{equation}
   for all $\phi,\psi\in\cD_{\Sigma}$.}
   \medskip

\noindent
We start by discussing the equality \eqref{eq:Pconvolution}, namely the convolution property for $P_{\Sigma}$.~To this aim, let us rewrite the LHS as
\begin{equation}
    \begin{split}
    P_{\Sigma}(\phi)[\psi]&=Z(\Sigma\times I)[R_{\Sigma}(\phi)\otimes\psi]\\
    &=Z((\Sigma\times I)\cup_{\Sigma}(\Sigma\times I))[R_{\Sigma}(\phi)\otimes\psi]\\
    &=\lim_{\mathcal T_{(\Sigma\times I)\cup_{\Sigma}(\Sigma\times I)}}\braket{\tilde{Z}_{\Delta}|R_{\Sigma}(\phi)\otimes\psi}_{\h^{\rm ind}_{\Sigma^*\sqcup\Sigma}}\\
    &=\lim_{S_{\Sigma\times I,\Sigma\times I,\Sigma}}\braket{\tilde{Z}_{\Delta_{1}\cup_{\delta}\Delta_{2}}|R_{\Sigma}(\phi)\otimes\psi}_{\h^{\rm ind}_{\Sigma^*\sqcup\Sigma}}\;,
    \end{split}
\end{equation}
where, in the first line, we used the fact that $\Sigma\times I\simeq (\Sigma\times I)\cup_{\Sigma}(\Sigma\times I)$, in the second line, we used the definition~\eqref{dLim} of $Z$, and, in the last line, we used Statement 1 of Lemma~\ref{lem:cutsubnet} to rewrite the limit of the net $\tilde{Z}_{\bullet}:\mathcal T_{(\Sigma\times I)\cup_{\Sigma}(\Sigma\times I)}\to\h^{\rm ind}_{\Sigma^*\sqcup\Sigma}$ as the limit of its subnet $\tilde{Z}_{\gamma_{\Sigma\times I,\Sigma\times I,\Sigma}(\bullet)}:S_{\Sigma\times I,\Sigma\times I,\Sigma}\to\h^{\rm ind}_{\Sigma^*\sqcup\Sigma}$ over those triangulations of $(\Sigma\times I)\cup_{\Sigma}(\Sigma\times I)$ which decompose as $\Delta_1\cup_\delta\Delta_{2}$ for some $\delta\in\mathcal T_{\Sigma}$.\footnote{Here, $M_1=\Sigma\times I$, $M_2=\Sigma\times I$ so that $\Sigma_1=\Sigma$ and $\Sigma_2=\Sigma$ (cfr.~Lemma~\ref{lem:cutsubnet}).}
Recalling the gluing property of $\tilde{Z}_{\bullet}$ (cfr.~Sec.~\ref{Sec:SubdivIndLim}), we have $\tilde{Z}_{\Delta_1\cup_{\delta}\Delta_{2}}=\braket{\tilde{Z}_{\Delta_1},\tilde{Z}_{\Delta_{2}}}_{\Sigma}\in \Hil_{\Sigma^*}^{\rm ind}\otimes\Hind$, an operator from $\Hind$ to $\Hind$.~Thus, as an operator, this is given by $\tilde{Z}_{\Delta_{2}}\circ\tilde{Z}_{\Delta_1}$, so that
\begin{equation}
    P_{\Sigma}(\phi)[\psi]=\lim_{S_{M,\Sigma\times I,\Sigma}}\braket{\tilde{Z}_{\Delta_{2}}^\dagger\phi|\tilde{Z}_{\Delta_{1}}\psi}_{\Hind}\;.
\end{equation}
Now, if we add a resolution of the identity on $\Hind$ given by an orthonormal basis $\mathbbm{1}_{\Hind}=\sum_{\alpha}\ket{\varphi_{\alpha}}\bra{\varphi_{\alpha}}$, we obtain
\begin{equation}\label{eq:convLHS}
    P_{\Sigma}(\phi)[\psi]=\lim_{S_{M,\Sigma\times I,\Sigma}}\sum_{\alpha\in\mathscr I}\braket{\phi|\tilde{Z}_{\Delta_{2}}\varphi_{\alpha}}\braket{\varphi_{\alpha}|\tilde{Z}_{\Delta_{1}}\psi}_{\Hind}\;.
\end{equation}
As the $\tilde{Z}_{\Delta}$ are Hilbert-Schmidt operators, the sum contains at most countably many non-zero elements which, together with the existence of the continuum limit of the theory assumed from the outset, ensures that this derivation is well posed.

\smallskip

Let us now look at the RHS of \eqref{eq:Pconvolution}.~By the definitions~\eqref{dLim} and~\eqref{def:PSigma} of $P_{\Sigma}$, this is given by
\begin{equation}\label{eq:cnovRHS1}
    \sum_{\alpha\in\mathscr I}P_{\Sigma}(\phi)[\varphi_\alpha]P_{\Sigma}(\varphi_\alpha)[\psi]=\sum_{\alpha\in\mathscr I}\lim_{\mathcal T_{\Sigma\times I}}\braket{\phi|\tilde{Z}_{\Delta_{1}}\varphi_{\alpha}}_{\Hind}\lim_{\mathcal T_{\Sigma\times I}}\braket{\varphi_{\alpha}|\tilde{Z}_{\Delta_{2}}\psi}_{\Hind}\;.
\end{equation}
For each $\alpha$ (and fixed $\psi$ and $\phi$), we can then apply Lemma~\ref{lem:contProd} with the two nets given by
\begin{equation}
        \begin{matrix}
            a_{\bullet}^{\alpha}:&\mathcal T_{\Sigma\times I}&\to&\mathbb C\;; & \text{and}\;\; & b_{\bullet}^{\alpha}:&\mathcal T_{\Sigma\times I}&\to&\mathbb C\;;\\
            & \Delta & \mapsto & \braket{\phi|\tilde{Z}_{\Delta}\varphi_{\alpha}}_{\Hind}, & & & \Delta & \mapsto & \braket{\varphi_{\alpha}|\tilde{Z}_{\Delta}\psi}_{\Hind},
        \end{matrix}
\end{equation}
and $f$ being the product of complex numbers, thus yielding
\begin{equation}
    \lim_{\mathcal T_{\Sigma\times I}}\braket{\phi|\tilde{Z}_{\Delta_{1}}\varphi_{\alpha}}_{\Hind}\lim_{\mathcal T_{\Sigma\times I}}\braket{\varphi_{\alpha}|\tilde{Z}_{\Delta_{2}}\psi}_{\Hind}=\lim_{\mathcal T_{\Sigma\times I}\times\mathcal T_{\Sigma\times I}}\braket{\phi|\tilde{Z}_{\Delta_{1}}\varphi_{\alpha}}_{\Hind}\!\braket{\varphi_{\alpha}|\tilde{Z}_{\Delta_{2}}\psi}_{\Hind}\;.
\end{equation}
Now, owing to Statement 2 of Lemma~\ref{lem:cutsubnet} and the uniqueness of the limit of a net (and of all its subnets) in a Hausdorff space, the latter is given by the following limit of the subnet over the directed poset $S_{\Sigma\times I,\Sigma\times I,\Sigma}$, that is
\begin{equation}
    \lim_{\mathcal T_{\Sigma\times I}\times\mathcal T_{\Sigma\times I}}\braket{\phi|\tilde{Z}_{\Delta_{1}}\varphi_{\alpha}}_{\Hind}\!\braket{\varphi_{\alpha}|\tilde{Z}_{\Delta_{2}}\psi}_{\Hind}=\lim_{S_{M,\Sigma\times I,\Sigma}}\braket{\phi|\tilde{Z}_{\Delta_{2}}\varphi_{\alpha}}\braket{\varphi_{\alpha}|\tilde{Z}_{\Delta_{1}}\psi}_{\Hind}\;.
\end{equation}
Plugging this back into \eqref{eq:cnovRHS1}, and provided that it is possible to exchange the limit over the poset $S_{\Sigma\times I,\Sigma\times I,\Sigma}$ with the sum over the index $\alpha$, we get
\begin{equation}
\begin{split}
    \sum_{\alpha\in\mathscr I}P_{\Sigma}(\phi)[\varphi_\alpha]P_{\Sigma}(\varphi_\alpha)[\psi]&=\sum_{\alpha\in\mathscr I}\lim_{S_{M,\Sigma\times I,\Sigma}}\braket{\phi|\tilde{Z}_{\Delta_{2}}\varphi_{\alpha}}\braket{\varphi_{\alpha}|\tilde{Z}_{\Delta_{1}}\psi}_{\Hind}\\
    &=\lim_{S_{M,\Sigma\times I,\Sigma}}\sum_{\alpha\in\mathscr I}\braket{\phi|\tilde{Z}_{\Delta_{2}}\varphi_{\alpha}}\braket{\varphi_{\alpha}|\tilde{Z}_{\Delta_{1}}\psi}_{\Hind}\;,
\end{split}
\end{equation}
which is exactly \eqref{eq:convLHS}, thus proving the convolution property \eqref{eq:Pconvolution} of $P_{\Sigma}$.

\medskip

Coming now to the gluing property \eqref{eq:gluingconv}, the following two observations are needed.~First, we notice that Lemma~\ref{lem:contProd} can be generalised to a finite product of spaces as follows.
\begin{lemma}\label{lem:gen_contProd}
    Consider a finite collection of directed posets $A_1,A_2,\dots,A_n$ with $n\in\N$, and a collection of Hausdorff topological spaces $X_1,X_2,\dots,X_n$. Let $a_\bullet^k:A_k\to X_k$ be a net, and assume that the limit $\lim_{A_k}a_\bullet^k=a_k$ exists in $X_k$ for every fixed $k\in\{1,\dots,n\}$. Let $f:\prod_{k=1}^nX_k\to Y$ be a continuous function from the product space $\prod_{k=1}^nX_k$ to an Hausdorff topological space $Y$. Then,
    \begin{equation*}
        f(a_1,a_2,\dots,a_n)=\lim_{\prod_{k=1}^nA_k}f(a_\bullet^1,a_\bullet^2,\dots,a_\bullet^n).
    \end{equation*}
\end{lemma}
\begin{proof}
    The proof is similar to that of Lemma~\ref{lem:contProd}.~Recalling that the product order is $(i_1,i_2,\dots,i_n)\geq (i'_1,i'_2,\dots,i'_n)$ if and only if $i_1\geq i'_1$ and $i_2\geq i'_2$ and so on, the product set $\prod_{k=1}^n A_k$ with the product order is a directed poset.~Consider now the net 
    \begin{equation*}
        \begin{matrix}
            \mathfrak{a}_\bullet: & \prod_{k=1}^n A_k &\to&\prod_{k=1}^n X_k\,;\\
            & (i_1,i_2,\dots,i_n)&\mapsto&(a_{i_1}^1,a^2_{i_2},\dots,a^n_{i_n})\,.
        \end{matrix}
    \end{equation*}
    Notice that the projections $\mathrm{pr}_m:\prod_{k=1}^n A_k\to A_m$ are order-preserving and cofinal. Thus,
    \begin{equation*}
    \lim_{\prod_{k=1}^nA_k}\pi_m(\mathfrak{a}_\bullet)=\lim_{\prod_{k=1}^nA_k}a^m_{\mathrm{pr}_m(\bullet)}=\lim_{A_m}a^m_{i_m}=a_m\;,
    \end{equation*}
    where $\pi_{m}$ is the projector on the $m$-th factor of $\prod_{k=1}^nX_k$. Since a net in a product space has a limit if and only if each projection has a limit, and the two limits coincide, then \[\lim_{\prod_{k=1}^nA_k}\mathfrak{a}_\bullet=(a_1,a_2,\dots,a_n).\]
    By continuity of $f$, the thesis follows.
\end{proof}
Our second observation is that Lemma~\ref{lem:cutsubnet} can also be generalised to the case of gluing of any finite number of manifolds.
\begin{lemma}\label{lem:gen_cutsubnet}
    Let $M=M_1\cup_{\Sigma_1}M_2\cup_{\Sigma_2}\dots\cup_{\Sigma_n}M_{n+1}$, with $\partial M_1=\sigma^*_1\sqcup\Sigma_1,\,\del M_{n+1}=\Sigma^*_n\sqcup\sigma_{n+1}$, and $\del M_{k}=\Sigma^*_{k-1}\sqcup\sigma_{k}\sqcup\Sigma_{k}$ for $k\in\{2,\dots,n\}$.~Let $\mathscr{C}=\{M_1,\dots,M_{n+1},\Sigma_1,\dots,\Sigma_n\}$ be the collection of the $d$-manifolds $M_l$ and the closed $(d-1)$-manifolds $\Sigma_i$, and let $S_{\mathscr{C}}$ be the subposet of $\prod_{l=1}^{n+1}\mathcal{T}_{M_l}$ given by
    \begin{equation}
        S_{\mathscr{C}}=\left\{\begin{matrix} (\Delta_1,\Delta_2,\dots,\Delta_{n+1})\in\prod_{l=1}^{n+1}\mathcal{T}_{M_l}\ \mid\\
            \exists(\delta_1,\dots,\delta_n)\in\prod_{i=1}^{n}\mathcal{T}_{\Sigma_i}\text{ and }(\delta'_1,\dots,\delta'_{n+1})\in\prod_{l=1}^{n+1}\mathcal{T}_{\sigma_l}\text{ such that }\\
             \partial \Delta_1=(\delta'_1)^*\sqcup\delta_1,\,\del \Delta_{n+1}=\delta^*_n\sqcup\delta'_{n+1}, \text{ and } \del \Delta_{k}=\delta^*_{k-1}\sqcup\delta'_{k}\sqcup\delta_{k} \text{ for }k\in\{2,\dots,n\}
        \end{matrix}\right\}\,.
    \end{equation}
    Then, $S_{\mathscr{C}}$ is a directed poset.~Moreover,
    \begin{enumerate}
        \item the map \begin{equation}
            \begin{matrix}
                \gamma_{\mathscr{C}}:&S_{\mathscr{C}}&\to&\mathcal{T}_{M}\,;\\
                &(\Delta_1,\Delta_2,\dots,\Delta_{n+1})&\mapsto&\Delta_1\cup_{\delta_1}\Delta_2\cup_{\delta_2}\dots\cup_{\delta_n}\Delta_{n+1},
            \end{matrix}
        \end{equation} is order-preserving and cofinal;
        \item the inclusion map $S_{\mathscr{C}}\hookrightarrow\prod_{l=1}^{n+1}\mathcal{T}_{M_l}$ is order-preserving and cofinal.
    \end{enumerate}
\end{lemma}
\begin{proof}
We prove the Lemma by induction.~The basis is the case $n=2$, already proved in Lemma~\ref{lem:cutsubnet}.~Let
\[\mathscr{C}_{n+1}=\{M_1,\dots,M_n,M_{n+1},\Sigma_1,\dots,\Sigma_{n-1}\Sigma_n\}\quad\text{and}\quad\mathscr{C}_n=\{M_1,\dots,M_{n},\Sigma_1,\dots,\Sigma_{n-1}\},\] and let us assume the Lemma to be true for the collection $\mathscr{C}_n$.~We are now going to prove it for $\mathscr{C}_{n+1}$.

$S_{\mathscr{C}_{n+1}}$ inherits the order from $\prod_{l=1}^{n+1}\mathcal{T}_{M_l}$, hence it is a poset.~It is also directed because, if $(\Delta_1,\dots,\Delta_n,\Delta_{n+1})$ and $(\Delta'_1,\dots,\Delta'_n,\Delta'_{n+1})$ are in $S_{\mathscr C_{n+1}}$, then there exists $(\Delta''_1,\dots,\Delta''_n)\succeq(\Delta_1,\dots,\Delta_n),(\Delta'_1,\dots,\Delta'_n)$ in $S_{\mathscr{C}_{n}}$ and $\Delta''_{n+1}\succeq\Delta_{n+1},\Delta'_{n+1}\in\mathcal{T}_{M_{n+1}}$.~Now, we can just focus on $\Delta''_n$ and $\Delta''_{n+1}$.~Their boundaries are: $\del\Delta''_n=\delta_{n-1}''\sqcup\delta$ with $\delta_{n-1}\in\mathcal{T}_{\Sigma^*_{n-1}\sqcup\sigma_n}$ and $\delta\in\mathcal{T}_{\Sigma_{n}}$, and $\del\Delta''_{n+1}=(\delta')^*\sqcup\delta_{n}''$ with $\delta_n''\in\mathcal{T}_{\sigma_{n+1}}$ and $\delta'\in\mathcal{T}_{\Sigma_{n}}$.~Since $\mathcal{T}_{\Sigma_n}$ is a directed set, there exists $\delta''\succeq\delta,\delta'$, and so, by Lemma~\ref{lem:guess}, there exist $\Delta'''_n\succeq\Delta''_n\in\mathcal{T}_{M_{n}}$ and $\Delta'''_{n+1}\succeq\Delta''_{n+1}\in\mathcal{T}_{M_{n+1}}$, with $\del\Delta'''_n=\delta_{n-1}''\sqcup\delta''$ and $\del\Delta'''_{n+1}=(\delta'')^*\sqcup\delta_{n}''$.~Hence, $(\Delta''_1,\dots,\Delta''_{n-1},\Delta'''_{n},\Delta'''_{n+1})$ is the majorant of $(\Delta_1,\dots,\Delta_n,\Delta_{n+1})$ and $(\Delta'_1,\dots,\Delta'_n,\Delta'_{n+1})$ in $S_{\mathscr{C}_{n+1}}$.
\smallskip

\noindent
Let us now prove the statements \emph{1} and \emph{2}.
\begin{itemize}
    \item[\emph{1.}] Let ${\rm pr}_{n+1}:S_{\mathscr{C}_{n+1}}\to \mathcal{T}_{M_{n+1}}$ be the projector on the $(n+1)$-th factor, and let $\overline{\rm pr}:S_{\mathscr{C}_{n+1}}\to S_{\mathscr{C}_n}$ be its complementary projection (the projection is seen as the restriction on $S_{\mathscr{C}_{n+1}}$ of the one defined on $\prod_{l=1}^{n+1}\mathcal{T}_{M_l}$).~The projections are order-preserving and cofinal maps.~Let then $\hat{M}=M_1\cup_{\Sigma_1}M_2\cup_{\Sigma_2}\dots\cup_{\Sigma_{n-1}}M_{n}$, we can easily convince ourself that
    \begin{equation}
        \gamma_{\mathscr{C}_{n+1}}=\gamma_{\hat{M},M_{n+1},\Sigma_n}\circ((\gamma_{\mathscr{C}_{n}}\circ \overline{\rm pr})\times \mathrm{pr}_{n+1})\,,
    \end{equation}
    where $\gamma_{\hat{M},M_{n+1},\Sigma_n}$ is the map defined in Lemma~\ref{lem:cutsubnet}.~Now, $(\gamma_{\mathscr{C}_{n}}\circ \overline{\rm pr})\times \mathrm{pr}_{n+1}:S_{\mathscr{C}_{n+1}}\to \mathcal{T}_{\hat M}\times \mathcal{T}_{M_{n+1}}$ is clearly order-preserving, and is cofinal because, by inductive hypothesis, given any $(\Delta,\Delta_{n+1})$, we can find an element $(\Delta'_1,\dots,\Delta'_n,\Delta'_{n+1})$ in $S_{\mathscr{C}_{n+1}}$ such that $\gamma_{\mathscr{C}_{n}}\circ \overline{\rm pr}(\Delta'_1,\dots,\Delta'_n,\Delta'_{n+1})\succeq \Delta$. Now, we construct a majorant of the triangulation on the common boundary $\Sigma_n$ of $\Delta'_{n},\Delta_{n+1},\Delta'_{n+1}$.~Let $\Delta''_{n}$ be the triangulation on $M_n$ with that boundary triangulation, and $\Delta''_{n+1}$ the one for $M_{n+1}$.~By Lemma~\ref{lem:guess}, $\Delta''_{n}\succeq\Delta'_{n}$ and $\Delta''_{n+1}\succeq\Delta'_{n+1}\Delta_{n+1}$, so $(\gamma_{\mathscr{C}_{n}}\circ \overline{\rm pr})\times \mathrm{pr}_{n+1}(\Delta'_1,\dots,\Delta''_n,\Delta''_{n+1})\succeq(\Delta,\Delta_{n+1})$.~Thus, $\gamma_{\mathscr{C}_{n+1}}$ is a composition of order-preserving and cofinal maps, hence it is order-preserving and cofinal.
    \item[\emph{2.}] The inclusion is order-preserving because the order on $S_{\mathscr{C}_{n+1}}$ is induced by the product order of the full space $\prod_{l=1}^{n+1}\mathcal{T}_{M_l}$.~It is cofinal because of Lemma~\ref{lem:guess}.~Let us denote the inclusion maps by $\iota_{\mathscr{C}_{n+1}}:S_{\mathscr{C}_{n+1}}\to\prod_{l=1}^{n+1}\mathcal{T}_{M_l}$ and $\iota_{\mathscr{C}_{n}}:S_{\mathscr{C}_{n}}\to\prod_{l=1}^{n}\mathcal{T}_{M_l}$.~Clearly, we have
    \begin{equation}
        \iota_{\mathscr{C}_{n+1}}=(\iota_{\mathscr{C}_{n}}\circ \overline{\rm pr})\times\mathrm{pr}_{n+1},
    \end{equation}
    and the proof is similar to that of the previous statement.~By the inductive hypothesis, given any $(\Delta_1,\dots,\Delta_n\Delta_{n+1})$, we can find an element $(\Delta'_1,\dots,\Delta'_n,\Delta'_{n+1})$ in $S_{\mathscr{C}_{n+1}}$ such that $\iota_{\mathscr{C}_{n}}\circ \overline{\rm pr}(\Delta'_1,\dots\allowbreak\dots,\Delta'_n,\Delta'_{n+1})\succeq (\Delta_1,\dots,\Delta_n)$.~Now, we construct a majorant of the triangulation on the common boundary $\Sigma_n$ of $\Delta'_{n},\Delta_{n+1},\Delta'_{n+1}$, and we call $\Delta''_{n}$ the triangulation on $M_n$ with that boundary triangulation, and $\Delta''_{n+1}$ the one for $M_{n+1}$.~By Lemma~\ref{lem:guess}, $\Delta''_{n}\succeq\Delta'_{n}$ and $\Delta''_{n+1}\succeq\Delta'_{n+1}\Delta_{n+1}$, so $(\iota_{\mathscr{C}_{n}}\circ \overline{\rm pr})\times \mathrm{pr}_{n+1}(\Delta'_1,\dots,\Delta''_n,\Delta''_{n+1})\succeq(\Delta,\Delta_{n+1})$.
\end{itemize}
\end{proof}
With the above Lemmas~\ref{lem:gen_contProd} and \ref{lem:gen_cutsubnet} at our disposal, we can now proceed to discuss the property \eqref{eq:gluingconv} for the gluing of manifolds along similar lines to our previous discussion of the convolution property \eqref{eq:Pconvolution} of $P_{\Sigma}$.~Let us start from the LHS of \eqref{eq:gluingconv}, for which we have
\begin{equation}
    \begin{split}
    Z(M\cup_\Sigma N)[R_{\Sigma_1}(\psi_1)\otimes\psi_2]&=Z(M\cup_{\Sigma}(\Sigma\times I)\cup_{\Sigma}N)[R_{\Sigma_1}(\psi_1)\otimes\psi_2]\\
    &=\lim_{\mathcal T_{M\cup_{\Sigma}(\Sigma\times I)\cup_{\Sigma}N}}\braket{\tilde{Z}_{\Delta}|R_{\Sigma_1}(\psi_1)\otimes\psi_2}_{\h^{\rm ind}_{\Sigma^*_1\sqcup\Sigma_2}}\\
    &=\lim_{S_{\mathscr{C}}}\braket{\tilde{Z}_{\Delta}|R_{\Sigma_1}(\psi_1)\otimes\psi_2}_{\h^{\rm ind}_{\Sigma_1^*\sqcup\Sigma_2}}.
    \end{split}
\end{equation}
Here, $\mathscr{C}=\{M,\Sigma\times I,N,\Sigma,\Sigma\}$ and, in the last line, we have invoked Lemma~\ref{lem:gen_cutsubnet}.~By the gluing property of $\tilde{Z}_{\bullet}$ (cfr.~property 3, below Eq.~\eqref{eq:ind-net} in Sec.~\ref{Sec:SubdivIndLim}), we have that $\tilde{Z}_{\Delta_M\cup_{\delta_1}\Delta_{\Sigma\times I}\cup_{\delta_2}\Delta_N}=\tilde{Z}_{\Delta_N}\circ\tilde{Z}_{\Delta_{\Sigma\times I}}\circ \tilde{Z}_{\Delta_M}$ as an operator from $\h^{\rm ind}_{\Sigma_1}$ to $\h^{\rm ind}_{\Sigma_2}$.~Therefore,
\begin{equation}
    Z(M\cup_\Sigma N)[R_{\Sigma_1}(\psi_1)\otimes\psi_2]=\lim_{S_{\mathscr{C}}}\braket{\tilde{Z}_{\Delta_{\Sigma\times I}}\circ \tilde{Z}_{\Delta_M}(\psi_1)|\tilde{Z}^{\dagger}_{\Delta_N}\psi_2}_{\Hind}\;,
\end{equation}
from which, inserting a resolution of the identity on the two copies of $\Hind$, we obtain
{\allowdisplaybreaks
\begin{align}
        Z(M\cup_\Sigma N)[R_{\Sigma_1}(\psi_1)\otimes\psi_2]&=\lim_{S_{\mathscr{C}}}\sum_{\beta\in\mathscr{I}}\braket{\tilde{Z}_{\Delta_{\Sigma\times I}}\circ \tilde{Z}_{\Delta_M}(\psi_1)|\varphi_\beta}_{\Hind}\braket{\varphi_\beta|\tilde{Z}^{\dagger}_{\Delta_N}\psi_2}_{\Hind}\nonumber\\
        &= \lim_{S_{\mathscr{C}}}\sum_{\alpha,\beta\in\mathscr{I}}\braket{\tilde{Z}_{\Delta_M}(\psi_1)|\varphi_\alpha}_{\Hind}\braket{\varphi_\alpha|\tilde{Z}_{\Delta_{\Sigma\times I}}^{\dagger}\varphi_\beta}_{\Hind}\braket{\varphi_\beta|\tilde{Z}^{\dagger}_{\Delta_N}\psi_2}_{\Hind}\nonumber\\
        &=\lim_{S_{\mathscr{C}}}\sum_{\alpha,\beta\in\mathscr{I}}\braket{\tilde{Z}_{\Delta_M}|R_{\Sigma_1}(\psi_1)\otimes\varphi_\alpha}_{\h^{\rm ind}_{\Sigma^*_1\sqcup\Sigma}}\braket{\tilde{Z}_{\Delta_{\Sigma\times I}}\varphi_\alpha|\varphi_\beta}_{\Hind}\\
        &\qquad\qquad\quad\times\braket{\tilde{Z}_{\Delta_N}|R_{\Sigma}(\varphi_\beta)\otimes\psi_2}_{\h^{\rm ind}_{\Sigma^*\sqcup\Sigma_2}}\,.\nonumber
\end{align}
}
Let us now look at the RHS of \eqref{eq:gluingconv}.~By the definitions \eqref{dLim} and \eqref{def:PSigma} of $Z(M)$ and $P_{\Sigma}$, we have
\begin{equation}
\begin{split}
    \sum_{\alpha,\beta\in\mathscr I}Z(M)[\psi_1\otimes \varphi_\alpha]P_{\Sigma}(\varphi_\alpha)[\varphi_\beta]Z(N)[R_\Sigma(\varphi_\beta)\otimes\psi_2]&=\sum_{\alpha,\beta\in\mathscr I}\lim_{\Delta_M\in\mathcal{T}_M}\braket{\tilde{Z}_{\Delta_M}|R_{\Sigma_1}(\psi_1)\otimes \varphi_\alpha}_{\h^{\rm ind}_{\Sigma^*_1\sqcup\Sigma}}\\
    &\quad\;\times\lim_{\Delta_{\Sigma\times I}\in\mathcal{T}_{\Sigma\times I}}\braket{\tilde{Z}_{\Delta_{\Sigma\times I}}|R_{\Sigma}(\varphi_{\alpha})\otimes\varphi_{\beta}}_{\h^{\rm ind}_{\Sigma^*\sqcup\Sigma}}\\
    &\quad\;\times\lim_{\Delta_N\in\mathcal{T}_N}\braket{\tilde{Z}_{\Delta_N}|R_\Sigma(\varphi_\beta)\otimes\psi_2}_{\h^{\rm ind}_{\Sigma^*\sqcup\Sigma_2}}.
\end{split}
\end{equation}
Using the Lemmas~\ref{lem:gen_contProd} and \ref{lem:gen_cutsubnet} above, we can collect the tree nets as one defined on $\mathcal{T}_M\times \mathcal{T}_{\Sigma\times I}\times \mathcal{T}_{N}$, and then pass the limit to the subnet $S_\mathscr{C}$, thus obtaining 
\begin{equation}
\begin{split}
    \sum_{\alpha,\beta\in\mathscr I}Z(M)[\psi_1\otimes \varphi_\alpha]P_{\Sigma}(\varphi_\alpha)[\varphi_\beta]Z(N)[R_\Sigma(\varphi_\beta)\otimes\psi_2]&=\sum_{\alpha,\beta\in\mathscr I}\lim_{S_\mathscr{C}}\braket{\tilde{Z}_{\Delta_M}|R_{\Sigma_1}(\psi_1)\otimes \varphi_\alpha}_{\h^{\rm ind}_{\Sigma^*_1\sqcup\Sigma}}\\
    &\quad\quad\times\braket{\tilde{Z}_{\Delta_{\Sigma\times I}}|R_{\Sigma}(\varphi_{\alpha})\otimes\varphi_{\beta}}_{\h^{\rm ind}_{\Sigma^*\sqcup\Sigma}}\\
    &\quad\quad\times\braket{\tilde{Z}_{\Delta_N}|R_\Sigma(\varphi_\beta)\otimes\psi_2}_{\h^{\rm ind}_{\Sigma^*\sqcup\Sigma_2}}.
\end{split}
\end{equation}
Provided that the limit and the sum can be exchanged, the same expression of the LHS is then recovered.

\addcontentsline{toc}{section}{References}
\bibliographystyle{unsrt}
\bibliography{biblio}

\end{document}